\documentclass{article}

\usepackage[accepted]{aistats2024}

\usepackage[utf8]{inputenc} %
\usepackage[T1]{fontenc}    %
\usepackage{hyperref}       %
\usepackage{url}            %
\usepackage{booktabs}       %
\usepackage{amsfonts}       %
\usepackage{nicefrac}       %
\usepackage{microtype}      %
\usepackage{xcolor}         %
\usepackage{amsmath,amsthm,amssymb}
\usepackage{empheq}
\usepackage{subcaption}
\usepackage{bm,bbm,xfrac,url,enumerate,stackrel}

\usepackage[inline]{enumitem}
\usepackage{float}

\makeatletter
\newtheorem*{rep@theorem}{\rep@title}
\newcommand{\newreptheorem}[2]{%
\newenvironment{rep#1}[1]{%
 \def\rep@title{#2 \ref{##1}}%
 \begin{rep@theorem}}%
 {\end{rep@theorem}}}
\makeatother
\newreptheorem{theorem}{Theorem}
\newtheorem*{rep@proposition}{\rep@title}
\newcommand{\newrepproposition}[2]{%
\newenvironment{rep#1}[1]{%
 \def\rep@title{#2 \ref{##1}}%
 \begin{rep@proposition}}%
 {\end{rep@proposition}}}
\makeatother
\newreptheorem{proposition}{Proposition}
\newtheorem*{rep@lemma}{\rep@title}
\newcommand{\newreplemma}[2]{%
\newenvironment{rep#1}[1]{%
 \def\rep@title{#2 \ref{##1}}%
 \begin{rep@lemma}}%
 {\end{rep@lemma}}}
\makeatother
\newreptheorem{lemma}{Lemma}

\theoremstyle{plain}
\newtheorem{proposition}{Proposition}
\newtheorem{lemma}{Lemma}
\newtheorem{theorem}{Theorem}

\newtheorem*{example*}{Example}

\newtheorem*{remark*}{Remark}
\theoremstyle{definition}
\newtheorem{definition}{Definition}

\usepackage{todonotes}

\newcommand{\iprod}[2]{\langle #1, #2 \rangle}
\newcommand{\mathsep}{,~}
\newcommand{\st}{\,\middle|\,}
\newcommand{\set}[1]{\left\lbrace #1 \right\rbrace}
\newcommand{\card}[1]{\left\lvert{#1}\right\rvert}
\newcommand{\absv}[1]{\card{#1}}
\newcommand{\norm}[2]{\left\lVert{#1}\right\rVert_{#2}}

\newcommand{\setR}{\mathbb R}
\newcommand{\Rbar}{\tilde{\mathbb R}}
\newcommand{\setN}{\mathbb N}

\newcommand{\pb}[1]{\mathbb P\left[#1\right]}

\newcommand{\expectVariable}[2]{\mathbb E_{#1}\left[#2\right]}

\DeclareMathOperator*{\argmin}{arg\,min}

\DeclareMathOperator*{\sign}{sign}

\DeclareMathOperator*{\Laplace}{\textsc{Lap}}

\newcommand{\median}{\textsc{Med}}
\newcommand{\qrmedian}[1]{\textsc{QrMed}_{#1}}

\newcommand{\mean}{\textsc{Mean}}

\newcommand{\stddev}{\textsc{SD}}

\newcommand{\vote}{\textsc{Vote}}

\newcommand{\brmean}[1]{\textsc{LrMean}_{#1}}
\newcommand{\clippedmean}{\textsc{ClMean}}
\newcommand{\clip}{\textsc{Clip}}

\newcommand{\mrdistance}{\textsc{MrDist}}

\newcommand{\projection}[1]{\pi{\left(#1\right)}}

\newcommand{\averagevote}{\textsc{Mean}}

\newcommand{\mehestan}{\textsc{Mehestan}}

\newcommand{\DPMehestan}{\textsc{DP-Mehestan}}

\newcommand{\alternative}{a}
\newcommand{\alternativebis}{b}
\newcommand{\alternativeter}{c}

\newcommand{\ALTERNATIVE}{A}
\newcommand{\ALTERNATIVEPAIR}{C}

\newcommand{\voter}{n}
\newcommand{\voterbis}{m}
\newcommand{\byzantine}{f}
\newcommand{\BYZANTINE}{F}

\newcommand{\VOTER}{N}

\newcommand{\votingrights}[1]{w_{#1}}
\newcommand{\votingrightsfamily}{\bm w}
\newcommand{\votingrightsbis}{v}
\newcommand{\votingrightsfamilybis}{\bm v}

\newcommand{\distribution}{\mathcal{D}}

\newcommand{\voterscore}[1]{\theta_{#1}}
\newcommand{\voterscorefamily}[1]{\bm{\theta}_{#1}}
\newcommand{\normalizedscore}[1]{\tilde \theta_{#1}}

\newcommand{\score}[1]{x_{#1}}
\newcommand{\scorefamily}{\bm{x}}

\newcommand{\globalscore}[1]{\rho_{#1}}

\newcommand{\votingresilience}{{\color{red} W}}
\newcommand{\lipschitz}{L}

\newcommand{\faulty}{f}
\newcommand{\uncertainty}[1]{\delta_{#1}}

\newcommand{\scaling}[1]{s_{#1}}

\newcommand{\translation}[1]{\tau_{#1}}
\newcommand{\scalingbound}{\bar s}

\newcommand{\groupnormalizer}[1]{\overrightarrow{\textsc{Norm}}_{#1}}
\newcommand{\normalizer}[1]{\textsc{Norm}_{#1}}

\newcommand{\standardnormalizer}[1]{\textsc{StdNorm}_{#1}}
\newcommand{\minmaxnormalizer}[1]{\textsc{MinMaxNorm}_{#1}}
\newcommand{\correlation}{\textsc{Correl}}

\newcommand{\Loss}{\mathcal L}
\newcommand{\permutation}{\sigma}
\newcommand{\PERMUTATION}{S}
\newcommand{\aggregation}{\textsc{Agg}}
\newcommand{\PAIRWISECOMPARABLEVOTER}[1]{\VOTER_{#1}^{\ALTERNATIVEPAIR}}
\newcommand{\COMPARABLEVOTER}[1]{\VOTER_{#1}^{\ALTERNATIVE}}

\newcommand{\polarization}[1]{\eta_{#1}}
\newcommand{\diameter}{\Delta}

\newcommand{\avgtruescaling}{\bar{\scaling{}^*}}
\newcommand{\avgtruetranslation}{\tilde{\tau{}}^*}

\usepackage{bm}
\usepackage{algorithm}
\usepackage{algorithmicx}

\begin{document}

\twocolumn[

\aistatstitle{Robust Sparse Voting}

\aistatsauthor{ Youssef Allouah \And Rachid Guerraoui \And  Lê-Nguyên Hoang \And Oscar Villemaud }

\aistatsaddress{ EPFL \And  EPFL \And Calicarpa, Tournesol Association \And EPFL } ]

\begin{abstract}
\vspace{-0.25cm}
Many applications, such as content moderation and recommendation, require  reviewing and scoring a large number of alternatives. Doing so robustly is however very challenging. Indeed, voters' inputs are inevitably \emph{sparse}: most alternatives are only scored by a small fraction of voters. This sparsity amplifies the effects of biased voters introducing \emph{unfairness}, and of malicious voters seeking to hack the voting process by reporting \emph{dishonest} scores. 

We give a precise definition of the problem of \emph{robust sparse voting}, highlight its underlying technical challenges, and present a novel voting mechanism addressing the problem. We prove that, using this mechanism, no voter can have more than a small parameterizable effect on each alternative's score; a property we call \emph{Lipschitz resilience}. We also identify conditions of voters comparability under which any unanimous preferences can be recovered, even when each voter provides sparse scores, on a scale that is potentially very different from any other voter's score scale. Proving these properties required us to introduce, analyze and carefully compose novel aggregation primitives which could be of independent interest.
\end{abstract}

\vspace{-0.25cm}
\section{INTRODUCTION}
\label{sec:introduction}
\vspace{-0.1cm}
Voting has proven to be an effective way to reach collective decisions despite irreconcilable preferences.
However, traditional voting schemes have  been designed to handle a tractable set of alternatives.
Mechanisms like the \emph{majority judgment}~\cite{BalinskiLaraki11}, \emph{Borda}'s count~\cite{emerson2013original}, \emph{Kemeny-Young}'s scheme~\cite{kemeny1962preference}, \emph{randomized Condorcet}~\cite{Hoang17} and \emph{Schulze method} \cite{schulze2003new}, among others, typically require voters to provide ballots whose size is linear in the number of alternatives, and whose computation time is polynomial.
Such approaches are inapplicable when
the number of alternatives is very large, 
e.g. when electing the best movie of the year, the best paper of a conference, or the best text of law to implement.
In such a context, voting is inherently \emph{sparse}: as voters can only judge a small fraction of all alternatives.
Sparsity amplifies two major issues: heterogeneity in \emph{expression styles} 
and vulnerability to \emph{malicious voters}.

\vspace{-0.1cm}
\paragraph{Expression styles.}
On the one hand, different reviewers may adopt very distinct expression styles~\cite{DBLP:conf/atal/WangS19}.
For example, in the case of scientific peer review,
junior reviewers might use only modest judgments, e.g. \emph{weak} accept/reject, 
while other reviewers may frequently use definitive judgments, e.g. \emph{strong} accept/reject.
Meanwhile, some may be systematically positive and rarely suggest rejection,
while others may be consistently harsh and almost always recommend rejection.
Thus, the resulting acceptance decision of a paper may depend more on the expression styles of the assigned reviewers, 
 than on the actual quality of the paper.
In addition, this phenomenon is exacerbated by the fact that the assignment of papers to reviewers is rarely uniformly random.
In practice, a paper is more likely to be reviewed by someone whose expertise is close to the paper's focus.
Some reviewers may also prefer reviewing top quality papers only,
while other reviewers may focus on papers that are easy to reject.
Thus, sparsity may not be random: it may be \emph{adversarial} for some papers
and raises a risk of \emph{systematic unfairness}.~%
\emph{Robust sparse voting} requires protections against diverging expression styles.

\vspace{-0.1cm}
\paragraph{Malicious voters.}
On the other hand, especially when the number of alternatives far exceeds 
what honest voters can collaboratively score, 
as for social media content,
we must expect the existence of alternatives that no honest voter has scored.
This makes classical solutions for the robust statistics toolbox, like the median, ill-suited
to protect the security of the vote,
as such alternatives with no honest voter's assessment will be arbitrarily manipulable by a single malicious voter.
\emph{Robust sparse voting} also requires protection against malicious voters targeting such alternatives.

\subsection{Contributions}
\begin{itemize}
\item We propose a formalization of the \emph{robust sparse voting} problem, by identifying two formal properties tackling the aforementioned issues of voting biases and malicious voters.
The first property, which we call \emph{sparse unanimity}, stipulates that any unanimous preference can be recovered when sufficiently many voters participate, despite diverging expression styles and sparse score reporting.
The second property, which we call $\lipschitz$-\emph{Lipschitz resilience}, is a security property guaranteeing that any contributor can have at most a parameterizable impact $\lipschitz$ on the scoring of alternatives.

\item We introduce and analyze two novel Lipschitz-resilient aggregation primitives: \emph{Quadratically Regularized Median} is a generalization of the median, and \emph{Lipschitz-Robustified Mean} successfully outputs the mean under the right conditions.
We believe our new aggregation primitives to be of independent interest.

\item 
We propose a new voting algorithm, which we call $\mehestan{}$\footnote{Mehestan
is the name of one of the earliest proto-parliaments in Asia.},
and formally prove that it solves the \emph{robust sparse voting} problem.
Underlying $\mehestan{}$ lie our novel aggregation primitives mentioned above, which we carefully compose to conduct two crucial operations:
transform the different scores to bring them to a common scale, thus diluting voting biases, and perform robust aggregation on the transformed scores.

\item We empirically compare \mehestan{} to natural baselines for robust sparse voting,
under various adversarial settings.
In particular, we evaluate \mehestan{} under 
high sparsity (voters score only a few alternatives),
biased sparsity (voters only score a biased subset of alternatives), 
and in the presence of malicious voters sending random scores.

\end{itemize}

\subsection{Applications}
Large-scale algorithms routinely address a massive number of ethical dilemmas.
Namely, whenever a user searches ``climate change'' or ``vaccines'' on YouTube, Facebook or Amazon,
the algorithmic answers have potential life-or-death consequences 
transcending geographical boundaries.
These algorithms, heavily reliant on (implicit) voting mechanisms using upvotes or star ratings, struggle with ethical decision-making due to the inherent sparsity of evaluations. Indeed, most alternatives receive scrutiny from only a small fraction of users, if any at all, creating a complex challenge in addressing these ethical quandaries.

In the context of online \emph{content recommendation}, resilience to arbitrary behavior has become critical.
Social media have become information battlegrounds~\cite{atallah2019,bradshaw2021},
and their recommendation algorithms have been weaponized by all sorts of private and public actors~\cite{huawei_nytimes,yue2019},
many of which leverage troll farms to fabricate misleading online activities~\cite{bradshaw19,neudert2019,woolley2020},
or even simply exploit the vulnerabilities of the social media's advertisement systems~\cite{EdelsonLM20}.
The sheer scale of disinformation campaigns is staggering, as exemplified by Facebook's removal of \emph{15 billion} fake accounts in just two years~\cite{facebook_15_billion}. Unfortunately, the algorithms that underpin content moderation, recommendations, and ad-targeting remain opaque, providing fertile ground for an extensive industry of fake accounts that continuously manipulate online content~\cite{moore2023fake}. This ongoing manipulation underscores the critical need for resilient, transparent, and accountable algorithms to ensure the integrity and trustworthiness of online information.

It is also worth mentioning low-stake applications such as \emph{online surveys}.
There, robust sparse voting algorithms can handle incomplete and potentially biased responses in online surveys and polls, providing accurate representations of public opinions even when faced with malicious attempts to influence the results or self-selection biases~\cite{bethlehem2010selection,schaurer2020investigating}.
Besides, in \emph{peer-to-peer platforms} like Airbnb and Uber, reputation is vital for trust and safety. Malicious users might attempt to damage the reputation of others or artificially boost their own~\cite{dellarocas2000immunizing,xiong2004peertrust}. A robust sparse voting algorithm can reduce the influence of such behavior while accommodating real-world feedback.

\subsection{Structure of the Paper}
Section~\ref{sec:model} proposes a formalization of the \emph{robust sparse voting} problem.
Section~\ref{sec:byzantine} introduces our new robust aggregation primitives.
Section~\ref{sec:mehestan} introduces \mehestan{}, and proves its \emph{resilience} and \emph{sparse unanimity} properties.
Section~\ref{sec:experiments} presents our empirical evaluation\footnote{Our code is available at \url{https://github.com/ysfalh/robust-voting}.} of \mehestan{} under adversarial settings.
Section~\ref{sec:rel_work} reviews related work
in social choice theory and robust statistics. 
Additional related work 
on recommender systems and robust voting 
is in Appendix~\ref{sec:add-rel-work}. Missing proofs are in appendices \ref{sec:app-median} to \ref{app:proof-main}. %
Appendix~\ref{sec:additional-exp} presents additional experiments on \mehestan{}.
Appendix~\ref{sec:naive_vote} exposes the difficulty of sparse voting by proving the impossibility of sparse unanimity for individually scaled preferences.
Appendix~\ref{sec:remarks} extends \mehestan{} to guarantee desirable properties such as differential privacy~\cite{dwork2014algorithmic}.

\section{ROBUST SPARSE VOTING}
\label{sec:model}

We consider a set $[\VOTER] = \set{1, \ldots, \VOTER}$ of voters, and a set $[\ALTERNATIVE] = \set{1, \ldots, \ALTERNATIVE}$ of alternatives to score.
Each voter $\voter \in [\VOTER]$ is 
asked to provide score $\voterscore{\voter \alternative} \in \setR$ for each alternative $\alternative \in [\ALTERNATIVE]$.
Moreover, we allow their vote to be \emph{sparse}:
voters may fail to score an alternative $\alternative$, in which case we denote $\voterscore{\voter \alternative} \triangleq \perp$.
Denote $\Rbar \triangleq \setR \cup \set{\perp}$.
Then each voter's input is a vector $\voterscore{\voter} \in \Rbar^\ALTERNATIVE$.
We denote $\voterscorefamily{} \triangleq (\voterscore{1}, \ldots, \voterscore{\VOTER})$ the tuple of voters' scores.
Note that assuming that input scores are bounded has no incidence on, nor is needed in, our theory.
For the sake of exposition, we assume that all voters are given a unit voting right.
The appendix details the more general case of continuous voting rights.

Following the classical Von Neumann-Morgenstern rationality framework~\cite{morgenstern1953theory}, we assume that each voter $\voter$'s cardinal preference $\voterscore{\voter}$ is defined up to a positive affine transformation.
That is, we  consider $\voterscore{\voter}, \voterscore{\voter}'$ to be \emph{equivalent}, and denote $\voterscore{\voter} \sim \voterscore{\voter}'$, if and only if 
the two vectors score the same subset $\ALTERNATIVE_\voter$ of alternatives, 
and there exist $\scaling{} > 0, \translation{} \in \setR$ such that $\voterscore{\voter \alternative} = \scaling{} \voterscore{\voter \alternative}' + \translation{}$ for every alternative $\alternative \in \ALTERNATIVE_\voter$.

Our goal is to aggregate the voters' partial scores for all alternatives.
In other words, we aim to construct a voting algorithm $\vote : \left( \Rbar^{\ALTERNATIVE} \right)^\VOTER \rightarrow \setR^\ALTERNATIVE$,
which maps tuple $\voterscorefamily{}$ to a score vector $\globalscore{} \in \setR^\ALTERNATIVE$.
We denote by $\vote_\alternative (\voterscorefamily{})$ the score given to alternative $a$ using $\vote$.

\subsection{The Lipschitz Resilience Property}

We now introduce Lipschitz resilience, our security property against malicious voters.
Classical solutions sometimes demand that the output of the algorithm be insensitive to malicious voters~\cite{lamport2019byzantine}.
But this is ill-suited to sparse voting where, on some alternatives, the majority of voters may be malicious. 
Instead, Lipschitz resilience demands that the maximal impact of any (malicious) voter be bounded.

To formalize this definition, let us define $\perp_\ALTERNATIVE \in \Rbar^\ALTERNATIVE$ the empty vector,
i.e. defined by $(\perp_{\ALTERNATIVE})_\alternative = \perp$ for all $\alternative \in [\ALTERNATIVE]$.

\begin{definition}[Lipschitz resilience]
\label{def:byz-resilience}
\vote{} guarantees {\em $\lipschitz$-Lipschitz resilience} if, 
for all inputs $\voterscorefamily{}$ and any voter $\voter \in [\VOTER]$,
discarding voter $\voter$'s inputs 
can affect each output of the vote by at most $\lipschitz$,
i.e.
\begin{align*}
    &\forall \voterscorefamily{}, \voter,  
    \forall \alternative \mathsep  %
    \absv{\vote_\alternative (\voterscorefamily{}) - \vote_\alternative(\voterscorefamily{-\voter}, \perp_\ALTERNATIVE)}
    \leq \lipschitz,
\end{align*}
where $\voterscorefamily{-\voter}$ is the tuple $\voterscorefamily{}$ deprived of voter $\voter$'s inputs.
\vote{} is simply said to be resilient if
\vote{} is $ \lipschitz$-Lipschitz resilient
for some $\lipschitz > 0$.
\end{definition}

\paragraph{Interpretation.}
Lipschitz resilience can be naturally interpreted as Lipschitz continuity,
if we consider the $\ell_0$-norm for the input tuple $\voterscorefamily{}$,
and the $\ell_\infty$-norm for the output vector $\vote(\voterscorefamily{})$.
In the appendix, we generalize this definition to continuous voting rights,
and show that the Lipschitz continuity interpretation still holds with the $\ell_1$-norm on the voting rights.
The variable $\lipschitz$ can be interpreted as a resilience measure:
$F$ malicious voters cannot deviate the final score of an alternative by more than $F \cdot \lipschitz$.

Note that this is a nontrivial condition to guarantee. 
In particular, algorithms based on identifying a subset of reference/anchor alternatives to scale all users' preferences will usually fail to provide Lipschitz resilience, 
as they often feature a discontinuity when the subset of reference alternatives or users is changed.

\paragraph{Use cases.}
Lipschitz resilience is particularly useful in three scenarios. 
First, it is critical for sparse voting where only a few honest voters score some alternatives.
Without Lipschitz resilience, malicious voters could manipulate the scores of such alternatives. 
This is especially harmful in applications where low scores lead to censorship, while high scores lead to celebration.
Second, Lipschitz resilience is important for \emph{stability} in systems where high volatility may be discreditable. 
Last, $\lipschitz$-Lipschitz resilience is desirable for privacy-sensitive applications, such as healthcare and finance. 
Our voting algorithm \mehestan{} guarantees $\varepsilon$-differential privacy~\cite{dwork2014algorithmic} when additionally injecting Laplacian noise proportional to $\lipschitz$. 
This adaptation is possible for any $\lipschitz$-Lipschitz resilient voting algorithm and is explained in Appendix~\ref{sec:differential_privacy}.

\subsection{The Sparse Unanimity Property}
\label{sec:unanimity}
Our second desirable property is \emph{sparse unanimity}, 
which guarantees that the voting algorithm recovers unanimous preferences despite sparsity.
More precisely, consider the situation where an algorithm is given $\VOTER$ (positive affine) transformations of the same ground-truth scores vector $\theta_* \in \setR^\ALTERNATIVE$, where each transformation has some hidden coordinates (i.e., is sparsified).
Sparse unanimity guarantees that, once enough voters participate,
the algorithm recovers $\theta_*$.

We first introduce notation: for any voter score $\voterscore{\voter} \in \Rbar^{\ALTERNATIVE}$, 
denote $\ALTERNATIVE_\voter$ the subset of alternatives $\alternative$ for which $\voterscore{\voter \alternative} \neq \perp$ has been reported.
For any subset $B \subset [\ALTERNATIVE]$, 
let $\voterscore{|B} \in \setR^{B}$ be the (partial) score vector obtained by selecting only the entries $\alternative \in B$ from the partial vector $\voterscore{}$.
Moreover, let 
$\VOTER_\alternative \triangleq \set{\voter \in \VOTER \st \alternative \in \ALTERNATIVE_\voter}$ be the set of voters who scored alternative $\alternative$.
Define 
$\ALTERNATIVEPAIR_{\voter \voterbis}(\voterscorefamily{}) \triangleq \set{(\alternative, \alternativebis) \in (\ALTERNATIVE_\voter \cap \ALTERNATIVE_\voterbis)^2 \st \alternative < \alternativebis, \voterscore{\voter \alternative} \neq \voterscore{\voter \alternativebis} ~\text{and}~ \voterscore{\voterbis \alternative} \neq \voterscore{\voterbis \alternativebis}}$
the set of couples of alternatives that both voters $\voter$ and $\voterbis$ scored, each providing distinct scores to the two alternatives.
We can now formalize \emph{sparse unanimity}
whose definition will be clarified right after.
\begin{definition}[Sparse unanimity]
\label{def:sparse-unanimity}
\vote{} is {\em sparsely unanimous} if,
for all $\voterscore{*} \in \setR^\ALTERNATIVE$, there exists $\VOTER_0\geq0$ 
such that,
whenever voters' scores are $\voterscore{*}${-\em unanimous}, {\em comparable} and $\VOTER_0${-\em scored},
\vote{} retrieves the unanimous preferences.
More precisely, the assumptions
\begin{align*}
\begin{array}{ll}
    \voterscore{*}{\text{-\em unanimity:}} 
        & \forall \voter \in [\VOTER] \mathsep
        \voterscore{\voter} \sim \voterscore{*|\ALTERNATIVE_\voter}, \\
    {\text{\em comparability:}} 
        & \forall \voter \neq \voterbis \in [\VOTER] \mathsep
        \ALTERNATIVEPAIR_{\voter \voterbis}(\voterscorefamily{}) \neq \varnothing, \\
    \VOTER_0{\text{-\em scored:}} 
        & \forall \alternative \in [\ALTERNATIVE] \mathsep 
        \card{\VOTER_\alternative} \geq \VOTER_0
\end{array}
\end{align*}
must imply $\vote (\voterscorefamily{}) \sim \voterscore{*}$.
\end{definition}

We start by clarifying the conditions of sparse unanimity.
First, \emph{$\voterscore{*}$-unanimity} requires the reported scores of each voter $\voter$
to be equivalent to unanimous preferences $\voterscore{*}$ (i.e., ground-truth scores), but only on the subset $\ALTERNATIVE_\voter$ of alternatives scored by voter $\voter$.
Second, \emph{comparability} requires, for every distinct voters $n,m \in [\VOTER]$, the existence of at least one pair of alternatives they both scored distinctly.
Intuitively, comparability limits the sparsity of the inputs.
Finally, voters' scores are \emph{$\VOTER_0$-scored} if every alternative was scored by at least $\VOTER_0$ voters.
One difficulty is that $\VOTER_0$ must only depend on $\voterscore{*}$,
and cannot be made to depend on which voters scored which alternatives.

Let us consider a numerical example to illustrate the sparse unanimity property:
\begin{example*}
\label{ex:sparse-unanimity}
Consider a situation with $\VOTER = 2K$ voters and $\ALTERNATIVE = 4$ alternatives.
The unanimous preferences are given by $\voterscore{*} = [-2,-1,1,2]$.
Odd voter $2k+1$ scores the first three alternatives $\ALTERNATIVE_1=\{1,2,3\}$ with $\voterscore{1} = [-3,-2,0,\perp]$,
and even voter $2k$ scores the last three alternatives $\ALTERNATIVE_2=\{2,3,4\}$ with $\voterscore{2} = [\perp, 0, 4, 6]$.
Voters' scores are $\voterscore{*}$-unanimous: we have $\voterscore{1} = \voterscore{*|\ALTERNATIVE_1} - 1$, and $\voterscore{2} = 2\, \voterscore{*|\ALTERNATIVE_2} + 2$.
Also, they verify comparability since both voters scored alternatives $2$ and $3$ differently.
Any sparse unanimity guarantee on $\vote{}$ implies the existence of a value $\VOTER_0(\voterscore{*})$,
for which we guarantee $\vote(\voterscorefamily{}) \sim \voterscore{*}$.
For $K \geq \VOTER_0(\voterscore{*}) / 2$, \vote{} would then be guaranteed to recover $\voterscore{*}$.
\end{example*}

\paragraph{Failure of naive solutions.}
Sparse unanimity is a minimally desirable \emph{correctness} property that any sparse voting algorithm ought to satisfy.
Yet, it is surprisingly nontrivial to guarantee
because of sparsity.
Indeed, if all voters score all alternatives, a natural solution is to aggregate voters' scores after any basic normalization of each voter's score (e.g., using min-max normalization).
Such solutions unfortunately fail for sparse voting.
More precisely, in Appendix~\ref{sec:naive_vote}, we prove an impossibility result:
(reasonable) aggregations of voters' scores with individual-based normalizations (i.e., the normalization does not depend on other voters) fail to be \emph{sparsely unanimous}.

\paragraph{Strengthening sparse unanimity.}
We leave open the question of strengthening sparse unanimity, 
and of constructing a Lipschitz resilient algorithm that satisfies this strengthened condition.
We believe this problem to be very challenging.
For instance, we conjecture that no Lipschitz resilient algorithm can guarantee what could be called \emph{sparse majority}: informally, if we require the recovery of any preference 
that is consensual among a \emph{majority} of sufficiently active voters,
then Lipschitz resilience cannot hold.

\subsection{Nontriviality}

We introduce a third property, which we call \emph{nontriviality}, 
requiring that the diameter of the output of the vote be at least $1$, 
under the assumptions of sparse unanimity.

\begin{definition}[Nontriviality]
\label{def:nontrivial}
\vote{} is {\em nontrivial} if, 
whenever voters' scores are $\voterscore{*}${-\em unanimous}, {\em comparable} and $\VOTER_0${-\em scored},
there exist $\alternative, \alternativebis \in [\ALTERNATIVE]$
such that $\vote_\alternative (\voterscorefamily{}) - \vote_\alternativebis (\voterscorefamily{}) \geq 1$.
\end{definition}

The nontriviality property ensures that the vote's result is placed on an informative scale. Although beneficial, this feature is not essential, and its removal does not simplify the challenge of robust sparse voting.

\section{ROBUST PRIMITIVES}
\label{sec:byzantine}

We now introduce robust
score aggregation functions used in $\mehestan{}$ to guarantee Lipschitz resilience.
For simplicity, in this section, we assume that there is only one alternative; i.e. $\ALTERNATIVE = 1$ and thus each $\voterscore{\voter
}$ is a \emph{scalar} (or non-reported value).

\subsection{Weighted Averaging and Median}

We first show that classical (robust) statistics operators, averaging and median, fail to be Lipschitz resilient.

\paragraph{Averaging.}
A widely used algorithm, e.g. for \emph{collaborative filtering} algorithms for group recommender systems~\cite{felfernig2018algorithms}, is the averaging of available data, i.e.
$\averagevote (\voterscorefamily{}) \triangleq \left( \sum_{\voter \in \VOTER^*} \voterscore{\voter} \right) / \card{\VOTER^*}$,
where $\VOTER^*$ is the set of voters $\voter$ who reported a score, i.e. $\voterscore{\voter} \neq \perp$.
It was proved to satisfy several desirable voting properties~\cite{pennock2000social}. 

\paragraph{Median.}
A popular robust mean estimator is the median, which we denote $\median{}$, and is the main ingredient of the popular \emph{majority judgement} \cite{BalinskiLaraki11} voting algorithm.
A median $M \triangleq \median{} (\votingrightsfamily{}, \voterscorefamily{})$ must divide voters with reported scores into two sets of equal sizes,
i.e. $\card{\set{\voter \in \VOTER^* : \voterscore{\voter} < M}} \leq \frac{1}{2} \card{\VOTER^*} \geq \card{\set{\voter \in \VOTER^* : \voterscore{\voter} > M}}$.

Proposition~\ref{prop:median-mean} below shows that weighted averaging and median fail to satisfy our resilience property.
\begin{proposition}
\label{prop:median-mean}
  Neither \mean{} nor \median{}
  is Lipschitz resilient.
\end{proposition}
The proof, which can be found in Appendix~\ref{sec:app-median}, instantiates a simple situation where the malicious voter can manipulate the outcome more than what is allowed by $\lipschitz$-Lipschitz resilience.

\subsection{Quadratically Regularized Median}

We now introduce the \emph{quadratically regularized median}, denoted $\qrmedian{}$.
$\qrmedian{\lipschitz}$ is parameterized by $\lipschitz > 0$,
and is defined as follows:
\begin{equation}
\label{eq:regmedian}
  \qrmedian{\lipschitz} (\voterscorefamily{}) \triangleq \argmin_{z \in \setR} ~ 
    \frac{1}{2 \lipschitz} z^2 + \sum_{\voter: \voterscore{\voter} \neq \perp} \absv{z - \voterscore{\voter}}.
\end{equation}
In practice, we approximate $\qrmedian{\lipschitz}$ by solving~\eqref{eq:regmedian}, 
which is an exponentially fast operation (in the approximation error) using gradient descent.
Note also that $\qrmedian{\lipschitz}$ corresponds to the median when $\lipschitz \rightarrow \infty$~\cite{minsker2015}.
Theorem~\ref{th:qrmed-br} guarantees that $\qrmedian{\lipschitz}$ is $\lipschitz$-Lipschitz resilient.
The full proof is deferred to Appendix~\ref{app:qrmed}.
\begin{theorem}
\label{th:qrmed-br}
$\qrmedian{\lipschitz}$ is well-defined and $\lipschitz$-Lipschitz resilient.

\end{theorem}

\vspace{-1em}
\begin{proof}[Sketch of proof]
  The objective~\eqref{eq:regmedian} minimized by $\qrmedian{\lipschitz}$ is $\tfrac{1}{\lipschitz}$-strongly convex, which implies that its minimizer is unique and that $\qrmedian{\lipschitz}$ is well-defined.
  $\lipschitz$-Lipschitz resilience follows from strong convexity, and the fact that the derivative of each summand $\absv{z-\voterscore{\voter}}$ is bounded by $1$.
\end{proof}

An additional property, which is used in the proof for \mehestan{}'s resilience, is the Lipschitz continuity of $\qrmedian{}$, as stated by Proposition~\ref{prop:qrmed-lipschitz} in Appendix~\ref{sec:app-other-proofs}.

\subsection{Lipschitz-Robustified Mean}

We now introduce \emph{Lipschitz-Robustified Mean}, which we denote $\brmean{}$, a primitive that returns the mean of any bounded inputs when sufficiently many voters participate, while satisfying $\lipschitz$-Lipschitz resilience.
It builds upon the \emph{clipped mean} $\clippedmean{}$ centered on $\mu$ and of radius $\diameter$ defined as:
\begin{align*}
    &\clippedmean{} (\voterscorefamily{} | \mu, \diameter) 
    \triangleq \mean (\clip(\voterscorefamily{} | \mu, \diameter)) \\
    &= \frac{1}{\card{\VOTER^*}} 
    \sum_{\voter \in \VOTER^*} \clip (\voterscore{\voter} | \mu, \diameter),
\end{align*}
where $\clip( x | \mu, \diameter) \triangleq \max \set{ \mu - \diameter, \min \set{ \mu + \diameter, x } }$ clips $x$ within the interval $[\mu - \diameter, \mu + \diameter]$,
and where $\VOTER^* \triangleq \set{\voter : \voterscore{\voter} \neq \emptyset}$.
$\brmean{}$ is then obtained by executing $\clippedmean{}$, centered on $\qrmedian{}$, with a radius that grows linearly with the number of users $\VOTER^*$:
\begin{align*}
    &\brmean{\lipschitz} (\voterscorefamily{} ) 
    \triangleq 
    \clippedmean{} \left( \voterscorefamily{} \st \qrmedian{\lipschitz/4} (\voterscorefamily{} ), \frac{\lipschitz \card{\VOTER^*}}{4} \right).
\end{align*}
As stated by Theorem~\ref{th:brmean} below, $\brmean{}$ verifies several properties.
The full proof is in Appendix~\ref{app:brmean}.
\begin{theorem}
\label{th:brmean}
$\brmean{\lipschitz}$ is $\lipschitz$-resilient.
Moreover, if there exists $\diameter>0$ such that $\card{\VOTER^*} \geq 8 \diameter / \lipschitz$ and $\voterscore{\voter} \in [-\diameter, \diameter]$ for all $\voter$,
then $\brmean{\lipschitz} (\voterscorefamily{}) = \mean{} (\voterscorefamily{})$.
\end{theorem}
\vspace{-.75em}
\begin{proof}
[Sketch of proof]
In the proof, we show that $\clippedmean{}$ is $1$-Lipschitz in the center $\mu$ and in the radius $\Delta$, 
and is also sufficiently resilient for small radii, when the number of voters is large enough.
The guarantee $\brmean{} = \mean{}$ then holds once enough voters participate, 
so that the radius of $\clippedmean{}$ could safely grow large enough to contain all voters' inputs.
\end{proof}

Remarkably, $\brmean{}$ eventually returns the mean of bounded inputs, provided that sufficiently many voters participate, 
despite being oblivious to the input bounds.
This is a critical property that will be at the heart of the sparse unanimity guarantee of \mehestan{}.
Designing a resilient aggregation with this feature turned out to be the most challenging aspect of our algorithm design.
An additional property, which we use in the proof for \mehestan{}'s resilience, is the Lipschitz continuity of $\brmean{}$ with respect to its inputs, as stated by Proposition~\ref{prop:brmean-lipschitz} in Appendix~\ref{sec:app-other-proofs}.

\section{OUR ALGORITHM: \mehestan{}}
\label{sec:mehestan}
In this section, we first introduce our algorithm \mehestan{}, and
conclude with theoretical guarantees.

\subsection{Description of $\mehestan{}$}
\label{subsec:mehestan-description}

$\mehestan{}$ proceeds in four principal steps:

\paragraph{1. Local normalization.}
First, \mehestan{} normalizes every score vector using min-max normalization, so that the minimal and maximal scores are respectively $0$ and $1$:
\begin{align}
    \label{eq:local-norm}
    \normalizedscore{\voter} \triangleq \frac{\voterscore{\voter} -  \min_{\alternative \in \ALTERNATIVE_{\voter}} \voterscore{\voter \alternative}}{\max_{\alternative \in \ALTERNATIVE_{\voter}} \voterscore{\voter \alternative}- \min_{\alternative \in \ALTERNATIVE_{\voter}} \voterscore{\voter \alternative}}.%
\end{align}
Note that min-max normalization is well-defined only if $\voterscore{\voter}$ has at least two distinct reported scores. 
If not, the scores are non-informative with respect to (multiplicative) scaling since they are equivalent to the vector of zeros, so we simply set $\normalizedscore{\voter}$ to be the vector of zeros, without any incidence on the theoretical guarantees.

\paragraph{2. Scaling factor search.}
For any voter $\voter \in [\VOTER]$, we define $\PAIRWISECOMPARABLEVOTER{\voter} \triangleq \set{\voterbis \in [\VOTER] \st \ALTERNATIVEPAIR_{\voter \voterbis} (\voterscorefamily{}) \neq \varnothing}$ the set of voters comparable to $\voter$.
In words, $\voterbis \in \PAIRWISECOMPARABLEVOTER{\voter}$ if and only if there exist two alternatives that $\voter$ and $\voterbis$ both scored differently.
For each voter $\voterbis \in \PAIRWISECOMPARABLEVOTER{\voter}$, 
we compute the comparative scaling $\scaling{\voter \voterbis}$ of voters $\voter$ and $\voterbis$ defined as 
\begin{equation}
    \label{eq:comp-scaling}
    \scaling{\voter \voterbis} 
    \triangleq \frac{1}{\card{\ALTERNATIVEPAIR_{\voter \voterbis}}}
    \sum_{(\alternative, \alternativebis) \in \ALTERNATIVEPAIR_{\voter \voterbis}} 
    \frac{|\normalizedscore{\voterbis \alternative} - \normalizedscore{\voterbis \alternativebis}|}{|\normalizedscore{\voter \alternative} - \normalizedscore{\voter \alternativebis}|}.
\end{equation}
From a high-level perspective, each comparative scaling $\scaling{\voter \voterbis}$ is an implicit vote by voter $\voterbis$ for the scaling factor of voter $\voter$.
In this step, \mehestan{} aggregates the comparative scalings $\scaling{\voter \voterbis}$ via $\brmean{}$:
\begin{equation}
    \label{eq:scaling}
    \scaling{\voter} \triangleq 
    1 + \brmean{\lipschitz/7} 
    \big(\set{
        \scaling{\voter\voterbis} - 1
        \st
        \voterbis \in \PAIRWISECOMPARABLEVOTER{\voter}
    }\big).
\end{equation}
Above, we do not directly take the $\brmean{}$ of the comparative scaling ratios $\scaling{\voter \voterbis}$ so that the default value equals $1$ in the absence of comparable voters.
Moreover, since each $\scaling{\voter \voterbis}$ corresponds to a vote for the scaling factor of $\scaling{\voter}$, the computation of $\scaling{\voter}$ can be interpreted as a search for a \emph{common scale}.

\paragraph{3. Translation factor search.}
While the previous step computes a common multiplicative scaling factor, the current step searches for a \emph{common translation}.
More precisely, define $\COMPARABLEVOTER{\voter} \triangleq \set{\voterbis \in [\VOTER] \st \ALTERNATIVE_{\voter \voterbis} \neq \emptyset}$ the set of translation-comparable voters. For each voter $\voterbis \in \COMPARABLEVOTER{\voter}$, we compute
\begin{equation}
    \label{eq:comp-translation}
    \translation{\voter \voterbis} \triangleq 
    \frac{1}{\card{\ALTERNATIVE_{\voter \voterbis}}} 
    \sum_{\alternative \in \ALTERNATIVE_{\voter \voterbis}}
    \left( \scaling{\voterbis} \normalizedscore{\voterbis \alternative} - \scaling{\voter} \normalizedscore{\voter \alternative} \right).
\end{equation}

In this step, \mehestan{} aggregates the comparative translation factors $\translation{\voter \voterbis}$ via $\brmean{}$:
\begin{equation}
\label{eq:translation}
  \translation{\voter} \triangleq
  \brmean{\lipschitz / 7} 
  \left(
  \set{\translation{\voter \voterbis} 
  \st
  \voterbis \in \COMPARABLEVOTER{\voter}}\right).
\end{equation}

\paragraph{4. Alternative-wise score aggregation.}
Finally, \mehestan{} linearly transforms the scores vectors with the obtained scaling and translation factors:
\begin{equation}
\label{eq:global-norm}
\hat{\theta}_{\voter} \coloneqq \scaling{\voter} \normalizedscore{\voter}+\translation{\voter}.
\end{equation}

We refer to this step as \emph{global} normalization, as opposed to \emph{local} normalization~\eqref{eq:local-norm}, because the transformation of each vector is dependent on all input scores.
Then, \mehestan{} aggregates the transformed scores along each alternative $\alternative \in [\ALTERNATIVE]$ via $\qrmedian{}$:
\begin{equation*}
\label{eq:mehestan}
    \mehestan_{\lipschitz \alternative}(\voterscorefamily{}) 
    = \qrmedian{\lipschitz/7} \left( \scaling{\voter} \normalizedscore{\voter \alternative}+\translation{\voter} \st \voter \in [\VOTER] \right).
\end{equation*}

The full procedure of \mehestan{} is summarized in Algorithm~\ref{alg:mehestan}.

\begin{algorithm}[t]
\textbf{Input}: 
Voters' scores $\voterscorefamily{}$, Lipschitz resilience $\lipschitz$\\
\textbf{Output}: The aggregate scores $\mehestan_{\lipschitz}(\votingrightsfamily{},\voterscorefamily{})$
\begin{algorithmic}[1]
	    \State
	    $\forall \voter,$ compute $\normalizedscore{\voter}$, the min-max normalization of $\voterscore{\voter}$
	\Comment{Local normalization}
	   \State $\forall \voter,\voterbis \in [\VOTER],$ compute $\scaling{\voter\voterbis}$ following Equation~\eqref{eq:comp-scaling}, if $n,m$ are comparable
	   \State
	   $\forall \voter \mathsep \scaling{\voter} \leftarrow 1 + \brmean{\lipschitz/7}
	   \hspace{-1pt}\left(\scaling{\voter\voterbis} - 1  \st \voterbis \in \PAIRWISECOMPARABLEVOTER{\voter} \right)$ 
	   \State $\forall \voter,\voterbis \in [\VOTER],$ compute $\translation{\voter\voterbis}$ following Equation~\eqref{eq:comp-translation}, if $n,m$ are comparable
	   \State
	   $\forall \voter \mathsep 
	   \translation{\voter} \leftarrow \brmean{\lipschitz/7} 
        \left( \votingrights{\voterbis}, \translation{\voter \voterbis} 
        \st \voterbis \in \COMPARABLEVOTER{\voter} \right)$
        \State
	    $\forall \voter \mathsep \hat{\theta}_{\voter} \leftarrow \scaling{\voter} \normalizedscore{\voter}+\translation{\voter}$
	    \Comment{Global normalization}
	    \State
	    $\forall \alternative \mathsep \globalscore{\alternative} \leftarrow \qrmedian{\lipschitz / 7} (\hat{\theta}_{\voter\alternative} | \voter \in [\VOTER])$
	    \State \textbf{return} $\globalscore{}$\;
	\caption{\mehestan{}}
	\label{alg:mehestan}
\end{algorithmic}
\end{algorithm}

\subsection{Theoretical Guarantees}
\label{subsec:theory}
Finally, we state the main result of our paper in Theorem~\ref{th:mehestan-br} below, showing that $\mehestan{}$ satisfies both properties of the robust sparse voting problem. 
We defer the full proof to Appendix~\ref{app:proof-main}.
\begin{theorem}
\label{th:mehestan-br}
$\mehestan{}_{\lipschitz}$ is $\lipschitz$-Lipschitz resilient, sparsely unanimous and nontrivial, with
\begin{equation}
    \label{eq:num_voters}
    \VOTER_0(\voterscore{*}) \triangleq 
    \frac{8}{\lipschitz}
     \left( \frac{\max_{\alternative, \alternativebis} \absv{\voterscore{* \alternative} - \voterscore{* \alternativebis}}}{\min_{\alternative, \alternativebis : \voterscore{* \alternative} \neq \voterscore{* \alternativebis}} \absv{\voterscore{* \alternative} - \voterscore{* \alternativebis}}} \right)^2.
\end{equation}
\end{theorem}
\vspace{-1em}
\begin{proof}[Sketch of proof]
Let us first address \emph{sparse unanimity}.
    Placing ourselves in the situation where voters' scores are $\voterscore{*}$-unanimous, we can write each input scores vector as
$\voterscore{\voter} = \scaling{\voter}^* \voterscore{*|\ALTERNATIVE_\voter} + \translation{\voter}^*$,
where $\scaling{\voter}^*>0$ and $\translation{\voter}^*$ are \emph{unknown} to the algorithm.
The main challenge of the proof is to show that the scores $\hat{\theta}_\voter$, obtained after global normalization~\eqref{eq:global-norm}, can be written in the form $\scaling{}^* \voterscore{*|\ALTERNATIVE_\voter}+\translation{}^*$, where $\scaling{}^*>0,\translation{}^* \in \setR$ are voter-independent.
The latter quantities are exactly what defines the ``common scale'' found by \mehestan{}, before applying an alternative-wise aggregation.
To do so, the proof proceeds as follows: once enough voters participate,
    the properties of $\brmean{}$, especially that it returns the mean in specific conditions (see Section~\ref{sec:byzantine}), allow showing that each comparative scaling factor (see Equation~\ref{eq:comp-scaling}) is in fact $\scaling{\voter \voterbis}= \scaling{\voterbis}^*/\scaling{\voter}^*$.
    We then show the scaling factor $\scaling{\voter}$ (see Equation~\ref{eq:scaling}) to equal  $\scaling{}^*/\scaling{\voter}^*$, where $\scaling{}^*$ is voter-independent.
    Similarly, we show the translation factor $\translation{\voter}$ (see Equation~\ref{eq:translation}) to equal $\translation{}^* - \scaling{}^* \translation{\voter}^*/\scaling{\voter}^*$, where $\translation{}^*$ is voter-independent.
    Overall, for every $\voter$, we obtain $\hat{\theta}_\voter=\scaling{\voter} \normalizedscore{\voter \alternative}+\translation{\voter} = \scaling{}^* \voterscore{*|\ALTERNATIVE_\voter}+\translation{}^*$.
    Then, the final aggregation performed with $\qrmedian{\lipschitz}$ returns $\scaling{}^* \voterscore{*}+\translation{}^*$, which satisfies sparse unanimity.

We now address \emph{Lipschitz resilience}.
Since we show $\brmean{\lipschitz}$ and $\qrmedian{\lipschitz}$ to be $\lipschitz$-Lipschitz resilient, all aggregation operations made in \mehestan{} are resilient, and the proof mainly composes the bounds guaranteed by $\lipschitz$-Lipschitz resilience.
    Note that we set the resilience parameter of $\brmean{}$ and $\qrmedian{}$ to $\lipschitz / 7$ so that the bounds of resilience can be correctly composed to yield the final $\lipschitz$-Lipschitz resilience.

Finally, \emph{nontriviality} from the fact that scaled individual scores are necessarily blown up,
compared to min-max normalization of the full score vector $\voterscore{*}$,
which already satisfies nontriviality.
\end{proof}

\paragraph{Trade-off discussion.}
The analysis of \mehestan{} underlying Theorem~\ref{th:mehestan-br} raises a tension between $\lipschitz$-Lipschitz resilience and sparse unanimity.
To see this, recall that verifying sparse unanimity requires enough voters to participate, which means that alternatives should be $\VOTER_0$-scored (see Definition~\ref{def:sparse-unanimity}).
In fact, as shown in Equation~\ref{eq:num_voters}, it is sufficient for $\VOTER_0$ to be proportional to $\tfrac{1}{\lipschitz}$ and a function of $\theta_*$ bounding the difference between voters’ scalings; i.e. to what extent voters express the same preferences $\theta_*$ differently.
Therefore, stronger Lipschitz resilience (lower $\lipschitz$) implies that more participation is needed (larger $\VOTER_0$) to recover unanimous preferences.
This trade-off raises an interesting research question: whether this trade-off is fundamental or algorithm-dependent.
Answering this question can lead to more efficient algorithms or tighter theoretical bounds for $\mehestan{}$.

\paragraph{Complexity.}
The collaborative scaling of voters' scores (steps 2 and 3 in Section~\ref{subsec:mehestan-description}) is the computational bottleneck of \mehestan{}.
In the worst case, for all pairs of voters, it requires going through all pairs of alternatives,
thereby yielding a $\mathcal O(\ALTERNATIVE^2 \VOTER^2)$ time complexity.
In practice, this heavy workload can be mitigated,
by performing collaborative scaling in an asynchronous rare manner,
and by assuming that the scaling factors do not vary much over time~\cite{beylerian2022tournesol}.

\section{EMPIRICAL EVALUATION}
\label{sec:experiments}
We report experiments\footnote{Our code is available at \url{https://github.com/ysfalh/robust-voting}.}
testing the performance of \mehestan{} under sparsity and malicious attacks.
 \begin{figure*}[t!]
 \vspace{-5pt}
\centering
\begin{subfigure}{.49\textwidth}
  \centering
  \includegraphics[width=\linewidth]{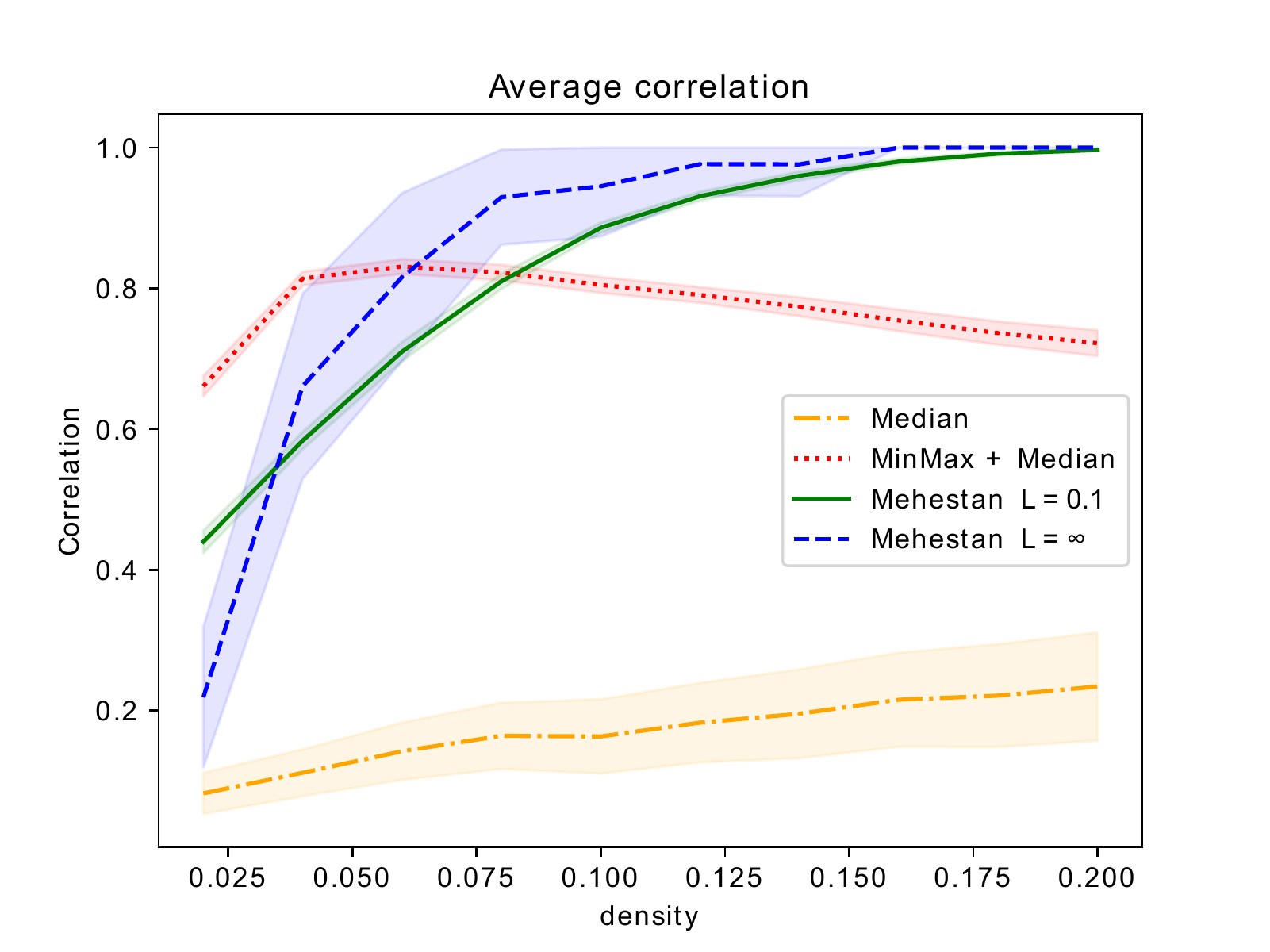}
  \caption{Performance under sparsity without malicious voters.}
  \label{fig:density}
\end{subfigure}\hfill
\begin{subfigure}{.49\textwidth}
  \centering
  \includegraphics[width=\linewidth]{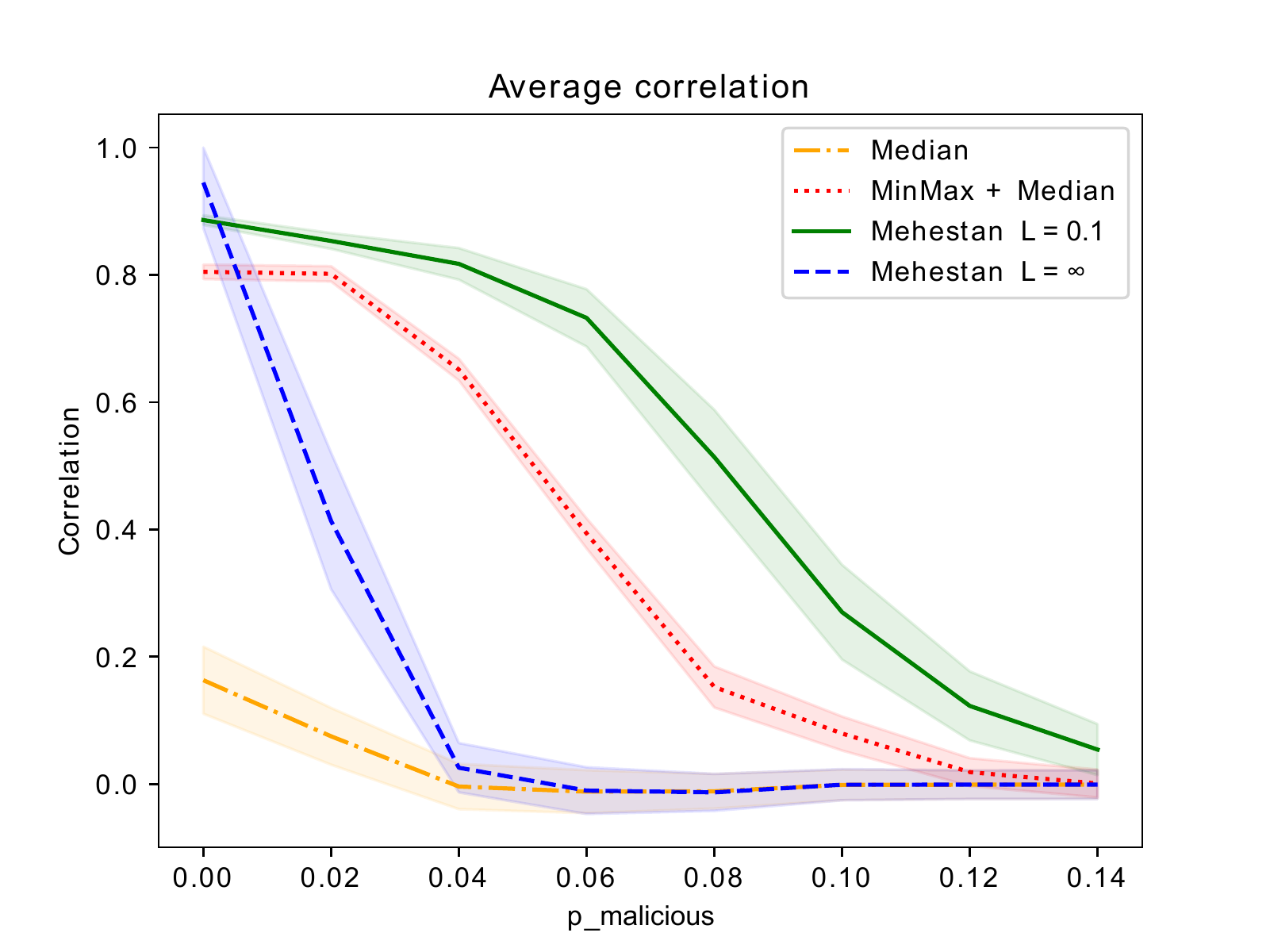}
  \caption{
  Performance with malicious voters ($\texttt{density}=0.1$). }
  
  \label{fig:sparsification}
\end{subfigure}
\vspace{-5pt}
\caption{
Performance of \mehestan{} under sparsity, with and without malicious voters.
}
\vspace{-12pt}
\label{fig:results}
\end{figure*}
\vspace{-5pt}
\subsection{Setting}

\paragraph{Data generation.} 

We generate synthetic data for $\VOTER = 150$ voters and $\ALTERNATIVE = 300$ alternatives.
We randomly draw a ground-truth score vector $\voterscore{*} \in \setR^\ALTERNATIVE$, by independently sampling the coordinates from the standard Gaussian distribution.
Each honest voter $\voter$ is assigned the score vector $\voterscore{\voter} = \scaling{\voter}^* \voterscore{*|\ALTERNATIVE_\voter} + \translation{\voter}^*$, where $\scaling{\voter}^* >0$ and  $\translation{\voter}^*$ are randomly drawn from the log-normal and normal distributions, respectively. 
To simulate sparsity, each alternative is scored by each voter following a Bernoulli trial with probability \texttt{density}. 
Additionally, to simulate biased sparsity, we remove the votes of half of the voters for the top 20\% alternatives (sorted according to $\theta_*$), and the votes of the other half of voters for the bottom 20\% alternatives.

\paragraph{Description of experiments.} Each experiment measures Pearson's correlation between the algorithms' output and the ground-truth preferences $\voterscore{*}$. 
Our experiments compare four voting algorithms: 
\emph{Median} (alternative-wise median), 
\emph{MinMax+Median} (alternative-wise median of min-max normalized scores), 
\emph{\mehestan{}} with resilience parameters $\lipschitz=0.1$ and $\lipschitz \rightarrow \infty$.
The first two algorithms serve as natural baselines (as discussed in Section~\ref{sec:unanimity}).
In the experiment reported in Figure~\ref{fig:sparsification}, we add malicious voters, whose votes are all the same, randomly drawn from the standard Gaussian. 
In Figure~\ref{fig:sparsification}, the x-axis is \texttt{p\_malicious}, which denotes the fraction of malicious voters. 
Note that this parameter takes reasonably large values in Figure~\ref{fig:sparsification}. 
Indeed, the malicious voters scores every alternative, as opposed to the honest voters who score $\approx 8\%$ of them (when $\texttt{density}=0.1$).
Each experiment is repeated $20$ times, with the seeds $1$ to $20$ for reproducibility. 
The average correlation values and the $95\%$ confidence intervals are plotted.
\subsection{Results}

Our experimental results are shown in Figure~\ref{fig:results}.
Observe on both plots that \textit{Median} fails to recover the ground-truth scores. 
This is expected, as (alternative-wise) \textit{Median} is not designed to tackle voting biases.
Below, we compare \mehestan{} with \textit{MinMax+Median}, which puts the focus on the global normalization procedure specific to \mehestan{} (see Algorithm~\ref{alg:mehestan}).
Next, we compare \mehestan{} with $\lipschitz \rightarrow \infty$ and $\lipschitz=0.1$, highlighting the impact of the resilience parameter $\lipschitz$.

\textbf{Impact of global normalization.} Figure~\ref{fig:density} shows that \mehestan{} performs well under sparsity,
and correctly recovers the unanimous preferences when the density is large enough (as guaranteed by sparse unanimity). 
 \textit{MinMax+Median} fails to do so, even for larger densities.
However, we observe that for the highest levels of sparsity (i.e. low density), \mehestan{} is less effective than \textit{MinMax+Median}. 
This can be explained by the fact that, at these densities, the search for scaling and translation factors may fail, because of the absence of comparable voters (see Section~\ref{sec:mehestan}).

\textbf{Impact of resilience parameter $\lipschitz$.} 
On the one hand, setting $\lipschitz=0.1$ only slightly hinders the performance of \mehestan{} in the absence of malicious voters~(Figure~\ref{fig:density}).
This is expected, as the non-malicious case does not require $\lipschitz$-Lipschitz resilience.
On the other hand, setting $\lipschitz=0.1$ enables better tolerance to malicious voters~(Figure~\ref{fig:sparsification}). 
Indeed, \mehestan{} with $\lipschitz \rightarrow \infty$ fares poorly in Figure~\ref{fig:sparsification}, as setting $\lipschitz \rightarrow \infty$ leaves the global normalization vulnerable to malicious manipulation.
This confirms the trade-off between resilience and sparse unanimity, as discussed in Section~\ref{subsec:theory}.

\section{RELATED WORK}
\label{sec:rel_work}
We cover closely related work below, and defer
additional related work to Appendix~\ref{sec:add-rel-work}.
\vspace{-0.3cm}
\paragraph{Social choice theory.}
The problem of \emph{sparse voting}, without malicious voters, has mainly been addressed in ordinal voting, 
i.e. voters provide incomplete \emph{rankings} or pairwise comparisons~\cite{negahban2012iterative,lu2013multi,moreno2016axiomatic,fotakis2021aggregating,imber2022approval}. 
There, voting seeks to retrieve a central ranking despite sparsity, which is very similar to sparse unanimity, 
except that the latter also recovers cardinal preferences,
which may be more suitable for some applications.
Additionally, intriguingly, \cite{DBLP:conf/atal/WangS19} proves that cardinal inputs allow constructing estimators that strictly outperform ordinal inputs,
even when considering arbitrary order-preserving input miscalibrations,
which suggests that leveraging cardinal inputs may be valuable.
The cardinal (robust) sparse voting setting has however been relatively understudied.
In this setting, the work of Meir et al.~\cite{meir2022sybil} is closest to ours, but we argue that our theoretical approach is more general.
Namely,~\cite{meir2022sybil}  
(i) does not tackle the general cardinal case, where voters provide scores for multiple alternatives;
(ii) assumes honest voters to be either passive (do not vote at all), or active (vote for all alternatives), while we allow honest voters to vote for a (strict) subset of alternatives; and
(iii) assumes the existence of a distinguished alternative, called ``reality'', which serves as an anchor to their safety and liveness properties.

\vspace{-0.3cm}
\paragraph{Robust statistics.}
A lot of prior work has provided a wide range of robust statistical estimators~\cite{morgenthaler2007survey}.
However, the theory of robust statistics has usually relied on majority-based principles.
To the best of our knowledge, our paper is the first to study Lipschitz resilience,
i.e. bounding the maximal impact of any data source,
which is arguably more adapted to a sparse setting.
Note that our algorithms rely on regularization to stabilize the estimation. 
A similar idea was previously used in signal processing~\cite{liu2019regularized}, 
in order to achieve robust mean and covariance estimation with incomplete data considering a monotone missing-data pattern.

\section{CONCLUSION}
\label{sec:conclusion}
This paper introduces the robust sparse voting problem, highlights its technical challenges, 
and presents \mehestan{}, a novel algorithm to solve it. 
Our work opens several research directions. Particularly appealing are  analyzing the strategyproofness of the system, and exploring connections with ordinal voting.
Another interesting direction is investigating properties such as order consistency and independence of irrelevant alternatives.
Overall, we regard our work as merely a first, hopefully inspiring, step towards understanding 
how a group of individuals should collaborate to securely evaluate an overwhelming amount of alternatives.

\section*{Acknowledgments}
This work was supported in part by SNSF grant 200021\_200477.
The authors are thankful to the anonymous reviewers for their constructive comments.

\bibliographystyle{plain}
\bibliography{references}

\section*{Checklist}

 \begin{enumerate}

 \item For all models and algorithms presented, check if you include:
 \begin{enumerate}
   \item A clear description of the mathematical setting, assumptions, algorithm, and/or model. [\textbf{Yes}/No/Not Applicable]
   \item An analysis of the properties and complexity (time, space, sample size) of any algorithm. [\textbf{Yes}/No/Not Applicable]
   \item (Optional) Anonymized source code, with specification of all dependencies, including external libraries. [\textbf{Yes}/No/Not Applicable]
 \end{enumerate}

 \item For any theoretical claim, check if you include:
 \begin{enumerate}
   \item Statements of the full set of assumptions of all theoretical results. [\textbf{Yes}/No/Not Applicable]
   \item Complete proofs of all theoretical results. [\textbf{Yes}/No/Not Applicable]
   \item Clear explanations of any assumptions. [\textbf{Yes}/No/Not Applicable]     
 \end{enumerate}

 \item For all figures and tables that present empirical results, check if you include:
 \begin{enumerate}
   \item The code, data, and instructions needed to reproduce the main experimental results (either in the supplemental material or as a URL). [\textbf{Yes}/No/Not Applicable]
   \item All the training details (e.g., data splits, hyperparameters, how they were chosen). [\textbf{Yes}/No/Not Applicable]
         \item A clear definition of the specific measure or statistics and error bars (e.g., with respect to the random seed after running experiments multiple times). [\textbf{Yes}/No/Not Applicable]
         \item A description of the computing infrastructure used. (e.g., type of GPUs, internal cluster, or cloud provider). [\textbf{Yes}/No/Not Applicable]
 \end{enumerate}

 \item If you are using existing assets (e.g., code, data, models) or curating/releasing new assets, check if you include:
 \begin{enumerate}
   \item Citations of the creator If your work uses existing assets. [Yes/No/\textbf{Not Applicable}]
   \item The license information of the assets, if applicable. [Yes/No/\textbf{Not Applicable}]
   \item New assets either in the supplemental material or as a URL, if applicable. [\textbf{Yes}/No/Not Applicable]
   \item Information about consent from data providers/curators. [Yes/No/\textbf{Not Applicable}]
   \item Discussion of sensible content if applicable, e.g., personally identifiable information or offensive content. [Yes/No/\textbf{Not Applicable}]
 \end{enumerate}

 \item If you used crowdsourcing or conducted research with human subjects, check if you include:
 \begin{enumerate}
   \item The full text of instructions given to participants and screenshots. [Yes/No/\textbf{Not Applicable}]
   \item Descriptions of potential participant risks, with links to Institutional Review Board (IRB) approvals if applicable. [Yes/No/\textbf{Not Applicable}]
   \item The estimated hourly wage paid to participants and the total amount spent on participant compensation. [Yes/No/\textbf{Not Applicable}]
 \end{enumerate}

 \end{enumerate}

\newpage
\onecolumn
\appendix

\section{Additional Related Work}
\label{sec:add-rel-work}

\paragraph{Social choice theory.}
This area~\cite{list2013social,endriss2017trends} is  concerned with combining individual preferences into a collective aggregate, to facilitate collective decisions.
An important result from social choice theory is Arrow's Impossibility~\cite{arrow1950difficulty},  stating that dictatorship is the only ordinal voting mechanism (i.e., voters only provide rankings instead of scores) that satisfies desirable fairness criteria. Fortunately, it is possible to escape this pessimistic result by including additional information from voters~\cite{sen1999possibility} (e.g., range voting~\cite{smith2000range}). As such, in our work, we assume that voters provide real-valued scores to alternatives (i.e., \emph{cardinal} voting). 
However, we do not require the voters to score all alternatives, nor to be all honest (i.e., non-malicious).

\noindent \textbf{Group recommender systems.}
An important use case of our work is that of group recommender systems (GRS). In contrast to a single-user recommender system (SRS)~\cite{resnick1997recommender}, a GRS~\cite{felfernig2018group} provides recommendations aimed at a group of people, rather than an individual. In practice, any algorithm used for group recommendation faces the problem of preference aggregation at some point. A categorization of the aggregation functions used in GRS, all inspired from social choice theory, was proposed in~\cite{felfernig2018algorithms}.
When the inputs of a GRS are sparse real-valued lists of ratings, especially when there could be malicious users, then our voting system can be used effectively for GRS.

In fact, there are situations where the preference aggregation techniques used in GRS can be employed to solve problems encountered in SRS.
Two such situations were discussed in~\cite{masthoff2003modeling}: (i) the cold-start problem, that is when new users join the system and the SRS has no previous data on them, and (ii) when there are multiple criteria rated for each alternative. Interestingly, single-user recommendation is another potential application of our mechanism.

\smallskip
\noindent \textbf{The scaling problem.}
When voters provide numerical ratings, preference aggregation functions need to consider the difference in users' internal rating scale ~\cite{pennock2000social,lemire2005scale,yan2013exploiting}. 
Scale invariance is then a desirable property for preference aggregation functions. 
Interestingly, it was shown in~\cite{lemire2005scale}  that some scale-invariant collaborative filtering algorithms outperform their non-invariant counterparts. The attempt to verify the scale invariance property may explain why collaborative filtering based recommender systems usually use weighted averaging to aggregate preferences~\cite{pennock2000social}.
This clearly makes such scaling-invariant solutions very vulnerable to Byzantine attacks.

A significant amount of previous work proposed solutions to address the scaling problem. 
Scale-invariant similarity metrics such as Pearson's correlation were leveraged in~\cite{lemire2005scale,su2009survey}. 
Various normalization techniques were used in~\cite{resnick1994grouplens,lemire2005scale,xie2015mathematical} to bring users' ratings to a single global scale. 
Section~\ref{sec:naive_vote} however shows that individual-based normalization techniques fail to solve the scaling problem in a sparse setting.

\smallskip
\noindent \textbf{Content recommendation with adversaries.}
Vanilla recommender systems are not immune to adversaries. 
For example, collaborative filtering is known to be vulnerable to shilling attacks~\cite{si2020shilling}, namely fake profiles and bogus ratings injection.
Meanwhile, graph-based recommender systems~\cite{fouss2007random} were shown in~\cite{fang2018poisoning} to recommend a target item a hundred times more, when only 1\% of the total users were injected fake profiles.
Most prior work investigated ways to defend against such attacks by leveraging trust (between voters). For instance,  trust network was used in~\cite{tran2009sybil} 
to limit the amount of votes collected from fake users via the max-flow concept. However, it is not clear how the quality of the vote can be evaluated.
 A trust-based SRS 
was developed in~\cite{andersen2008trust} 
(with binary recommendations) based on random walks that are incentive compatible~\cite{yu2009dsybil}.

\smallskip
\noindent \textbf{Distributed Byzantine voting.}
Voting in the presence of a limited fraction of Byzantine voters has been addressed recently in the distributed setting. This line of research tackles a generalization of the Byzantine agreement problem~\cite{lamport2019byzantine} in distributed computing. For example, the algorithms proposed in~\cite{chauhan2013democratic,tseng2017voting} yield the single winner of an election held in the presence of a fraction of Byzantine voters, using the plurality voting rule. Voters' rankings were used in~\cite{melnyk2018byzantine}  to output an aggregated ranking leveraging the Kemeny rule~\cite{kemeny1962preference}, while verifying a correctness condition.
Sparse unanimity (Definition~\ref{def:sparse-unanimity}) is closely related to the correctness properties
sought in the aforementioned works. However, in contrast, our voting mechanism tolerates sparse inputs and can utilize scores (and not just rankings).
Finally, instead of considering a constant fraction of Byzantine players, we seek $\votingresilience$-Byzantine resilience (Definition~\ref{def:byz-resilience}).

\section{Preliminaries}

\subsection{Mathematical Reminders}
\label{sec:math-reminder}

In this section, we recall the notions of Lipschitz continuity and strong convexity.
We define these notions for a general multidimensional space, although we use them for the particular one-dimensional case in the paper.
Let $d \in \setN$ be the dimension of the space $\setR^d$.
Denote by $\iprod{\cdot}{\cdot}$ the Euclidean scalar product in $\setR^d$, and by $\norm{\cdot}{2}$ the Euclidean norm.

\begin{definition}[Lipschitz continuity]
\label{def:lipschitz}
Let $L \geq 0$, and consider two metric spaces $(X, d_X)$ and $(Z, d_Z)$.
A function $g \colon X \to Z$ is $L$-Lipschitz continuous, if
\begin{equation*}
    \forall x, y \in X \mathsep d_Z(g(x), g(y)) \leq L d_X(x, y).
\end{equation*}
\end{definition}

\begin{definition}[Strong convexity]
\label{def:strong-convex}
Let $\mu \geq 0$.
A subdifferentiable function $g \colon \setR^d \to \setR$ is $\mu$-strongly convex if 
for every $x,y \in \setR^d$, and any subgradients $h_x \in \nabla g(x)$ and $h_y \in \nabla g(y)$, 
it holds that
\begin{equation*}
    \iprod{h_x - h_y}{x-y} \geq \mu \norm{x-y}{2}^2.
\end{equation*}
\end{definition}

\subsection{Generalization to Continuous Voting Rights}
\label{sec:generalization}

In this section, we generalize the setting of the main part of the paper to include continuous voting rights.
In addition to address a more general setting, which can be useful in applications~\cite{beylerian2022tournesol},
accounting for continuous voting rights yields an arguably more natural insight into the concept of Lipschitz resilience.

We start by formalizing voting rights by an assignment of a nonnegative weight $\votingrights{\voter} \geq 0$ to each voter $\voter$.
We denote $\votingrightsfamily{} = (\votingrights{1}, \ldots, \votingrights{\VOTER})$ the tuple of voting rights.
Interesting Lipschitz resilience is then defined, for fixed inputs $\voterscorefamily{}$,
a Lipschitz continuity of the map $\votingrightsfamily{} \mapsto \vote(\votingrightsfamily{}, \voterscorefamily{})$,
from $(\setR^\VOTER, \norm{\cdot}{1})$ to $(\setR^\ALTERNATIVE, \norm{\cdot}{\infty})$.
Put differently, we define Lipschitz resilience as follows.

\begin{definition}
\label{def:lipschitz-resilience}
    A vote $\vote{} : \setR^\VOTER \times X^\VOTER \rightarrow \setR^ \ALTERNATIVE$ is $\lipschitz$-Lipschitz resilient if,
    \begin{equation}
        \forall \voterscorefamily{} \in X^\VOTER \mathsep
        \forall \votingrightsfamily{}, \votingrightsfamily{}' \in \setR^\VOTER \mathsep 
        \norm{\vote{}(\votingrightsfamily{}, \voterscorefamily{}) - \vote{}(\votingrightsfamily{}', \voterscorefamily{})}{\infty}
        \leq \lipschitz \norm{\votingrightsfamily{} - \votingrightsfamily{}'}{1}.
    \end{equation}
\end{definition}

Note that, assuming that not reporting a score is tantamount to having zero voting right (which is the case for all our algorithms),
this definition clearly generalizes the one in the main part of the paper,
which amounts to flipping $\votingrights{\voter} = 1$ to $\votingrights{\voter} = 0$.

Let us also formalize the generalization of sparse unanimity to varying voting rights.
Essentially, the only modified condition is the $\VOTER_0$-scored,
which we replace by a $W_0$-scored condition 
that essentially say that the set of voters that scored an alternative must have cumulative voting right of at least $W_0$.

\begin{definition}
    An input $(\votingrightsfamily{}, \voterscorefamily{})$ is $W_0$-scored if, 
    for any alternative $\alternative \in [\ALTERNATIVE]$, 
    we have $\sum_{\voter \in \VOTER_\alternative} \votingrights{\voter} \geq W_0$.
\end{definition}

\begin{definition}
    A vote is sparsely unanimous if, 
    for all $\voterscore{*} \in \setR^\ALTERNATIVE$,
    there exists $W_0 \geq 0$ such that,
    whenever the inputs $(\votingrightsfamily{}, \voterscorefamily{})$ are $\voterscore{*}$-unanimous, comparable and $W_0$-scored,
    we have $\vote(\votingrightsfamily{}, \voterscorefamily{}) \sim \voterscore{*}$.
\end{definition}

\section{Proof of Proposition~\ref{prop:median-mean}: The Mean and Median are not Resilient}
\label{sec:app-median}

The mean and median are well-known to be generalizable to varying voting rights, 
yielding the weighted mean and the weighted median.
Note that the median may be ill-defined, as there may be multiple real numbers satisfying the corresponding properties.
In this case, we define $\median{} (\votingrightsfamily{}, \voterscorefamily{})$ to be the one closest to zero, which can be easily proven to be unique.
We recall the statement of Proposition~\ref{prop:median-mean} below for convenience.

\begin{repproposition}{prop:median-mean}
  \mean{} and \median{}
  are not Lipschitz resilient.
\end{repproposition}
\begin{proof}
First, we prove that $\averagevote$ is arbitrarily manipulable by any voter with a positive voting right.
More precisely, 
for any voter $\byzantine \in [\VOTER]$ with voting right $\votingrights{\byzantine} > 0$,
for any voting right $\votingrights{\voter}$ and score $\voterscore{\voter}$ for voters $\voter \neq \byzantine$,
and for any target score $x \in \setR$,
there exists a malicious score reporting $\voterscore{\byzantine} \in \setR$ such that 
$\averagevote(\votingrightsfamily{}, \voterscorefamily{}) = x$.
Indeed, it suffices to consider $$\voterscore{\byzantine} = \frac{1}{\votingrights{\byzantine}} \sum_{\voter \in [\VOTER]} \votingrights{\voter} x - \frac{1}{\votingrights{\byzantine}} \sum_{\voter \neq \byzantine} \votingrights{\voter} \voterscore{\voter}.$$
Arbitrare manipulability then clearly implies that $\averagevote$ cannot be Lipschitz resilient.
  
Second, we prove the assertion for $\median{}$.
  The proposition is trivial in the case $\VOTER = 1$, as $\median{}$ is then arbitrarily manipulable by the single voter.
  But to provide further insight into Lipschitz resilience, we prove it in the case where the number of voters is very large.
  Let $\lipschitz > 0$ and $\VOTER_0 \in \setN$.
  Consider $\VOTER \triangleq 2 \VOTER_0 +1$, with $\votingrights{\voter} = 1$ for all $\voter \in [\VOTER]$.
  Assume moreover that $\voterscore{\voter} \triangleq 0$ for $\voter \in [\VOTER_0]$, 
  and $\voterscore{\voter} \triangleq 2  \lipschitz$ for $\voter \in  \set{\VOTER_0+1, \VOTER_0+2, \ldots, 2\VOTER_0}$.
  Now define $\byzantine \triangleq 2 \VOTER_0+1$, i.e. the malicious voter is the last voter.
  Then the median without accounting for voter $\byzantine$ is $0$.
  But by reporting $\voterscore{\byzantine} \triangleq 2 \lipschitz$, then $\median{} (\votingrightsfamily{}, \voterscorefamily{}) = 2 \lipschitz$.
  Thus, the intervention of the malicious voter modifies the output score by $2 \lipschitz$,
  which is strictly larger than $\lipschitz$.
  Thus \median{} fails to be $\lipschitz$-Lipschitz resilient, for any value of $\lipschitz$.
  This proves the proposition.
\end{proof}
\section{Proof of Theorem~\ref{th:qrmed-br}: $\qrmedian{}$ is Lipschitz Resilient}
\label{app:qrmed}

Let $\lipschitz>0$.
We generalize $\qrmedian{\lipschitz}$ to varying voting rights as follows:
\begin{equation}
    \label{eq:def-qrmed-appendix}
  \qrmedian{\lipschitz} (\votingrightsfamily{}, \voterscorefamily{}) \triangleq \argmin_{z \in \setR} ~ 
  \set {
    \Loss_{\qrmedian{\lipschitz}} (z | \votingrightsfamily, \voterscorefamily{}) \triangleq
    \frac{1}{2 \lipschitz} z^2 + \sum_{\voter \in \VOTER^*} \votingrights{\voter} \absv{z - \voterscore{\voter}}.
  }.
\end{equation}

We recall Theorem~\ref{th:qrmed-br} below for convenience.
\begin{reptheorem}{th:qrmed-br}
$\qrmedian{\lipschitz}$ is well-defined and $\lipschitz$-Lipschitz resilient.
\end{reptheorem}
\begin{proof}
  Consider any inputs $(\votingrightsfamily{}, \voterscorefamily{})$ and $(\votingrightsfamily{}', \voterscorefamily{})$.
  Denote $q \triangleq \qrmedian{\lipschitz} (\votingrightsfamily{}, \voterscorefamily{})$
  and $q' \triangleq \qrmedian{\lipschitz} (\votingrightsfamily{}', \voterscorefamily{})$,
  and let $\Delta_{\votingrights{}} \triangleq \norm{\votingrights{} - \votingrights{}'}{1}$
  and $\Delta_q \triangleq \absv{q' - q}$.
  We aim to prove that we must have $\Delta_q \leq \lipschitz \Delta_{\votingrights{}}$.  

  Denote $\ell(z) \triangleq \Loss_{\qrmedian{\lipschitz}} (z | \votingrightsfamily{}, \voterscorefamily{})$
  and $\ell'(z) \triangleq \Loss_{\qrmedian{\lipschitz}} (z | \votingrightsfamily{}', \voterscorefamily{})$,
  Note that their difference is
  \begin{align}
    \ell(z) - \ell'(z) 
    &= \sum_{\voter \in \VOTER^*} (\votingrights{\voter} - \votingrights{\voter}') \absv{z - \voterscore{\voter}}.
  \end{align}
  The subderivatives of this difference can thus be bounded by
  \begin{align}
      \sup \absv{\partial (\ell - \ell')(z)}
      &\leq \sum_{\voter \in \VOTER^*} \absv{\votingrights{\voter} - \votingrights{\voter}'} 
        \sup \sign \left(z - \voterscore{\voter}\right) 
      \leq \sum_{\voter \in \VOTER^*} \absv{\votingrights{\voter} - \votingrights{\voter}'}
      = \norm{\votingrights{} - \votingrights{}'}{1} 
      = \Delta_{\votingrights{}}.
  \end{align}
  Thus, for any $g' \in \partial \ell'(z)$, 
  there exists $g \in \partial \ell(z)$ such that $\absv{g-g'} \leq \Delta_{\votingrights{}}$.

  Now, without loss of generality, assume that $q' > q$ (the case $q^{-\BYZANTINE} < q$ can be treated similarly).
  Now note that, since $q$ minimizes $\ell$, we must have $0 \in \partial \ell(q)$.
  Thus, in particular, $\sup \partial \ell(q) \geq 0$.
  
  Similarly, we have $\inf \partial \ell'(q') \leq 0$.
  Using what we have seen above, this implies that there must be $g \in \partial \ell(q')$ 
  such that $\absv{g - \inf \partial \ell'(q')} \leq \Delta_{\votingrights{}}$,
  which implies $g \leq \Delta_{\votingrights{}}$.
  In particular, $\inf \partial \ell(q') \leq \Delta_{\votingrights{}}$.

  Now, since $\ell$ is a sum of a quadratic term with a coefficient $(1/2\lipschitz)$, and of a convex term,
  it is clearly $(1/\lipschitz)$-strongly convex (see Definition~\ref{def:strong-convex}).
  Therefore we have
  $(\inf \partial \ell(q') - \sup \partial \ell(q)) (q' - q) \geq (q'-q)^2/ \lipschitz$,
  which implies $\inf \partial \ell(q') \geq \sup \partial \ell(q) + \Delta_q / \lipschitz \geq \Delta_q / \lipschitz$.
  Since we already showed that the left-hand side is at most $\Delta_{\votingrights{}}$,
  it follows that $\Delta_q \leq \lipschitz \Delta_{\votingrights{}}$, which is what was needed.
\end{proof}

\section{Proof of Theorem~\ref{th:brmean}: $\brmean{}$ is Resilient}
\label{app:brmean}

In this section, we provide the complete proof of the $\lipschitz$-Lipschitz resilience of $\brmean{}$.
First, we prove auxiliary properties on the clipped mean $\clippedmean{}$ operator, 
which we generalize to varying voting rights as follows:
\begin{align*}
    &\clippedmean{} (\votingrightsfamily{}, \voterscorefamily{} | \mu, \diameter) 
    \triangleq \mean (\votingrightsfamily{}, \clip(\voterscorefamily{} | \mu, \diameter)) 
    = \frac{1}{\norm{\votingrightsfamily}{1}} 
    \sum_{\voter \in [\VOTER]} \votingrights{\voter} \clip (\voterscore{\voter} | \mu, \diameter),
\end{align*}
where $\clip( x | \mu, \diameter) \triangleq \max \set{ \mu - \diameter, \min \set{ \mu + \diameter, x } }$ clips $x$ within the interval $[\mu - \diameter, \mu + \diameter]$.

\begin{lemma}
\label{lemma:clip_lipschitz_radius}
\clip{} is $1$-Lipschitz continuous with respect to the radius, i.e.
\begin{equation}
    \forall x, \mu, \diameter, \diameter' \in \setR \mathsep 
    \absv{ \clip{} (x | \mu, \diameter) - \clip{} (x | \mu, \diameter') } \leq \absv{\diameter - \diameter'}.
\end{equation}
\end{lemma}
\begin{proof}
Let $x, \mu, \diameter, \diameter' \in \setR$.
Without loss of generality, assume $\diameter' \geq \diameter$.
If $x \in [\mu - \diameter, \mu + \diameter]$, then $x \in [\mu - \diameter', \mu + \diameter']$, 
and hence $\clip{} (x | \mu, \diameter) = x = \clip{} (x | \mu, \diameter')$.
Thus, the statement clearly holds.
Otherwise, if $x > \mu + \diameter$, then $\clip{} (x | \mu, \diameter) = \mu + \diameter$ and 
$\clip{} (x | \mu, \diameter') = \min \set{ \mu + \diameter', x } \in [\mu + \diameter, \mu + \diameter']$.
We then have $\absv{ \clip{} (x | \mu, \diameter) - \clip{} (x | \mu, \diameter') } \leq (\mu + \diameter') - (\mu + \diameter) = \absv{\diameter - \diameter'}$.
The case $x < \mu + \diameter$ is treated similarly. This concludes the proof.
\end{proof}

\begin{lemma}
\label{lemma:clippedmean_lipschitz_radius}
\clippedmean{} is $1$-Lipschitz continuous with respect to the radius, i.e.
\begin{equation}
    \forall \votingrightsfamily{}, \voterscorefamily{} \mathsep 
    \forall \mu, \diameter, \diameter' \in \setR \mathsep 
    \absv{ \clippedmean{} (\votingrightsfamily{}, \voterscorefamily{} | \mu, \diameter) - \clippedmean{} (\votingrightsfamily{}, \voterscorefamily{} | \mu, \diameter') } 
    \leq \absv{\diameter - \diameter'}.
\end{equation}
\end{lemma}
\begin{proof}
Let $\mu, \diameter, \diameter' \in \setR$.
By triangle inequality and Lemma~\ref{lemma:clip_lipschitz_radius}, we have
\begin{align*}
    &\absv{ \clippedmean{} (\votingrightsfamily{}, \voterscorefamily{} | \mu, \diameter) - \clippedmean{} (\votingrightsfamily{}, \voterscorefamily{} | \mu, \diameter') } =\\ 
    &\qquad \qquad \absv{ \frac{1}{\norm{\votingrightsfamily}{1}} 
    \sum_{\voter \in [\VOTER]} \votingrights{\voter} \clip (\voterscore{\voter} | \mu, \diameter) - \frac{1}{\norm{\votingrightsfamily}{1}} 
    \sum_{\voter \in [\VOTER]} \votingrights{\voter} \clip (\voterscore{\voter} | \mu, \diameter') }\\
&\leq \frac{1}{\norm{\votingrightsfamily{}}{1}} \sum \votingrights{\voter} \absv{\clip{} (\voterscore{\voter} | \mu, \diameter) - \clip{} (\voterscore{\voter} | \mu, \diameter') } \\
&\leq \frac{1}{\norm{\votingrightsfamily{}}{1}} \sum \votingrights{\voter} \absv{ \diameter - \diameter' } = \absv{ \diameter - \diameter' }.
\end{align*}
This concludes the proof.
\end{proof}

\begin{lemma}
\label{lemma:clip_lipschitz_center}
\clip{} is $1$-Lipschitz continuous with respect to the center, i.e.
\begin{equation}
    \forall x, \diameter, \mu, \mu' \in \setR \mathsep 
    \absv{ \clip{} (x | \mu, \diameter) - \clip{} (x | \mu', \diameter) } \leq \absv{\mu - \mu'}.
\end{equation}
\end{lemma}
\begin{proof}
Let $x, \mu, \mu', \diameter \in \setR$.
Without loss of generality, assume $\mu \leq \mu'$.

Suppose for now that $\mu + \diameter \leq \mu' - \diameter$.
We then consider the five (possibly empty) intervals $(-\infty, \mu - \diameter]$, $[\mu - \diameter, \mu + \diameter]$, 
$[\mu + \diameter, \mu' - \diameter]$, $[\mu' - \diameter, \mu' + \diameter]$ and $[\mu' + \diameter, +\infty)$, which cover $\setR$.
If $x$ is in the first interval, then $\clip{} (x | \mu, \diameter) = \mu - \diameter$ and $\clip{} (x | \mu', \diameter) = \mu' - \diameter$, 
and their absolute difference is equal to $\absv{\mu-\mu'}$.
If $x$ is in the second or third interval, we have $\clip{} (x | \mu, \diameter) \geq \mu - \diameter$ and $\clip{} (x | \mu', \diameter) = \mu' - \diameter$,
and their absolute difference is thus at most $\absv{\mu-\mu'}$.
If $x$ is in the fourth interval, then we have $\clip{} (x | \mu, \diameter) = \mu + \diameter$ and $\clip{} (x | \mu', \diameter) \leq \mu' + \diameter$,
and their absolute difference is thus at most $\absv{\mu-\mu'}$.
Finally, the fifth interval is akin to the first interval.

Now assume that $\mu + \diameter \geq \mu' - \diameter$.
We now consider the five (possibly empty) intervals $(-\infty, \mu - \diameter]$, $[\mu - \diameter, \mu' - \diameter]$, 
$[\mu' - \diameter, \mu + \diameter]$, $[\mu + \diameter, \mu' + \diameter]$ and $[\mu' + \diameter, +\infty)$, which cover $\setR$.
The first and fifth intervals are still treated similarly as before.
Now assume that $x$ is in the second interval.
Then $\clip{} (x | \mu, \diameter) \geq \mu - \diameter$ and $\clip{} (x | \mu', \diameter) = \mu' - \diameter$, 
which implies that their absolute difference is at most $\absv{\mu-\mu'}$.
The fourth interval is treated symmetrically.
Finally, when $x$ is in the third interval, then $\clip{} (x | \mu, \diameter) = \clip{} (x | \mu', \diameter)$,
and their absolute difference is zero.
In all cases, the inequality of the lemma holds.
This concludes the proof.
\end{proof}

\begin{lemma}
\label{lemma:clippedmean_lipschitz_center}
\clippedmean{} is 1-Lipschitz continuous with respect to the center, i.e.
\begin{equation}
    \forall \votingrightsfamily{}, \voterscorefamily{} \mathsep 
    \forall \diameter, \mu, \mu' \in \setR \mathsep 
    \absv{ \clippedmean{} (\votingrightsfamily{}, \voterscorefamily{} | \mu, \diameter) - \clippedmean{} (\votingrightsfamily{}, \voterscorefamily{} | \mu', \diameter) } 
    \leq \absv{\mu - \mu'}.
\end{equation}
\end{lemma}
\begin{proof}
Let $\diameter, \mu, \mu' \in \setR$.
By triangle inequality and Lemma~\ref{lemma:clip_lipschitz_center}, we have
\begin{align*}
    &\absv{ \clippedmean{} (\votingrightsfamily{}, \voterscorefamily{} | \mu, \diameter) - \clippedmean{} (\votingrightsfamily{}, \voterscorefamily{} | \mu', \diameter) } =\\ 
    &\qquad \qquad \absv{ \frac{1}{\norm{\votingrightsfamily}{1}} 
    \sum_{\voter \in [\VOTER]} \votingrights{\voter} \clip (\voterscore{\voter} | \mu, \diameter) - \frac{1}{\norm{\votingrightsfamily}{1}} 
    \sum_{\voter \in [\VOTER]} \votingrights{\voter} \clip (\voterscore{\voter} | \mu', \diameter) }\\
&\leq \frac{1}{\norm{\votingrightsfamily{}}{1}} \sum \votingrights{\voter} \absv{\clip{} (\voterscore{\voter} | \mu, \diameter) - \clip{} (\voterscore{\voter} | \mu', \diameter) } \\
&\leq \frac{1}{\norm{\votingrightsfamily{}}{1}} \sum \votingrights{\voter} \absv{ \mu - \mu' } = \absv{ \mu - \mu' }.
\end{align*}
This concludes the proof.
\end{proof}

\begin{lemma}
\label{lemma:clippedmean_byzantine_resilience}
If $\votingrightsfamilybis \geq 0$, then
$\absv{ \clippedmean{} (\votingrightsfamily{}, \voterscorefamily{} | \mu, \diameter) - \clippedmean{} (\votingrightsfamily + \votingrightsfamilybis{}, \voterscorefamily{} | \mu, \diameter) } \leq 2 \frac{\norm{\votingrightsfamilybis}{1}}{\norm{\votingrightsfamily}{1}} \diameter$.
\end{lemma}

\begin{proof}
Denote $\bm{y} \triangleq \clip(\voterscorefamily{} | \mu, \diameter)$, and $\bar y \triangleq \clippedmean{} (\votingrightsfamily{}, \voterscorefamily{} | \mu, \diameter)$.
Then, we have
\begin{align*}
    &\absv{ \clippedmean{} (\votingrightsfamily{}, \voterscorefamily{} | \mu, \diameter) - \clippedmean{} (\votingrightsfamily + \votingrightsfamilybis{}, \voterscorefamily{} | \mu, \diameter) } \nonumber \\
    & =\absv{ (\bar y - \mu) - \frac{1}{\norm{\votingrightsfamily + \votingrightsfamilybis}{1}} \sum_{\voter \in [\VOTER]} (\votingrights{\voter} + \votingrightsbis_\voter) (y_\voter - \mu) } \\
    &= \absv{ (\bar y - \mu) - \frac{\norm{\votingrightsfamily}{1}}{\norm{\votingrightsfamily}{1} + \norm{\votingrightsfamilybis}{1}} (\bar y - \mu) - \frac{1}{\norm{\votingrightsfamily}{1} + \norm{\votingrightsfamilybis}{1}} \sum_{\voter \in [\VOTER]} \votingrightsbis_\voter (y_\voter - \mu) } \\
    &\leq \left( 1 - \frac{\norm{\votingrightsfamily}{1}}{\norm{\votingrightsfamily}{1} + \norm{\votingrightsfamilybis}{1}} \right) \absv{ \bar y - \mu } 
    + \frac{1}{\norm{\votingrightsfamily}{1} + \norm{\votingrightsfamilybis}{1}} \sum_{\voter \in [\VOTER]} \votingrightsbis_\voter \absv{ y_\voter - \mu } \\
    &\leq \left( 1 - \frac{\norm{\votingrightsfamily}{1}}{\norm{\votingrightsfamily}{1} + \norm{\votingrightsfamilybis}{1}} \right) \diameter 
    + \frac{1}{\norm{\votingrightsfamily}{1} + \norm{\votingrightsfamilybis}{1}} \sum_{\voter \in [\VOTER]} \votingrightsbis_\voter \diameter \\
    &\leq \frac{2 \norm{\votingrightsfamilybis{}}{1}}{\norm{\votingrightsfamily}{1} + \norm{\votingrightsfamilybis}{1}} \diameter 
    \leq \frac{2 \norm{\votingrightsfamilybis{}}{1}}{\norm{\votingrightsfamily}{1}} \diameter,
\end{align*}
which concludes the proof.
\end{proof}

Let $\lipschitz > 0$.
We generalize $\brmean{\lipschitz}$ to varying voting rights as follows:
\begin{equation*}
    \brmean{\lipschitz} (\votingrightsfamily{}, \voterscorefamily{} ) 
    \triangleq \clippedmean{} \left( \votingrightsfamily{}, \voterscorefamily{} \st \qrmedian{\lipschitz / 4} (\votingrightsfamily{}, \voterscorefamily{} ), \frac{\norm{\lipschitz \votingrightsfamily{}}{1}}{4} \right).
\end{equation*}
We finally recall and prove Theorem~\ref{th:brmean} below.

\begin{reptheorem}{th:brmean}
$\brmean{\lipschitz}$ is $\lipschitz$-resilient.
Moreover, if there exists $\diameter>0$ such that $\norm{\votingrightsfamily}{1} \geq 8 \diameter / \lipschitz$ and $\voterscore{\voter} \in [-\diameter, \diameter]$ for all $\voter$,
then $\brmean{\lipschitz} (\votingrightsfamily{}, \voterscorefamily{}) = \mean{} (\votingrightsfamily{}, \voterscorefamily{})$.
\end{reptheorem}

\begin{proof}
Denote $q_{-\BYZANTINE} \triangleq \qrmedian{\lipschitz / 4} (\votingrightsfamily_{-\BYZANTINE}, \voterscorefamily{} )$
and $q \triangleq \qrmedian{\lipschitz / 4} (\votingrightsfamily{}, \voterscorefamily{} )$.
Then, using the triangle inequality, we have
\begin{align*}
    & \absv{ \brmean{\lipschitz} (\votingrightsfamily_{-\BYZANTINE}, \voterscorefamily{}) - \brmean{\lipschitz} (\votingrightsfamily, \voterscorefamily{}) } \nonumber \\
    & = \absv{ 
        \clippedmean{} \left( \votingrightsfamily_{-\BYZANTINE}, \voterscorefamily{} \st q_{-\BYZANTINE}, \frac{\lipschitz \norm{\votingrightsfamily{}_{-\BYZANTINE}}{1}}{4} \right) 
        - \clippedmean{} \left( \votingrightsfamily, \voterscorefamily{} \st q, \frac{\lipschitz \norm{\votingrightsfamily{}}{1}}{4} \right)
    }\\
    & \leq \absv{ 
        \clippedmean{} \left( \votingrightsfamily_{-\BYZANTINE}, \voterscorefamily{} \st q_{-\BYZANTINE}, \frac{\lipschitz \norm{\votingrightsfamily{}_{-\BYZANTINE}}{1}}{4} \right) 
        - \clippedmean{} \left( \votingrightsfamily_{-\BYZANTINE}, \voterscorefamily{} \st q_{-\BYZANTINE}, \frac{\lipschitz \norm{\votingrightsfamily{}}{1}}{4} \right)
    } \nonumber \\ 
    &\quad+ \absv{ 
        \clippedmean{} \left( \votingrightsfamily_{-\BYZANTINE}, \voterscorefamily{} \st q_{-\BYZANTINE}, \frac{\lipschitz \norm{\votingrightsfamily{}}{1}}{4} \right) 
        - \clippedmean{} \left( \votingrightsfamily_{-\BYZANTINE}, \voterscorefamily{} \st q, \frac{\lipschitz \norm{\votingrightsfamily{}}{1}}{4} \right)
    } \nonumber \\ 
    &\quad+ \absv{ 
        \clippedmean{} \left( \votingrightsfamily{}_{-\BYZANTINE}, \voterscorefamily{} \st q, \frac{\lipschitz \norm{\votingrightsfamily{}}{1}}{4} \right) 
        - \clippedmean{} \left( \votingrightsfamily{}, \voterscorefamily{} \st q, \frac{\lipschitz \norm{\votingrightsfamily{}}{1}}{4} \right)
    } \\
    & \leq \frac{\lipschitz \norm{\votingrightsfamily{} - \votingrightsfamily_{-\BYZANTINE}}{1}}{4}
    + \absv{ q_{-\BYZANTINE} - q }
    + 2 \frac{\norm{\votingrightsfamily_{\BYZANTINE}}{1}}{\norm{\votingrightsfamily}{1}} \frac{\lipschitz \norm{\votingrightsfamily{}}{1}}{4}
    \leq \lipschitz \norm{\votingrightsfamily_{\BYZANTINE}}{1},
\end{align*}
where we used lemmas~\ref{lemma:clippedmean_lipschitz_radius}, \ref{lemma:clippedmean_lipschitz_center} and \ref{lemma:clippedmean_byzantine_resilience} successively.

Now, we prove the second part of the theorem.
Assume that there exists $\diameter>0$ such that $\norm{\votingrightsfamily}{1} \geq 8 \diameter / \lipschitz$ and $\voterscore{\voter} \in [-\diameter, \diameter]$ for all $\voter$.
By the latter assumption, observe that we have $q \in [-\diameter, \diameter]$, as one can check that the subderivatives of the loss minimized by $\qrmedian{}$ are positive at $\diameter$ and negative at $-\Delta$ (similar to the proof of Theorem~\ref{th:qrmed-br}).
Therefore, using the fact that $\norm{\votingrightsfamily}{1} \geq 8 \diameter / \lipschitz$, the clipping interval $[ q - \frac{\lipschitz \norm{\votingrightsfamily}{1}}{4}, q + \frac{\lipschitz \norm{\votingrightsfamily}{1}}{4} ]$ then contains all of $[-\diameter, \diameter]$, and thus also all voter scores.
Therefore, under these conditions, $\brmean{\lipschitz}$ returns the mean.
This concludes the proof.
\end{proof}

\section{Proof of Theorem~\ref{th:mehestan-br}}
\label{app:proof-main}

For convenience, we restate Theorem~\ref{th:mehestan-br}, proof of which is decomposed in lemmas in the next section.

\begin{reptheorem}{th:mehestan-br}
$\mehestan{}_{\lipschitz}$ is $\lipschitz$-resilient and sparsely unanimous.
\end{reptheorem}

\subsection{Overview of the Proof}

Below, we sketch the proof of Theorem~\ref{th:mehestan-br} in a guided proof with high-level intuitions.
We first discuss resilience, which requires controlling the worst-case impact of malicious voters at each step of \mehestan{}.
We then prove the sparse unanimity guarantee.
The missing full proofs can be found in Appendix~\ref{app:main-miss-proof}.

\subsubsection{\mehestan{} is Resilient}

We now sketch the proof of the resilience of \mehestan{}.
First, note that $\brmean{}$ ensures resilience in the search of the scaling and translation factors (steps 2 and 3 in Section~\ref{sec:mehestan}).
The guarantee a priori depends on the maximum score of each input vector $\|\voterscore{\voter}\|_{\infty}$, but the prior use of min-max pre-normalization allows to remove the dependency on this quantity.
\begin{lemma}
\label{lemma:scaling-resilience}
  For any subset $\BYZANTINE \subset [\VOTER]$, denote $\scaling{\voter}^{- \BYZANTINE}$ the scaling obtained by involving only the voters $\voter \notin \BYZANTINE$.
  Then, for any alternative $\alternative \in [\ALTERNATIVE]$, we have $\absv{ \scaling{\voter} \normalizedscore{\voter \alternative} - \scaling{\voter}^{- \BYZANTINE} \normalizedscore{\voter \alternative}} \leq \frac{\lipschitz \norm{\votingrightsfamily{}^{\BYZANTINE}}{1}}{7}$.
\end{lemma}

\begin{proof}
By the $(\lipschitz / 7)$-resilience of $\brmean{\lipschitz / 7}$ shown in Theorem~\ref{th:brmean}, we have
$\absv{ \scaling{\voter} - \scaling{\voter}^{- \BYZANTINE} } \leq \lipschitz {\norm{\votingrightsfamily{}^{-\BYZANTINE}}{1}} / {7\norm{\normalizedscore{\voter}}{\infty}}$.
It follows that, for any alternative $\alternative \in [\ALTERNATIVE]$, we have
$\absv{ \scaling{\voter} \normalizedscore{\voter \alternative} - \scaling{\voter}^{- \BYZANTINE} \normalizedscore{\voter \alternative}} 
= \absv{ \scaling{\voter}  - \scaling{\voter}^{- \BYZANTINE} } \absv{\normalizedscore{\voter \alternative}}
\leq \absv{ \scaling{\voter}  - \scaling{\voter}^{- \BYZANTINE} } \norm{\normalizedscore{\voter}}{\infty}
\leq \frac{\lipschitz \norm{\votingrightsfamily{}^{\BYZANTINE}}{1}}{7}$.
\end{proof}

\begin{lemma}
\label{lemma:translation-resilience}
  For any subset $\BYZANTINE \subset [\VOTER]$, denote $\translation{\voter}^{- \BYZANTINE}$ the translation obtained by involving only the voters $\voter \notin \BYZANTINE$.
  Then, for any alternative $\alternative \in [\ALTERNATIVE]$, we have $\absv{ \translation{\voter} - \translation{\voter}^{- \BYZANTINE} } \leq \frac{5\lipschitz \norm{\votingrightsfamily{}^{\BYZANTINE}}{1}}{7}$.
\end{lemma}

\begin{proof}[Proof sketch (The full proof is in Appendix~\ref{app:mehestan-br})]
One complication is that malicious voters $\byzantine \in \BYZANTINE$ affect the comparative translation factors $\translation{\voter \voterbis}$ 
by affecting the scaling factors $\scaling{\voter}$ and $\scaling{\voterbis}$ that appear in their computations, even when $\voterbis \notin \BYZANTINE$.
Fortunately, combining Lemma~\ref{lemma:scaling-resilience}, 
Proposition~\ref{prop:brmean-lipschitz}
allows to bound this impact.
Combining this to Theorem~\ref{th:brmean} for the direct impact of malicious voters through $\translation{\voter \byzantine}$ allows to conclude.
\end{proof}

Combining our two lemmas above, and Theorem~\ref{th:qrmed-br} on the resilience of $\qrmedian{}$ used in step 4 in Section~\ref{sec:mehestan}, 
guarantees the resilience of \mehestan{}.

\begin{lemma}
\label{lemma:mehestan-byzantine}
$\mehestan{}_{\lipschitz}$ is $\lipschitz$-Lipschitz resilient.
\end{lemma}
\begin{proof}[Proof sketch (The full proof is in Appendix~\ref{app:mehestan-br})]
The proof is akin to Lemma~\ref{lemma:translation-resilience}, 
by leveraging previous bounds, the Lipschitz continuity of $\qrmedian{}$ (Proposition~\ref{prop:qrmed-lipschitz}) and its resilience (Theorem~\ref{th:qrmed-br}).
\end{proof}

\subsubsection{\mehestan{} is Sparsely Unanimous}

We now sketch our proof of sparse unanimity, whose full proofs are provided in Appendix~\ref{app:mehestan_unanimity}.
For any $\voterscore{*} \in \setR^\ALTERNATIVE$, we first define the scaling bound
\begin{equation}
     \scalingbound{} (\voterscore{*}) \triangleq 
     \frac{\max_{\alternative, \alternativebis} \absv{\voterscore{* \alternative} - \voterscore{* \alternativebis}}}{\min_{\alternative, \alternativebis : \voterscore{* \alternative} \neq \voterscore{* \alternativebis}} \absv{\voterscore{* \alternative} - \voterscore{* \alternativebis}}},
\end{equation}
which is trivially scale-invariant.
This quantity bounds the largest possible multiplicative scaling between any voter's pre-normalized scores, and the min-max normalization $\normalizedscore{*} \triangleq \minmaxnormalizer{} (\voterscore{*})$ of $\voterscore{*}$.
Intuitively, it is an important quantity; the larger it is, the more voters will be needed to re-scale appropriately the voters' pre-normalized scores.
Interestingly, as the following lemma states it, $\scalingbound{} (\voterscore{})$ also bounds the translation discrepancies between the voters' pre-normalized scores and the min-max normalized scores $\normalizedscore{*}$.

\begin{lemma}
\label{lemma:pre-normalization-bound}
Suppose $\voterscore{*}$-unanimity and comparability.
Then, for any $\voter \in [\VOTER]$, there must exist $\scaling{\voter}^*$ and $\translation{\voter}^*$ such that 
$\normalizedscore{\voter} = \scaling{\voter}^* \normalizedscore{*|\ALTERNATIVE_\voter} + \translation{\voter}^*$,
with $1 \leq \scaling{\voter}^* \leq \scalingbound{} (\voterscore{*})$ and $- \scalingbound{} (\voterscore{*}) \leq \translation{\voter}^* \leq 0$.
\end{lemma}

\begin{proof}[Proof sketch (the full proof is 
in Appendix~\ref{app:mehestan_unanimity})]
Denoting $\alternative_\voter$ and $\alternativebis_\voter$ the best and worst alternatives scored by voter $\voter$,
as opposed to the best and worst alternatives $\alternative$ and $\alternativebis$ according to $\voterscore{*}$,
we can see that $\scaling{\voter}^* = \frac{\voterscore{* \alternative} - \voterscore{* \alternativebis}}{\voterscore{* \alternative_\voter} - \voterscore{* \alternativebis_\voter}} \in [1, \scalingbound{} (\voterscore{*})]$.
The bound on $\translation{\voter}^*$ is then obtained by looking at the normalized score of $\alternativebis_\voter$.
\end{proof}

The previous lemma says that the scaling and translation factors $\scaling{\voter}^*$ and $\translation{\voter}^*$ that must be learned for each voter are bounded. 
Thus, this will also be the case of the relative scaling and translations $\scaling{\voter \voterbis}$ and $\translation{\voter \voterbis}$.
Therefore, assuming sufficiently many voting rights have been allocated, 
applying $\brmean{}$ to such quantities will return a mean.
Combining this observation with the linearity of the mean then allows to guarantee that the adequate scaling and translation will be inferred for all voters, as precisely proved by the following lemma.

\begin{lemma}
\label{lemma:identical-rescaling}
Suppose $\voterscore{*}$-unanimity, comparability, and that alternatives are $(8 \scalingbound{} (\voterscore{*})^2 / \lipschitz)$-scored.
Then the voters' re-scaled scores are consistent, 
in the sense that $\scaling{\voter} \normalizedscore{\voter \alternative} + \translation{\voter} = \scaling{\voterbis} \normalizedscore{\voterbis \alternative} + \translation{\voterbis}$,
for all voters $\voter, \voterbis \in [\VOTER]$ and alternatives $\alternative \in \ALTERNATIVE_{\voter \voterbis}$ that both voters scored.
\end{lemma}

\begin{proof}[Sketch of proof (the full proof is 
in Appendix~\ref{app:mehestan_unanimity})]
By comparability, Lemma~\ref{lemma:pre-normalization-bound} and Theorem~\ref{th:brmean}, 
$\scaling{\voter} 
= \mean (\votingrightsfamily{}, (\scaling{\voterbis}^* / \scaling{\voter}^* )_{\voterbis \in [\VOTER]})
= \mean (\votingrightsfamily{}, \vec{\scaling{}^*}) / \scaling{\voter}^*$.
Thus $\scaling{\voter} \scaling{\voter}^*$, which is the overall multiplicative re-scaling compared to $\normalizedscore{*}$, is independent of $\voter$.
Similarly, Theorem~\ref{th:brmean} guarantees that
$\translation{\voter} = \mean{} (\votingrightsfamily{}, (\scaling{\voterbis} \translation{\voterbis}^*)_{\voterbis \in [\VOTER]}) - \scaling{\voter} \translation{\voter}^*$. 
Thus the overall translation, with respect to $\normalizedscore{*}$, is independent of the voter $\voter$.
\end{proof}

It is noteworthy that the global multiplicative rescaling is a weighted average of the values $\scaling{\voter}^*$, which are known to be at least 1.
Thus, under $(8 \scalingbound{} (\voterscore{*})^2 / \lipschitz)$-scored condition, the global normalization step (see Algorithm~\ref{alg:mehestan}) expands the scores' multiplicative scales.
Aggregating correctly scaled unanimous preferences then allows to recover these preferences.

\begin{lemma}
\label{lemma:sparse-majoritarian}
Under $\voterscore{*}$-unanimity, comparability, and that alternatives are $(8 \scalingbound{} (\voterscore{*})^2 / \lipschitz)$-scored,
$\voterscore{*}$ is recovered, i.e.
$\mehestan{}_\lipschitz (\voterscorefamily{}) \sim \voterscore{*}$.
In particular, $\mehestan_{\lipschitz}$ guarantees sparse unanimity, for $W_0 \triangleq 8 \scalingbound{} (\voterscore{*})^2 / \lipschitz$.
\end{lemma}

\begin{proof}
By Lemma~\ref{lemma:identical-rescaling}, under $\voterscore{*}$-unanimity, comparability and $(8 \scalingbound{} (\voterscore{*})^2 / \lipschitz)$-scored,
we know that there exists $\avgtruescaling \triangleq \mean (\votingrightsfamily{}, \vec{s^*}) \in [1, \scalingbound{}(\voterscore{*})]$
and $\avgtruetranslation \triangleq \mean (\votingrightsfamily{}, (\scaling{\voterbis} \translation{\voterbis}^*)_{\voterbis \in [\VOTER]}) \in [-\scalingbound{}(\voterscore{*})^2, 0]$
such that 
$\scaling{\voter} \normalizedscore{\voter \alternative} + \translation{\voter}
= \avgtruescaling \normalizedscore{* \alternative} + \avgtruetranslation
\in [-\scalingbound{}(\voterscore{*}), \scalingbound{}(\voterscore{*})^2]$
for all voters $\voter$ and alternatives $\alternative \in \ALTERNATIVE_\voter$.
Now any alternative $\alternative$ has received at least $8 \scalingbound{}(\voterscore{*})^2 / \lipschitz$ votes (measured in voting rights).
We know that these votes are all identical and equal to $\avgtruescaling \normalizedscore{* \alternative} + \avgtruetranslation$.
But then, the optimality condition of $\qrmedian{\lipschitz}$ shows that
we must have $\qrmedian{\lipschitz} (\votingrightsfamily{}, (\scaling{\voter} \normalizedscore{\voter \alternative} + \translation{\voter})_{\voter \in [\VOTER]})
= \avgtruescaling \normalizedscore{* \alternative} + \avgtruetranslation$.
Thus, for all alternatives $\alternative \in [\ALTERNATIVE]$,
we must have $\mehestan_{\lipschitz, \alternative} (\votingrightsfamily{}, \voterscorefamily{}) = \avgtruescaling \normalizedscore{* \alternative} + \avgtruetranslation$, 
which is a positive affine transformation of $\normalizedscore{*}$ (and thus of $\voterscore{*}$).
\end{proof}

\subsection{Missing Proofs}
\label{app:main-miss-proof}
In this section, we provide missing proofs of lemmas from the previous section.
\label{app:mehestan-br}

\begin{replemma}{lemma:translation-resilience}
  Denote $\translation{\voter}^{+\faulty}$ the scaling by including inputs from a new voter $\faulty \notin [\VOTER]$.
  Then, for any alternative $\alternative \in [\ALTERNATIVE]$, we have $\absv{\translation{\voter}^{+\faulty} - \translation{\voter}} \leq \frac{5 \lipschitz \votingrights{\faulty}}{7}$.
\end{replemma}
\begin{proof}
Let $\faulty \in \BYZANTINE$ and $\alternative \in [\ALTERNATIVE]$.
Let us introduce the following variable
\begin{equation}
  t_{\voter} \triangleq
  \brmean{\lipschitz / 7} 
  \set{ \votingrights{\voterbis}, t_{\voter \voterbis} 
  \st
  \voterbis \in \COMPARABLEVOTER{\voter} }.
\end{equation}
where $t_{\voter\voterbis} \triangleq \frac{1}{\card{\ALTERNATIVE_{\voter \voterbis}}}  (\sum_{\alternative \in \ALTERNATIVE_{\voter \voterbis}} \scaling{\voterbis}^{+\faulty} \normalizedscore{\voterbis\alternative}-\scaling{\voter}^{+\faulty} \normalizedscore{\voter\alternative})$.

First, we will prove that
\begin{equation}
    \label{ineq:prop-scale-resilience}
    \absv{\translation{\voter} - t_{\voter}} \leq \frac{4\lipschitz \votingrights{\faulty}}{7}.
\end{equation}

For this, consider $\xi = \left [ (\scaling{\voterbis}-\scaling{\voterbis}^{+\faulty}) \normalizedscore{\voterbis\alternative}-(\scaling{\voter}-\scaling{\voter}^{+\faulty}) \normalizedscore{\voter\alternative} \right]
_{\substack{\voterbis \in [\VOTER]\\ \alternative \in \ALTERNATIVE_{\voter\voterbis}}}.$

We deduce from Lemma~\ref{lemma:scaling-resilience} that $\norm{\xi}{\infty} \leq \frac{2\lipschitz\votingrights{\faulty}}{7}$.
Now, by using the previous $\xi$ in Proposition~\ref{prop:brmean-lipschitz} of $\brmean{}$, we obtain Inequality~\eqref{ineq:prop-scale-resilience}.
Second, by Proposition~\ref{prop:brmean-lipschitz}, we know that
\begin{equation}
    \label{ineq:prop-scale-resilience-2}
    \absv{\translation{\voter}^{+\faulty} - t_{\voter}} \leq \frac{\lipschitz \votingrights{\faulty}}{7}.
\end{equation}

We conclude by using the triangular inequality and inequalities~\eqref{ineq:prop-scale-resilience} and \eqref{ineq:prop-scale-resilience-2}.
\end{proof}

\begin{replemma}{lemma:mehestan-byzantine}
$\mehestan_{\lipschitz}$ is $\lipschitz$-Lipschitz resilient.
\end{replemma}
\begin{proof}
Let $\faulty \in \BYZANTINE$ and $\alternative \in [\ALTERNATIVE]$.
Let us introduce the following variable
\begin{equation}
  r_{\alternative} \triangleq \qrmedian{\lipschitz / 7} (\votingrightsfamily{}, (\scaling{\voter}^{+\faulty} \normalizedscore{\voter}+\translation{\voter}^{+\faulty})_{\voter \in [\VOTER]}).
\end{equation}

First, we will prove that
\begin{equation}
\label{ineq:th-byz-resilience-1}
    \absv{r_{\alternative} - \mehestan_{\lipschitz}(\votingrightsfamily{},\voterscorefamily{})} \leq \frac{6\votingrights{\faulty}}{7\votingrights{}}.
\end{equation}

For this, consider $\xi = \left [ (\scaling{\voter}-\scaling{\voter}^{+\faulty})\normalizedscore{\voter} + (\translation{\voter}-\translation{\voter}^{+\faulty}) \right]
_{\voter \in [\VOTER]}.$

We deduce from propositions~\ref{lemma:scaling-resilience} and \ref{lemma:translation-resilience} that $\norm{\xi}{\infty} \leq \frac{6\lipschitz \votingrights{\faulty}}{7}$.
Now, by using the previous $\xi$ in Proposition~\ref{prop:qrmed-lipschitz} (Lipschitz continuity of $\qrmedian{}$ in score inputs), we obtain Inequality~\eqref{ineq:th-byz-resilience-1}.
Second, by Proposition~\ref{prop:qrmed-lipschitz}, we know that
\begin{equation}
    \label{ineq:th-byz-resilience-2}
    \absv{r_{\alternative} - \mehestan_{\lipschitz}(\votingrights{},\voterscore{}^{+\faulty})} \leq \frac{\lipschitz \votingrights{\faulty}}{7}.
\end{equation}
We conclude by using the triangular inequality and inequalities~\eqref{ineq:th-byz-resilience-1} and \eqref{ineq:th-byz-resilience-2}.
\end{proof}
\label{app:mehestan_unanimity}

\begin{replemma}{lemma:pre-normalization-bound}
Suppose $\voterscore{*}$-unanimity and comparability.
Then, for any $\voter \in [\VOTER]$, there must exist $\scaling{\voter}^*$ and $\translation{\voter}^*$ such that 
$\normalizedscore{\voter} = \scaling{\voter}^* \normalizedscore{*|\ALTERNATIVE_\voter} + \translation{\voter}^*$,
with $1 \leq \scaling{\voter}^* \leq \scalingbound{} (\voterscore{*})$ and $- \scalingbound{} (\voterscore{*}) \leq \translation{\voter}^* \leq 0$.
\end{replemma}

\begin{proof}
Let $\voter \in [\VOTER]$.
By $\voterscore{*}$-unanimity, we know that 
$\normalizedscore{\voter} \sim \voterscore{\voter} \sim \voterscore{*|\ALTERNATIVE_\voter} \sim \normalizedscore{*|\ALTERNATIVE_\voter}$.
Thus there exists $\scaling{\voter}^*$ and $\translation{\voter}^*$ such that 
$\normalizedscore{\voter} = \scaling{\voter}^* \normalizedscore{*|\ALTERNATIVE_\voter} + \translation{\voter}^*$.
Now consider $\alternative_\voter$ the best-scored alternative by voter $\voter$, and $\alternativebis_\voter$ their worst-scored alternative.
By comparability, we have $\ALTERNATIVEPAIR_{\voter \voterbis} \neq \emptyset$, which implies that there exists two alternatives $\alternative, \alternativebis \in \ALTERNATIVE_\voter$ that voter $\voter$ scored differently.
Therefore, $\voterscore{\voter \alternative} > \voterscore{\voter \alternativebis}$.
Then we must have $\normalizedscore{\voter \alternative} > \normalizedscore{\voter \alternativebis}$.
In particular, by min-max normalization, we then have $1 = \max_{\alternativeter \in \ALTERNATIVE_\voter} \normalizedscore{\voter \alternativeter} - \min_{\alternativeter \in \ALTERNATIVE_\voter} \normalizedscore{\voter \alternativeter} 
= \normalizedscore{\voter \alternative} - \normalizedscore{\voter \alternativebis}
= (\scaling{\voter}^* \normalizedscore{*\alternative} + \translation{\voter}^*) - (\scaling{\voterbis}^* \normalizedscore{*\alternativebis} + \translation{\voterbis}^*)
= \scaling{\voter}^* (\normalizedscore{*\alternative} - \normalizedscore{*\alternativebis})$.
Therefore $\scaling{\voter}^* = \frac{1}{\normalizedscore{*\alternative} - \normalizedscore{*\alternativebis}} 
= \frac{\max_{\alternativeter, \alternativeter'} \absv{\normalizedscore{*\alternativeter} - \normalizedscore{*\alternativeter'}}}{\normalizedscore{*\alternative} - \normalizedscore{*\alternativebis}} 
\leq \scalingbound{} (\normalizedscore{*})
= \scalingbound{} (\voterscore{*})$.
Moreover, since by min-max normalization, we must also have $\normalizedscore{*\alternative} - \normalizedscore{*\alternativebis} \leq 1$,
we also conclude that $\scaling{\voter}^* \geq 1$.
Finally, observe that $0 = \normalizedscore{\voter \alternativebis} = \scaling{\voter}^* \normalizedscore{* \alternativebis} + \translation{\voter}^*$.
Thus $\translation{\voter} = - \scaling{\voter}^* \normalizedscore{* \alternativebis}$.
Since $\normalizedscore{* \alternativebis} \in [0,1]$, we must then have $\translation{\voter}^* \in [- \scaling{\voter}^*, 0] \subset [-\scalingbound{} (\voterscore{*}), 0]$.
\end{proof}

\begin{replemma}{lemma:identical-rescaling}
Suppose $\voterscore{*}$-unanimity, comparability, and that alternatives $(8 \scalingbound{} (\voterscore{*})^2) / \lipschitz)$-scored.
Then the voters' re-scaled scores are consistent, 
in the sense that $\scaling{\voter} \normalizedscore{\voter \alternative} + \translation{\voter} = \scaling{\voterbis} \normalizedscore{\voterbis \alternative} + \translation{\voterbis}$,
for all voters $\voter, \voterbis \in [\VOTER]$ and alternatives $\alternative \in \ALTERNATIVE_{\voter \voterbis}$ that both voters scored.
\end{replemma}

\begin{proof}
Under $\voterscore{*}$-unanimity and comparability,
by virtue of Lemma~\ref{lemma:pre-normalization-bound},
we know that $\scaling{\voter \voterbis} 
    = \frac{1}{\card{\ALTERNATIVEPAIR_{\voter \voterbis}}}
    \sum_{(\alternative, \alternativebis) \in \ALTERNATIVEPAIR_{\voter \voterbis}} 
    \frac{\absv{\normalizedscore{\voterbis \alternative} - \normalizedscore{\voterbis \alternativebis}}}{\absv{\normalizedscore{* \alternative} - \normalizedscore{* \alternativebis}}}
    \frac{\absv{\normalizedscore{* \alternative} - \normalizedscore{* \alternativebis}}}{\absv{\normalizedscore{\voter \alternative} - \normalizedscore{\voter \alternativebis}}}
    = \frac{1}{\card{\ALTERNATIVEPAIR_{\voter \voterbis}}}
    \sum_{(\alternative, \alternativebis) \in \ALTERNATIVEPAIR_{\voter \voterbis}} \frac{\scaling{\voterbis}^*}{\scaling{\voter}^*}
    = \frac{\scaling{\voterbis}^*}{\scaling{\voter}^*}$.
Using the bounds of Lemma~\ref{lemma:pre-normalization-bound}, we know that $\scaling{\voter \voterbis} \in [0, \scalingbound{}(\voterscore{*})]$.
In particular, $\absv{\scaling{\voter \voterbis}} -1 \leq \scalingbound{}(\voterscore{*})$ (using also $\scalingbound{}(\voterscore{*}) \geq 1$).
Theorem~\ref{th:brmean} (combined with comparability) then guarantees that for $\norm{\votingrightsfamily{}}{1} \geq 8 \scaling{} (\voterscore{*}) / \lipschitz$,
we have $\scaling{\voter} \triangleq 1+ \brmean{\lipschitz} (\votingrightsfamily{}, \vec{s}_{\voter} - 1) 
= 1 + \mean{} (\votingrightsfamily{}, \vec{s}_{\voter} - 1)
= \frac{1}{\norm{\votingrightsfamily{}}{1}} \sum_\voterbis \votingrights{\voterbis} \frac{\scaling{\voterbis}^*}{\scaling{\voter}^*} 
= \frac{1}{\scaling{\voter}^*} \mean (\votingrightsfamily{}, \vec{s^*})$,
where $\vec{s}_\voter \triangleq (\scaling{\voter \voterbis})_{\voterbis \in [\VOTER]}$
and $\vec{s^*} \triangleq (\scaling{\voter}^*)_{\voter \in [\VOTER]}$.
Crucially, we then have $\scaling{\voter} \scaling{\voter}^* = \mean (\votingrightsfamily{}, \vec{s^*})$, which is independent from $\voter$.
In particular, we then have $\scaling{\voter} \scaling{\voter}^* = \scaling{\voterbis} \scaling{\voterbis}^*$ for all voters $\voter, \voterbis$.
In fact, we make the additional remark that $\scaling{\voter} \scaling{\voter}^*$ is an average of values $\scaling{\voter}^*$, which all belong to $[1, \scalingbound{}(\voterscore{*})]$.
Thus $\scaling{\voter} \scaling{\voter}^* \in [1, \scalingbound{}(\voterscore{*})]$.

Now we note that we have the equality $\translation{\voter \voterbis} 
= \frac{1}{\ALTERNATIVE_{\voter \voterbis}} \sum_{\alternative \in \ALTERNATIVE_{\voter \voterbis}} \left( \scaling{\voterbis} (\scaling{\voterbis}^* \normalizedscore{* \alternative} + \translation{\voterbis}^*) - \scaling{\voter} (\scaling{\voter}^* \normalizedscore{* \alternative} + \translation{\voter}^*) \right)
= \frac{1}{\ALTERNATIVE_{\voter \voterbis}} \sum_{\alternative \in \ALTERNATIVE_{\voter \voterbis}} \left( \scaling{\voterbis} \translation{\voterbis}^* - \scaling{\voter}  \translation{\voter}^* \right) 
= \scaling{\voterbis} \translation{\voterbis}^* - \scaling{\voter}  \translation{\voter}^*$,
using the equality $\scaling{\voter} \scaling{\voter}^* = \scaling{\voterbis} \scaling{\voterbis}^*$.
But given Lemma~\ref{lemma:pre-normalization-bound},
we know that $- \scalingbound{} (\voterscore{*})^2 \leq \scaling{\voter} \translation{\voter}^* \leq 0$,
thus $\absv{\translation{\voter \voterbis}} \leq \scalingbound{} (\voterscore{*})^2$.
Theorem~\ref{th:brmean} (combined with comparability) then guarantees that for $\norm{\votingrightsfamily{}}{1} \geq 8 \scaling{} (\voterscore{*})^2 / \lipschitz$,
we have $\translation{\voter} 
\triangleq \brmean{\lipschitz} (\votingrightsfamily{}, \vec{\translation{}}_\voter) 
= \mean{} (\votingrightsfamily{}, \vec{\translation{}}_\voter)
= \mean{} (\votingrightsfamily{}, (\scaling{\voterbis} \translation{\voterbis}^*)_{\voterbis \in [\VOTER]}) - \scaling{\voter} \translation{\voter}^*$.
As a result, $\translation{\voter} + \scaling{\voter} \translation{\voter}^* = \mean{} (\votingrightsfamily{}, (\scaling{\voterbis} \translation{\voterbis}^*)_{\voterbis \in [\VOTER]})$,
which is independent from $\voter$.

We conclude by noting that, for any voter $\voter \in [\VOTER]$ and any alternative $\alternative \in \ALTERNATIVE_\voter$, we then have
$\scaling{\voter} \normalizedscore{\voter \alternative} + \translation{\voter}
= \scaling{\voter} (\scaling{\voter}^* \normalizedscore{* \alternative} + \translation{\voter}^*) + \translation{\voter}
= \scaling{\voter} \scaling{\voter}^* \normalizedscore{* \alternative} + \translation{\voter} + \scaling{\voter} \translation{\voter}^*
= \mean (\votingrightsfamily{}, \vec{s^*}) \normalizedscore{* \alternative} + \mean (\votingrightsfamily{}, (\scaling{\voterbis} \translation{\voterbis}^*)_{\voterbis \in [\VOTER]})$,
which is independent from voter $\voter$.
\end{proof}

\section{Other Proofs}
\label{sec:app-other-proofs}

\begin{proposition}
\label{prop:qrmed-lipschitz}
$\qrmedian{}$ is $1$-Lipschitz continuous with respect to the $\ell_\infty$-norm.
That is, for any $\xi \in \setR^\VOTER$, we have
    \begin{equation*}
        \absv{\qrmedian{\lipschitz} (\votingrightsfamily{}, \voterscorefamily{}+\xi) - \qrmedian{\lipschitz} (\votingrightsfamily{}, \voterscorefamily{})}
        \leq \norm{\xi}{\infty}.
    \end{equation*}
\end{proposition}
\begin{proof}
Let $\xi \in \setR^\VOTER$. For simplicity, we will use the notation $q \triangleq \qrmedian{\lipschitz} (\votingrightsfamily{}, \voterscorefamily{})$.
Let us also denote for all $\alternative \in [\ALTERNATIVE]$ the loss function as follows:
\begin{align}
  L_{\xi}(x) = \qrmedian{\lipschitz} (x|\votingrightsfamily{}, \voterscorefamily{}+\xi) \triangleq \sum_{\voter \in [\VOTER]} \votingrights{\voter } \absv{x - \voterscore{\voter} - \xi_\voter} + \frac{1}{2\lipschitz} x^2.
\end{align}

We will prove that
\begin{align}
    &\partial L_{\xi}(q+\norm{\xi}{\infty}) \cap \setR^+ \neq \varnothing, \text{ and} \label{eq:qrmed-properties-1} \\
    &\partial L_{\xi}(q-\norm{\xi}{\infty}) \cap \setR^- \neq \varnothing. \label{eq:qrmed-properties-2}
\end{align}

Given that $L_{\xi}$ is strictly convex, and since it is differentiable almost everywhere, its derivative $L_{\xi}'$ is increasing. 
We will only show \eqref{eq:qrmed-properties-1} holds, as \eqref{eq:qrmed-properties-2} is an analogous case.

Now, since $q$ minimizes $L$, we know that $0 \in \partial L(q)$. This implies that
\begin{align*}
    0 &= \frac{q}{\lipschitz} + \sum_{\voter \in [\VOTER]} \votingrights{\voter} \sign(q-\voterscore{\voter})
    \leq \frac{1}{\lipschitz} \norm{\xi}{\infty} + \frac{q}{\lipschitz} + \sum_{\voter \in [\VOTER]} \votingrights{\voter} \sign(q-\voterscore{\voter})\\
    &\leq \frac{1}{\lipschitz} \cdot (q+\norm{\xi}{\infty}) + \sum_{\voter \in [\VOTER]} \votingrights{\voter} \sign(q+\norm{\xi}{\infty}-\voterscore{\voter}-\xi_\voter)
    \in \partial L_{\xi}(q+\norm{\xi}{\infty}).
\end{align*}
This proves \eqref{eq:qrmed-properties-1} (and \eqref{eq:qrmed-properties-2} by analogy) and concludes the proof.
\end{proof}

\begin{lemma}
\label{lemma:clip-lipschitz}
\clip{} is 1-Lipschitz continuous, i.e.
\begin{equation*}
    \forall x, \xi, \diameter, \mu \in \setR \mathsep 
    \absv{ \clip{} (x + \xi | \mu, \diameter) - \clip{} (x | \mu, \diameter) } \leq \absv{\xi}.
\end{equation*}
\end{lemma}
\begin{proof}
We conclude by remarking that the function $x \mapsto \clip{}(x | \mu, \diameter)$ is piecewise, continuous, and that its subdifferential is a subset of $[0,1]$ at every point.
\end{proof}

\begin{lemma}
\label{lemma:clippedmean-lipschitz}
$\clippedmean{}$ is $1$-Lipschitz continuous with respect to the input scores.
Formally, for any $\xi \in \setR^\VOTER$, we have
\begin{equation*}
    \forall \votingrightsfamily{}, \voterscorefamily{} \mathsep 
    \forall \diameter, \mu \in \setR \mathsep 
    \absv{ \clippedmean{} (\votingrightsfamily{}, \voterscorefamily{} + \xi | \mu, \diameter) - \clippedmean{} (\votingrightsfamily{}, \voterscorefamily{} | \mu, \diameter) } 
    \leq \norm{\xi}{\infty}.
\end{equation*}
\end{lemma}
\begin{proof}
By triangle inequality and Lemma~\ref{lemma:clip_lipschitz_center}, we have
$\absv{ \clippedmean{} (\votingrightsfamily{}, \voterscorefamily{} + \xi | \mu, \diameter) - \clippedmean{} (\votingrightsfamily{}, \voterscorefamily{} | \mu, \diameter) }
\leq \frac{1}{\norm{\votingrightsfamily{}}{1}} \sum \votingrights{\voter} \absv{\clip{} (x_\voter +\xi_\voter | \mu, \diameter) - \clip{} (x_\voter | \mu, \diameter) }
\leq \frac{1}{\norm{\votingrightsfamily{}}{1}} \sum \votingrights{\voter} \absv{ \xi_\voter } \leq \norm{\xi}{\infty}$.
\end{proof}

\begin{proposition}
\label{prop:brmean-lipschitz}
$\brmean{}$ is $2$-Lipschitz continuous with respect to the $\ell_\infty$-norm.
Formally, for any $\xi \in \setR^\VOTER$, we have
    \begin{equation*}
        \absv{\brmean{\lipschitz} (\votingrightsfamily{}, \voterscorefamily{}+\xi) - \brmean{\lipschitz} (\votingrightsfamily{}, \voterscorefamily{})}
        \leq \norm{\xi}{\infty}.
    \end{equation*}
\end{proposition}
\begin{proof}
Denote $q_{+\xi} \triangleq \qrmedian{\lipschitz / 4} (\votingrightsfamily{}, \voterscorefamily{} + \xi)$
and $q \triangleq \qrmedian{\lipschitz / 4} (\votingrightsfamily{}, \voterscorefamily{} )$.
Using the triangle inequality, we have
\begin{align*}
    &\absv{\brmean{\lipschitz} (\votingrightsfamily{}, \voterscorefamily{} + \xi) - \brmean{\lipschitz} (\votingrightsfamily{}, \voterscorefamily{})}
    \\ 
    &\quad =\absv{\clippedmean{} \left( \votingrightsfamily{}, \voterscorefamily{} + \xi \st q_{+\xi}, \frac{\lipschitz \norm{\votingrightsfamily{}}{1}}{4} \right) - \clippedmean{} \left( \votingrightsfamily{}, \voterscorefamily{} \st q, \frac{\lipschitz \norm{\votingrightsfamily{}}{1}}{4} \right)} \\
    &\quad \leq \absv{\clippedmean{} \left( \votingrightsfamily{}, \voterscorefamily{} + \xi \st q_{+\xi}, \frac{\lipschitz \norm{\votingrightsfamily{}}{1}}{4} \right) - \clippedmean{} \left( \votingrightsfamily{}, \voterscorefamily{} + \xi \st q, \frac{\lipschitz \norm{\votingrightsfamily{}}{1}}{4} \right)} \\
    &\qquad + \absv{\clippedmean{} \left( \votingrightsfamily{}, \voterscorefamily{} + \xi \st q, \frac{\lipschitz \norm{\votingrightsfamily{}}{1}}{4} \right) - \clippedmean{} \left( \votingrightsfamily{}, \voterscorefamily{} \st q, \frac{\lipschitz \norm{\votingrightsfamily{}}{1}}{4} \right)} \\
    &\quad \leq \absv{q_{+\xi} - q}
    + \absv{\clippedmean{} \left( \votingrightsfamily{}, \voterscorefamily{} + \xi \st q, \frac{\lipschitz \norm{\votingrightsfamily{}}{1}}{4} \right) - \clippedmean{} \left( \votingrightsfamily{}, \voterscorefamily{} \st q, \frac{\lipschitz \norm{\votingrightsfamily{}}{1}}{4} \right)}\\
    &\quad \leq \norm{\xi}{\infty} + \absv{\clippedmean{} \left( \votingrightsfamily{}, \voterscorefamily{} + \xi \st q, \frac{\lipschitz \norm{\votingrightsfamily{}}{1}}{4} \right) - \clippedmean{} \left( \votingrightsfamily{}, \voterscorefamily{} \st q, \frac{\lipschitz \norm{\votingrightsfamily{}}{1}}{4} \right)}\\
    &\quad \leq 2\norm{\xi}{\infty},
\end{align*}
where the last three steps are successively due to Lemma~\ref{lemma:clippedmean_lipschitz_center}, Proposition~\ref{prop:qrmed-lipschitz} and Lemma~\ref{lemma:clippedmean-lipschitz}.
\end{proof}

\begin{lemma}
\label{lemma:qrmed-med}
$\qrmedian{\lipschitz} (\votingrightsfamily{}, \voterscorefamily{})$ has the same sign as $\median(\votingrightsfamily{}, \voterscorefamily{})$ and $\absv{\qrmedian{\lipschitz} (\votingrightsfamily{}, \voterscorefamily{})} \leq \absv{\median(\votingrightsfamily{}, \voterscorefamily{})}$.
\end{lemma}
\begin{proof}
Recall the notation from Equation~\ref{eq:def-qrmed-appendix}.
If $\sign(\qrmedian{\lipschitz} (\votingrightsfamily{}, \voterscorefamily{})) \neq \sign(\median(\votingrightsfamily{}, \voterscorefamily{}))$, then $\Loss_{\qrmedian{\lipschitz}}(0) < \Loss_{\qrmedian{\lipschitz}}(\qrmedian{\lipschitz} (\votingrightsfamily{}, \voterscorefamily{})) $. 
This contradicts the fact that $\qrmedian{\lipschitz} (\votingrightsfamily{}, \scorefamily{})$ minimizes $\Loss_{\qrmedian{\lipschitz}}$. This proves the first assertion.

If $\absv{\qrmedian{\lipschitz} (\votingrightsfamily{}, \voterscorefamily{})} > \absv{\median(\votingrightsfamily{}, \voterscorefamily{})}$, then we have $\Loss_{\qrmedian{\lipschitz}}(\median (\votingrightsfamily{}, \voterscorefamily{})) < \Loss_{\qrmedian{\lipschitz}}(\qrmedian{\lipschitz} (\votingrightsfamily{}, \voterscorefamily{}))$. 
This contradicts the fact that $\qrmedian{\lipschitz} (\votingrightsfamily{}, \voterscorefamily{})$ minimizes $\Loss_{\qrmedian{\lipschitz}}$.
This proves the second assertion and concludes the proof.
\end{proof}

\section{Technical lemmas}

\begin{lemma}
\label{lem:cesaro}
Let $(u_n)_{n \geq 1}$ a sequence of real numbers.
Consider the sequence $(S_n = \frac{1}{n} \sum_{k=1}^n u_k)_{n \geq 1}$.
If $(u_n)_{n \geq 1}$ converges to $l \in \setR$, then $(S_n)_{n \geq 1}$ converges to $l$ as well.
\end{lemma}
\begin{proof}
Assume $(u_n)_{n \geq 1}$ converges to $l \in \setR$.
Let $\varepsilon > 0$.

By the convergence of $(u_n)_{n \geq 1}$ to $l$, we know that there exists $N_0$ such that for all $n \geq N_0, \absv{u_n - l} \leq \epsilon/2$.
Also, we know that there exists $N_1 \geq N_0$ such that for all $n \geq N_1$, we have $\absv{\frac{1}{n} \sum_{k=1}^{N_1} u_k} \leq \epsilon/2$.
As a consequence, we have for all $n > N_1$
\begin{align*}
    \absv{S_n - l}
     \leq \absv{\frac{1}{n} \sum_{k=1}^{N_1} u_k} + \absv{\frac{1}{n} \sum_{k=N_1 + 1}^{n} (u_k - l)}
     \leq \varepsilon/2 + \varepsilon/2 = \varepsilon.
\end{align*}
\end{proof}

\begin{lemma}
\label{lem:med-cesaro}
Let $(u_n)_{n \geq 1}$ a sequence of real numbers.
Consider the sequence $(M_n = \median{}((u_k)_{1 \leq k \leq n})$.
If $(u_n)_{n \geq 1}$ converges to $l \in \setR$, then $(M_n)_{n \geq 1}$ converges to $l$ as well.
\end{lemma}
\begin{proof}
Assume $(u_n)_{n \geq 1}$ converges to $l \in \setR$.
Let $\varepsilon > 0$.
By the convergence of $(u_n)_{n \geq 1}$ to $l$, we know that there exists $N_0$ such that for all $n \geq N_0, \absv{u_n - l} \leq \epsilon$.
This implies that for all $n \geq 2 N_0 + 1$, we have $\absv{M_n - l} \leq \epsilon$. Indeed, when $n \geq 2 N_0 + 1$ there is a majority (at least $N_0+1$ among $2N_0+1$) of terms of the sequence in the interval $[l-\epsilon, l+\epsilon]$.
\end{proof}
\section{Sparse Unanimity Requires Collaborative Preference Normalization}
\label{sec:naive_vote}

In this section, we present an impossibility theorem, which roughly says that any \emph{coordinate-wise} vote with \emph{individually normalized} scores must violate sparse unanimity.
Our result highlights a central and nontrivial challenge for \emph{sparse voting}, even in the absence of disagreeing voters.
In spirit, we essentially prove that any scale-resilient sparse voting algorithm must leverage \emph{collaborative} preference scaling, as \mehestan{} does.
To formalize our impossibility theorem, we first need to introduce some assumptions on what seem to be reasonable 
\emph{individual-based normalizations} and \emph{score aggregations}.

\subsection{Individual-based Normalization}

Intuitively, any \emph{robust sparse voting} algorithm must make sure that its output will not be affected by the scaling used by voters when they report their scores.
The simplest way to guarantee this is to perform a score \emph{normalization} on voters' reported scores.
In this section, we define what a score normalization is, and what desirable properties it ought to have.
First we define a {\bf normalizer} as a function $\groupnormalizer{} : \left( \setR^{\leq \ALTERNATIVE} \right)^\VOTER \rightarrow \left( \setR^{\leq \ALTERNATIVE} \right)^\VOTER$ that preserves voters' preferences, 
i.e. such that, for any voter $\voter \in [\VOTER]$, we have $\groupnormalizer{\voter} (\voterscore{}) \sim \voterscore{\voter}$.
Below, we list other desirable properties.

\begin{definition}
\begin{enumerate}
\item  A normalizer $\groupnormalizer{}$ is {\bf individual-based} if a voter's normalized scores only depend on the voters' reported scores, i.e.
  \begin{equation}
      \exists \normalizer{} : \setR^{\leq \ALTERNATIVE} \rightarrow \setR^{\leq \ALTERNATIVE} \mathsep
      \forall \voterscorefamily{} \in \left( \setR^{\leq \ALTERNATIVE} \right)^\VOTER \mathsep
      \forall \voter \in [\VOTER] \mathsep
      \groupnormalizer{\voter} (\voterscorefamily{}) = \normalizer{} (\voterscore{\voter}),
  \end{equation}

\item  A normalizer $\groupnormalizer{}$ is {\bf scale-invariant} 
  if the normalized scores are independent of the preference scaling of the reported scores, i.e.
  \begin{equation}
    \forall \voterscorefamily{}, \voterscorefamily{}' \in \left( \setR^{\leq \ALTERNATIVE} \right)^\VOTER \mathsep
      \forall \voter \in [\VOTER] \mathsep \voterscore{\voter} \sim \voterscore{\voter}' 
      ~~\Longrightarrow~~
      \groupnormalizer{} (\voterscorefamily{}) = \groupnormalizer{} (\voterscorefamily{}').
  \end{equation}

\item  A normalizer $\groupnormalizer{}$ is {\bf neutral} if it treats all alternatives symmetrically.
  More precisely, denote $\PERMUTATION(\ALTERNATIVE)$ the set of permutations of $[\ALTERNATIVE]$.
  For any $\voterscore{} \in \setR^{\leq \ALTERNATIVE}$ and $\permutation \in \PERMUTATION(\ALTERNATIVE)$,
  we define $(\permutation \cdot \voterscore{})_\alternative \triangleq \voterscore{\permutation(\alternative)}$ 
  if the entry $\permutation(\alternative)$ of partial vector $\voterscore{}$ exists 
  (otherwise $(\permutation \cdot \voterscore{})_\alternative$ is not defined).
  Similarly, we define $(\permutation \cdot \voterscorefamily{})_\voter \triangleq \permutation \cdot \voterscore{\voter}$.
  {\bf Neutrality} then demands that
  \begin{equation}
      \forall \voterscorefamily{} \in \left( \setR^{\leq \ALTERNATIVE} \right)^\VOTER \mathsep
      \forall \permutation \in \PERMUTATION(\ALTERNATIVE) \mathsep
      \groupnormalizer{} (\permutation \cdot \voterscorefamily{})
      = \permutation \cdot \groupnormalizer{} (\voterscorefamily{}).
  \end{equation}

\item  An individual-based normalizer $\groupnormalizer{}$ is {\bf stable} if the function $x \mapsto \normalizer{} (0, x, 1)$ is Lipschitz continuous on $[0,1]$.
\end{enumerate}
\end{definition}

As an example of a (single-voter) normalizer, standardization is given by
$\standardnormalizer{\alternative}(\score{}) \triangleq \frac{\score{\alternative} - \mean(\score{})}{\stddev(\score{})}$,
where $\ALTERNATIVE_{\score{}}$ is the subset of alternatives scored by the score vector $\score{} \in \setR^{\leq \ALTERNATIVE}$, 
$\mean(\score{}) \triangleq \frac{1}{\card{\ALTERNATIVE_{\score{}}}} \sum_{\alternative \in \ALTERNATIVE_{\score{}}} \score{\alternative}$
is the mean of the scores
and $\stddev^2(\score{}) \triangleq \frac{1}{\card{\ALTERNATIVE_{\score{}}} - 1} \sum_{\alternative \in \ALTERNATIVE_{\score{}}} (\score{\alternative} - \mean(\score{}))^2$ is their standard deviation, 
assuming $\stddev(\score{}) > 0$.
If $\stddev(\score{}) = 0$, then we may simply set $\standardnormalizer{\alternative}(\score{}) \triangleq 0$
for all scored alternatives $\alternative \in \ALTERNATIVE_{\score{}}$.
Another popular normalizer is min-max normalization, given by
$\minmaxnormalizer{\alternative} (\score{}) 
\triangleq \frac{\score{\alternative} - \min_{\alternativebis \in \ALTERNATIVE_{\score{}}} \score{\alternativebis}}
{\max_{\alternativebis \in \ALTERNATIVE_{\score{}}} \score{\alternativebis} - \min_{\alternativebis \in \ALTERNATIVE_{\score{}}} \score{\alternativebis}}$,
assuming $\max_{\alternativebis \in \ALTERNATIVE_{\score{}}} \score{\alternativebis} > \min_{\alternativebis \in \ALTERNATIVE_{\score{}}} \score{\alternativebis}$
(otherwise, we return $\minmaxnormalizer{\alternative}(\score{}) \triangleq 0$ for all scored alternatives $\alternative \in \ALTERNATIVE_{\score{}}$).
Applying such single-voter normalizers to all voters clearly yield normalizers $\overrightarrow{\standardnormalizer{}}$ and $\overrightarrow{\minmaxnormalizer{}}$.

\begin{proposition}
Standardization and min-max normalizers are individual-based, scale-invariant, neutral and stable.
\end{proposition}
\begin{proof}
   They are clearly individual-based, scale-invariant and neutral normalizers.
   Plus, min-max normalizer is clearly stable.
   To show that standardization is stable, consider $g \colon x \mapsto \standardnormalizer{}(0,x,1) = \sqrt{2} (-x-1,2x-1,2-x) / \sqrt{(x+1)^2 + (2x-1)^2 + (x-2)^2} $.
   It follows that $g$ is continuously differentiable, and therefore Lipschitz continuous on $[0,1]$.
\end{proof}

\subsection{Score Aggregation}

A score aggregation is a function $\aggregation : \left( \setR_+ \times \setR^{\leq \ALTERNATIVE} \right)^\VOTER \rightarrow \setR^\ALTERNATIVE$.
Below, we identify properties that score aggregations may have.

\begin{definition}
\begin{enumerate}
\item A score aggregation $\aggregation$ is {\bf coordinate-wise} 
if the score computed by an alternative only depends on the reported scores for this alternative, i.e.,
for any alternative $\alternative \in [\ALTERNATIVE]$,
\begin{equation}
    \exists \aggregation_\alternative : \left( \setR_+ \times \setR \right)^{\leq \VOTER} \rightarrow \setR \mathsep 
    \forall \votingrightsfamily{}, \voterscorefamily{} \mathsep 
    \left( \aggregation (\votingrightsfamily{}, \voterscorefamily{}) \right)_\alternative
    = \aggregation_\alternative \left( (\votingrights{\voter})_{\voter \in \VOTER_\alternative}, (\voterscore{\voter \alternative})_{\voter \in \VOTER_\alternative} \right).
\end{equation}

\item A score aggregation $\aggregation$ is {\bf anonymous}
if it treats alternatives symmetrically, i.e.
\begin{equation}
      \forall \votingrightsfamily{} \in \setR_+^\VOTER \mathsep
      \forall \voterscorefamily{} \in \left( \setR^{\leq \ALTERNATIVE} \right)^\VOTER \mathsep
      \forall \permutation \in \PERMUTATION(\ALTERNATIVE) \mathsep
      \aggregation{} (\votingrightsfamily{}, \permutation \cdot \voterscorefamily{})
      = \permutation \cdot \aggregation{} (\votingrightsfamily{}, \voterscorefamily{}).
\end{equation}

\item A score aggregation $\aggregation$ is {\bf neutral}
if it treats voters symmetrically, i.e.
\begin{equation}
      \forall \votingrightsfamily{} \in \setR_+^\VOTER \mathsep
      \forall \voterscorefamily{} \in \left( \setR^{\leq \ALTERNATIVE} \right)^\VOTER \mathsep
      \forall \permutation \in \PERMUTATION(\VOTER) \mathsep
      \aggregation{} (\permutation \cdot \votingrightsfamily{}, \permutation \cdot \voterscorefamily{})
      = \aggregation{} (\votingrightsfamily{}, \voterscorefamily{}),
\end{equation}
where the action of $\permutation$ on $\voterscorefamily{}$ is defined by $(\permutation \cdot \voterscorefamily{})_\voter \triangleq \voterscore{\permutation(\voter)}$.

\item A coordinate-wise score aggregation $\aggregation$ is {\bf max-dominated} if each of its coordinates is dominated by the $\max$ aggregation, i.e.
\begin{equation}
    \exists \lambda > 0 \mathsep
    \forall \votingrightsfamily{} \in \setR_+^\VOTER \mathsep
    \forall \voterscore{} \in \setR^{\leq \ALTERNATIVE} \mathsep
    \norm{\aggregation(\votingrights{}, \voterscore{})}{\infty} 
    \leq \lambda \max_{\voter \in [\VOTER], \alternative \in [\ALTERNATIVE]} \absv{\voterscore{\voter \alternative}}.
\end{equation}

\item A coordinate-wise score aggregation $\aggregation$ is {\bf locally Lipschitz continuous} if each of its coordinates is locally Lipschitz continuous with respect to the $\ell_\infty$-norm.

{\label{def:asymptotically-correct}
\item A coordinate-wise score aggregation $\aggregation$ is {\bf asymptotically correct} if each of its coordinates can recover any score $\voterscore{*}$, once the input is a sufficiently large sequence converging to $\voterscore{*}$, i.e. $\lim_{\voter \rightarrow \infty} \xi_\voter = 0$ implies that
\begin{equation}
    \forall \voterscore{*} \in \setR^\ALTERNATIVE \mathsep
    \forall \varepsilon > 0 \mathsep
    \exists \VOTER_0 > 0 \mathsep
    \norm{\aggregation(\vec{1}^{\VOTER_0}, (\voterscore{*}+\xi_\voter)_{\voter \in [\VOTER_0]}) - \voterscore{*}}{\infty} 
    \leq \varepsilon,
\end{equation}
where $\vec{1}^{\VOTER_0} \in \setR^{\VOTER_0}$ is the vector whose entries all equal 1.}
\end{enumerate}
\end{definition}

\begin{lemma}
\label{lemma:qrmed_correctness}
$\qrmedian{\lipschitz}$ is asymptotically correct.
\end{lemma}
\begin{proof}
Let $x \in \setR$, and $(\xi_\voter)_{\voter \geq 1}$ a sequence of real numbers converging to $0$. We will assume that $x\geq0$, the case $x\leq0$ being analogous.
Now, let $\varepsilon>0$.
By convergence of $(\xi_\voter)_{\voter \geq 1}$, we know that there is $\VOTER_0$ such that for all $\voter \geq \VOTER_0+1$, we have $\absv{\xi_\voter} \leq \varepsilon/2$.
We then introduce $\VOTER_1 \triangleq 2\VOTER_0 + \max(0,\tfrac{1}{L}(x-\varepsilon))$ and the function $L \colon t \mapsto \frac{1}{2L}t^2 + \sum_{1 \leq \voter \leq \VOTER_1} \absv{t-x-\xi_\voter}$. The function $L$ is strictly convex, subdifferentiable, and its minimum is attained at $q \triangleq \qrmedian{L}((x+\xi_\voter)_{1 \leq \voter \leq \VOTER_1})$.
Thanks to the aforementioned result about the convergence of $(\xi_\voter)_{\voter \geq 1}$, we can show that $\partial L(x+\varepsilon) \subset \setR^+$.
Similarly, we can show that $\partial L(x-\varepsilon) \subset \setR^-$.
We conclude by convexity of $L$ that we necessarily have $q \in [x-\varepsilon,x+\varepsilon]$, i.e. $\absv{x-q} \leq \varepsilon$.
\end{proof}

\begin{proposition}
The mean, the median and $\qrmedian{}$ are all coordinate-wise, anonymous, neutral, max-dominated, locally Lipschitz continuous and asymptotically correct score aggregations.
\end{proposition}
\begin{proof}
   It is straightforward that the mean, the median and $\qrmedian{}$ are coordinate-wise, anonymous, and neutral score aggregations.
   The mean and the median are asymptotically correct by Lemmas \ref{lem:cesaro} and \ref{lem:med-cesaro} respectively. 
   $\qrmedian{}$ is asymptotically correct by Lemma~\ref{th:qrmed-br}.
   The mean and the median are trivially max-dominated. $\qrmedian{}$ is max-dominated since $\absv{\qrmedian{}} \leq \absv{\median{}}$ by Lemma~\ref{lemma:qrmed_correctness}.
   The mean is trivially Lipschitz continuous. $\qrmedian{}$ is Lipschitz continuous by Proposition~\ref{prop:qrmed-lipschitz}. Since $\median{} = \qrmedian{0}$, the median is also Lipschitz continuous.
   They are thus locally Lipschitz continuous.
\end{proof}

\subsection{The Impossibility Theorem}

We now state our impossibility theorem. 
We stress that the theorem does not assume any Byzantine voter; 
in fact, as demanded by \emph{sparse unanimity}, it assumes that all voters are honest and express the same preference $\voterscore{*}$, 
albeit each voter only scores a (potentially small) subset of all alternatives.

\begin{theorem}
\label{th:basicvote-majoritarian}
    Given any {\bf individual-based}, scale-invariant, neutral and stable normalizer $\groupnormalizer{}$ 
    and any {\bf coordinate-wise}, anonymous, neutral, max-dominated, locally Lipschitz continuous and asymptotically correct score aggregation $\aggregation$,
    $\vote{} (\votingrightsfamily{}, \voterscorefamily{}) \triangleq \aggregation (\votingrightsfamily{}, \groupnormalizer{}(\voterscorefamily{}))$ fails to be {\bf sparsely unanimous}.
\end{theorem}

\begin{proof}[Sketch of proof]
   Our proof assumes $\voterscore{* \alternative} \triangleq \alternative$, for all alternatives $\alternative \in [\ALTERNATIVE]$.
   We consider $\VOTER \triangleq K \cdot (\ALTERNATIVE - 2)$ voters, each with a unit voting right, 
   with voters $\lbrace \voter, \dots, \voter + K-1 \rbrace$ reporting the scores of alternatives 1, 2 and $\voter+2$.
   The assumptions on $\groupnormalizer{}$ and $\aggregation$ then imply 
   that most alternatives $\alternative$ will receive roughly the same score, especially for $\alternative$ large enough.
   This then implies that the vote outputs scores that are hardly correlated with $\voterscore{*}$.
   In fact, the correlation goes to $0$ in the limit $\ALTERNATIVE, K \rightarrow \infty$.
   Appendix~\ref{app:impossibility} provides the full proof.
\end{proof}

Theorem~\ref{th:basicvote-majoritarian} suggests that \emph{sparse unanimity} cannot be achieved with individual score normalization. 
Instead, Robust Sparse voting seems to require adapting a voter's score normalization based on other voters' scores, 
i.e. the score normalization must be \emph{collaborative}.
It is critical to note that this may create a vulnerability in practice, as Byzantine voters may leverage their impact on other voters' scores to
scale these scores as best fits their purposes.
Typically, whenever a voter $\voter$ prefers $\alternative$ to $\alternativebis$,
a disagreeing Byzantine voter may want to make voter $\voter$'s preference scale vanish, 
so that the vote essentially considers that voter $\voter$ is nearly indifferent between $\alternative$ and $\alternativebis$.
Instead, our solution to Robust Sparse voting builds upon new provably $\lipschitz$-Lipschitz-resilient primitives.

\subsection{Proof of our Impossibility Theorem}
\label{app:impossibility}

\begin{reptheorem}{th:basicvote-majoritarian}
    Given any {\bf individual-based}, scale-invariant, neutral and stable normalizer $\groupnormalizer{}$ 
    and any {\bf coordinate-wise}, anonymous, neutral, max-dominated, locally Lipschitz continuous and asymptotically correct score aggregation $\aggregation$,
    the vote $\vote{} (\votingrightsfamily{}, \voterscorefamily{}) \triangleq \aggregation (\votingrightsfamily, \groupnormalizer{}(\voterscorefamily{}))$ fails to be {\bf sparsely unanimous}.
\end{reptheorem}

\begin{proof}
   Consider $\voterscore{*}$ defined by $\voterscore{* \alternative} = \alternative$, for all alternatives $\alternative \in [\ALTERNATIVE]$.
   Assume that there are $\VOTER \triangleq K \cdot (\ALTERNATIVE - 2)$ voters, with voters $\lbrace \voter, \dots, \voter + K-1 \rbrace$ reporting the scores of alternatives 1, 2 and $\voter+2$.
   Note that any alternative $\alternative \geq 3$ is then scored only by the group of voters $\lbrace \alternative-2, \dots, \alternative+K-3 \rbrace$.
   We then give the same voting right $\votingrights{\voter} \triangleq \votingrights{0}$ to all voters $\voter$,
    which we choose to be large enough to guarantee the conditions of sparse unanimity.
   We will then consider the limit of this setting, as $\ALTERNATIVE, K \rightarrow \infty$.
   
   By scale invariance, we know that $\normalizer{} (\voterscore{\voter}) = \normalizer{} \left( \frac{\voterscore{\voter} - 1}{\voter +1} \right) = \normalizer{} (0,\frac{1}{\voter+1},1)$,
   where the entries of $(0,\frac{1}{\voter+1},1)$ correspond to the alternatives 1, 2 and $\voter+2$.
   Since $\normalizer{}$ is stable, we have $\normalizer{} (\voterscore{\voter}) = \normalizer{} (0,0,1) + \mathcal O(1/\voter) = (x,x,z) + \mathcal O(1/\voter)$,
   where $x,z \in \setR$ are fully determined by the normalization function $\normalizer{}$
   (note that we use the neutrality of $\normalizer{}$ to guarantee that $x$ and $z$ do not depend on which alternatives' scores have been reported by voter $\voter$).
   Moreover, we necessarily have $x \neq z$ because $\normalizer{}$ is a normalizer, and thus preserves strict order.
   Since $F$ is coordinate-wise, for any alternative $\alternative$, there exists a function $f_\alternative$ such that
   $F_\alternative (\voterscorefamily{}) = f_\alternative (\voterscorefamily{\alternative})$.
   Since $F$ is neutral, we must have $f_\alternative = f_\alternativebis \triangleq f$ for any two alternatives $\alternative, \alternativebis \in [\ALTERNATIVE]$.
   Since $F$ is anonymous and neutral, we must have $f(\voterscorefamily{\alternative}) = f(\voterscorefamily{\alternativebis})$ if both $\voterscorefamily{\alternative}$ and $\voterscorefamily{\alternativebis}$ contain the same entries (but potentially from different voters).
   By local Lipschitz continuity of $F$, if these entries are equal to $z + \mathcal O(1/\alternative)$, then  $F_\alternative(\voterscorefamily{}) = f(z,\dots,z) + \mathcal O(1/\alternative)$.
   
   Now, recall that for all $x,y \in \setR^\ALTERNATIVE \mathsep \correlation(x,y) = \langle \frac{\projection{x}}{\norm{\projection{x}}{2}}, \frac{\projection{y}}{\norm{\projection{y}}{2}} \rangle$, where for all $x \in \setR^\ALTERNATIVE,\ \projection{x} = x - \mean(x)$.

   Denote $\globalscore{} \triangleq \vote{} ( \votingrights{}, \voterscorefamily{} )$.
   We then have $\globalscore{1} = f(\vec x_1)$ and $\globalscore{2} = f(\vec x_2)$,
   where $x_{1 \voter} = x + \mathcal O(1/\voter)$ and $x_{2 \voter} = x + \mathcal O(1/\voter)$.
   Since $F$ is max-dominated, it directly follows that $f(\vec x_1), f(\vec x_2) = \mathcal O(1)$.
   For $\alternative \geq 3$, we then have $\globalscore{\alternative} = f(z,\dots,z) + \mathcal O(1/\alternative)$.
   Moreover, since $F$ is asymptotically correct, once $K$ is large enough, we have $\absv{f(z,\dots,z) - z} \leq \absv{z-x}/2$. 
   In particular, since $z \neq x$, this means that there exists a constant $c_1 > 0$ such that $\absv{f(z,\dots,z) - x} > c_1$.
   
   Thus, on one hand
   \begin{align*}
       \mean(\globalscore{}) 
       = \frac{\mathcal O(1) + \mathcal O(1) + \ALTERNATIVE f(z,\dots,z) + \sum_{\alternative} \mathcal O(1/\alternative)}{\ALTERNATIVE} 
       = f(z,\dots,z) + \mathcal O \left( \frac{\ln \ALTERNATIVE}{\ALTERNATIVE} \right).
   \end{align*}
   
   And on the other hand, we have that there exists a constant $c_2 > 0$ such that
    \begin{align*}
       \norm{\globalscore{} - \mean(\globalscore{})}{2}^2 > c_2.
   \end{align*}

   Also, thanks to the choice of $\voterscore{*}$, by computing $\frac{\projection{\voterscore{*}}}{\norm{\projection{\voterscore{*}}}{2}}$ we have
   \begin{equation*}
       \norm{\frac{\projection{\voterscore{*}}}{\norm{\projection{\voterscore{*}}}{2}}}{\infty} = \mathcal{O}(\frac{1}{\sqrt{\ALTERNATIVE}}).
   \end{equation*}
   
   We can now conclude that
   \begin{align*}
       \absv{\correlation(\globalscore{}, \voterscore{*})}
       & = \absv{\langle  \frac{\projection{\voterscore{*}}}{\norm{\projection{\voterscore{*}}}{2}},
       \frac{\projection{\globalscore{}}}{\norm{\projection{\globalscore{}}}{2}}
       \rangle} \\
       & \leq \norm{\frac{\projection{\voterscore{*}}}{\norm{\projection{\voterscore{*}}}{2}}}{\infty} \cdot \left ( \mathcal{O}(1) + \mathcal{O}(1) + \ALTERNATIVE \mathcal{O}(\frac{\ln \ALTERNATIVE}{\ALTERNATIVE}) + \sum_\alternative \mathcal{O}(\frac{1}{\alternative}) \right ) \\
       &= \mathcal{O}(\frac{\ln \ALTERNATIVE}{\sqrt{\ALTERNATIVE}}).
   \end{align*}
\end{proof}

\section{Additional Experiments}
\label{sec:additional-exp}

In this section, we present the results of additional experiments aiming at testing \mehestan{}'s ability to tolerate sparsity.
Recall the experimental setup from Section~\ref{sec:experiments} in the main body.
We conduct the following additional experiments: 
\begin{itemize}
    \item the same experiments as Section~\ref{sec:experiments} with $\voterscore{*}$ following a Uniform (Figure~\ref{fig:uniform}), Gaussian (Figure~\ref{fig:gaussian}) and Cauchy (Figure~\ref{fig:cauchy}) distribution. All were performed both with and without biased sparsity.
    \item the measurement of the performance depending on the sparsity bias, with $\voterscore{*}$ following a Uniform, Gaussian, and Cauchy distribution (Figure~\ref{fig:bias}).
\end{itemize}

\begin{figure}[H]
\centering
\begin{subfigure}{.48\textwidth}
  \centering
  \includegraphics[width=\linewidth]{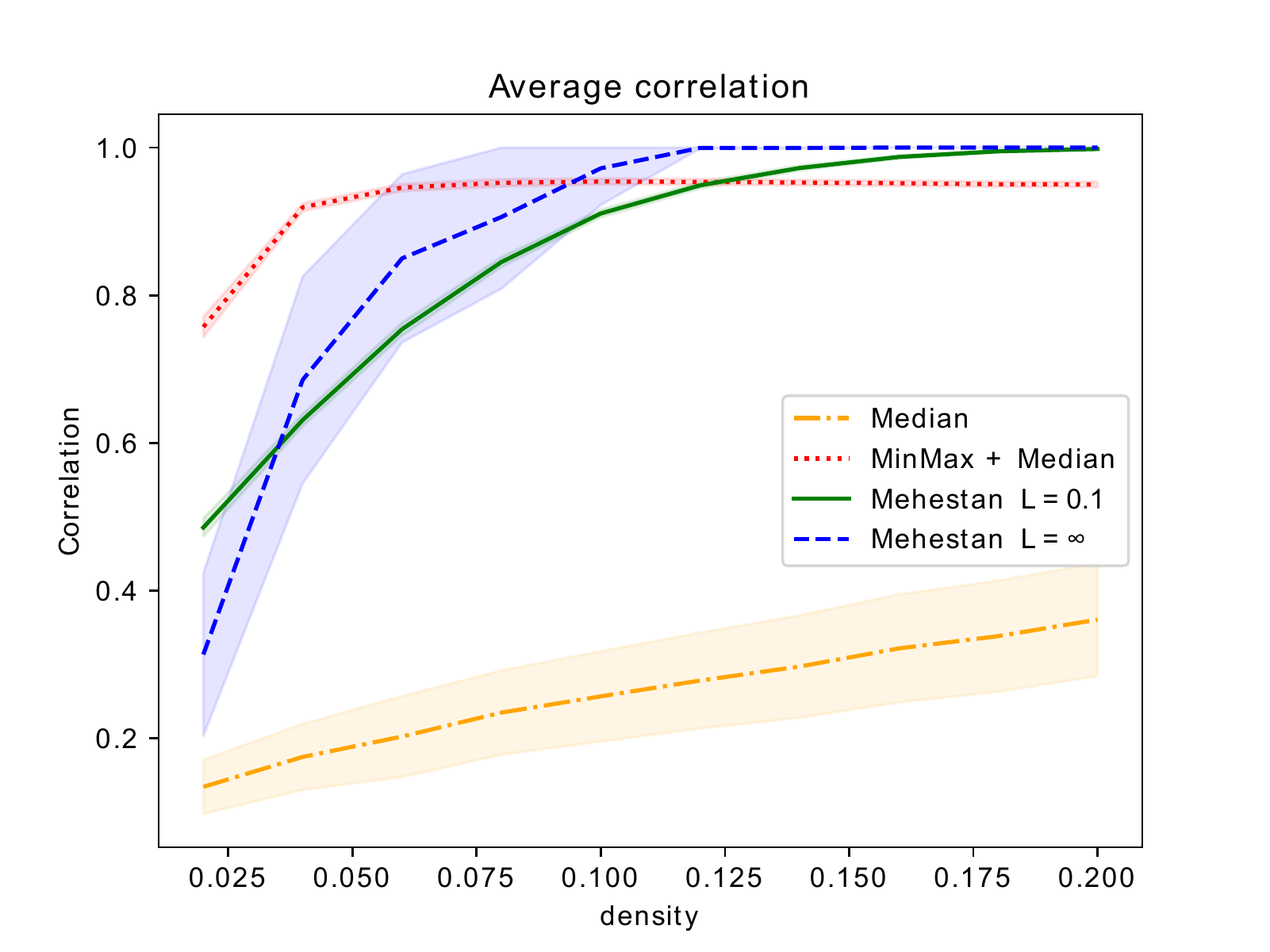}
  \caption{Influence of the \textbf{density} (with bias in sparsity)}
  \label{fig:sm}
\end{subfigure}\hfill
\begin{subfigure}{.48\textwidth}
  \centering
  \includegraphics[width=\linewidth]{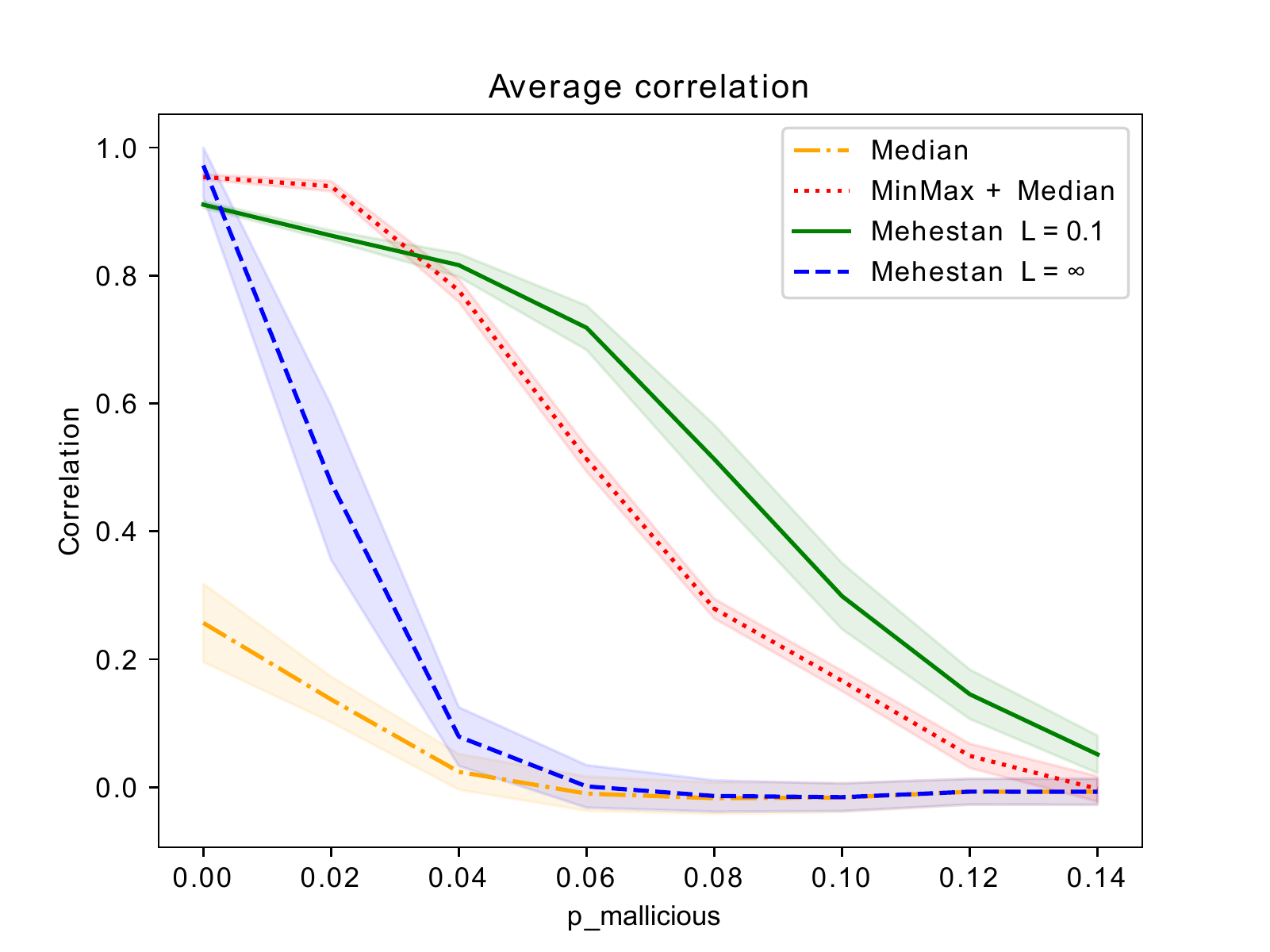}
  \caption{Influence of the \textbf{fraction of malicious voters} (with bias in sparsity)}
  \label{fig:nextreme}
\end{subfigure}
\begin{subfigure}{.48\textwidth}
  \centering
  \includegraphics[width=\linewidth]{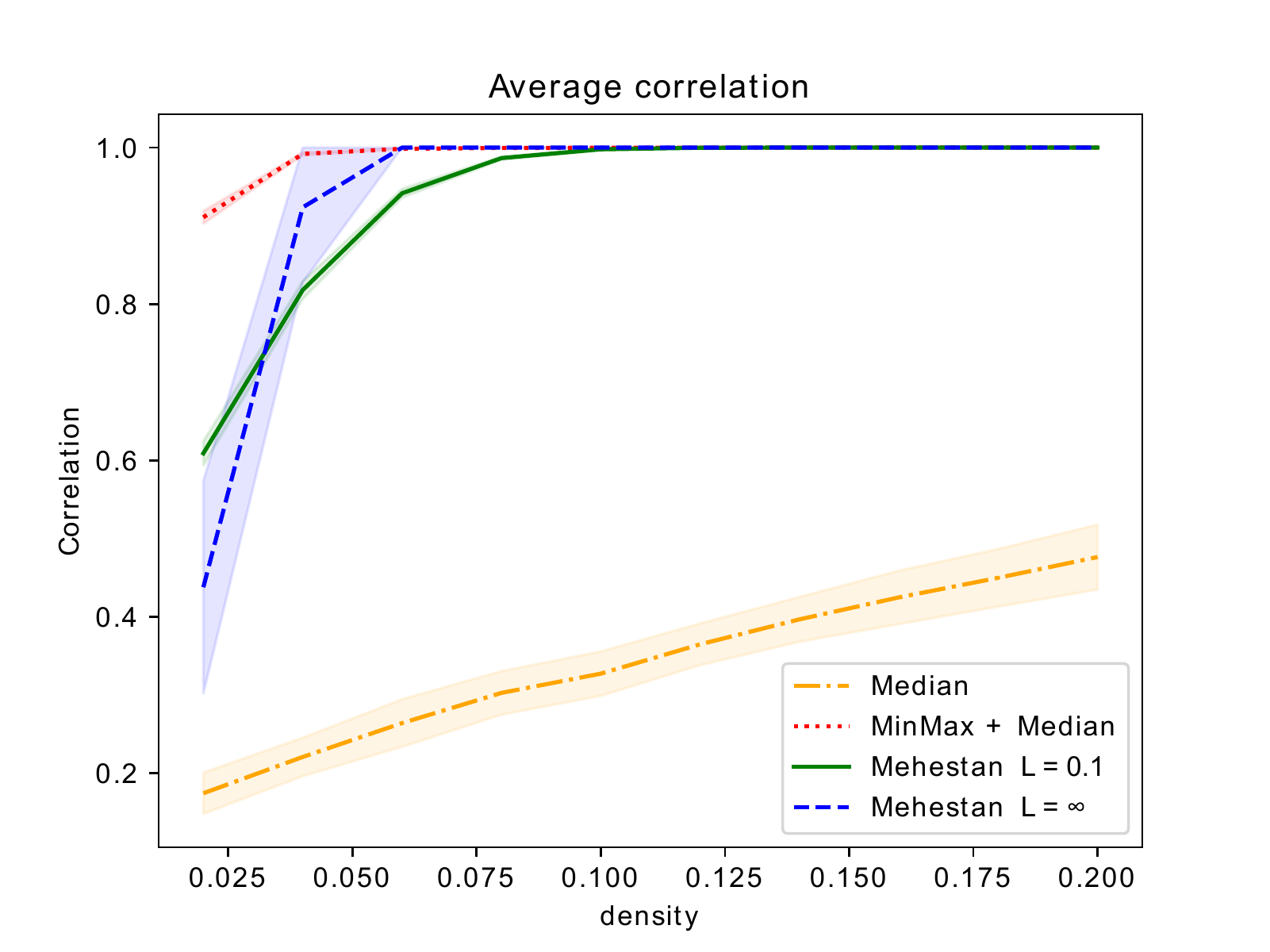}
  \caption{Influence of the \textbf{density} (without bias in sparsity)}
  \label{fig:sm}
\end{subfigure}\hfill
\begin{subfigure}{.48\textwidth}
  \centering
  \includegraphics[width=\linewidth]{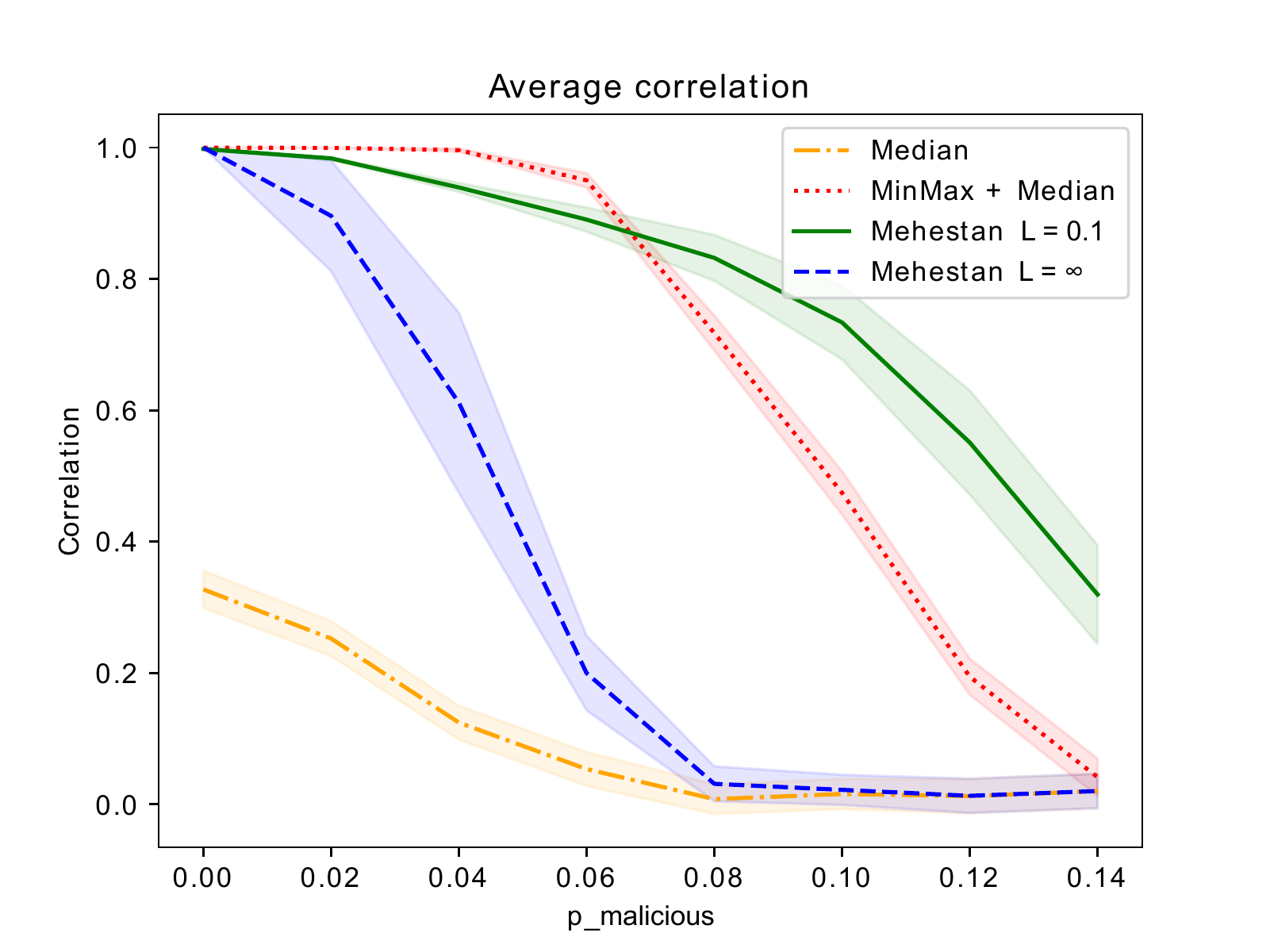}
  \caption{Influence of the \textbf{fraction of malicious voters} (without bias in sparsity)}
  \label{fig:nextreme}
\end{subfigure}
\caption{
Performance of \mehestan{} under sparsity, with and without malicious voters (\textbf{Uniform} distribution of $\theta_*$)
}
\label{fig:uniform}
\end{figure}

\begin{figure}[H]
\centering
\begin{subfigure}{.48\textwidth}
  \centering
  \includegraphics[width=\linewidth]{img/gaussian/dens_extr.pdf}
  \caption{Influence of the \textbf{density} (with bias in sparsity)}
  \label{fig:sm}
\end{subfigure}\hfill
\begin{subfigure}{.48\textwidth}
  \centering
  \includegraphics[width=\linewidth]{img/gaussian/byz_extr.pdf}
  \caption{Influence of the \textbf{fraction of malicious voters} (with bias in sparsity)}
  \label{fig:nextreme}
\end{subfigure}
\begin{subfigure}{.48\textwidth}
  \centering
  \includegraphics[width=\linewidth]{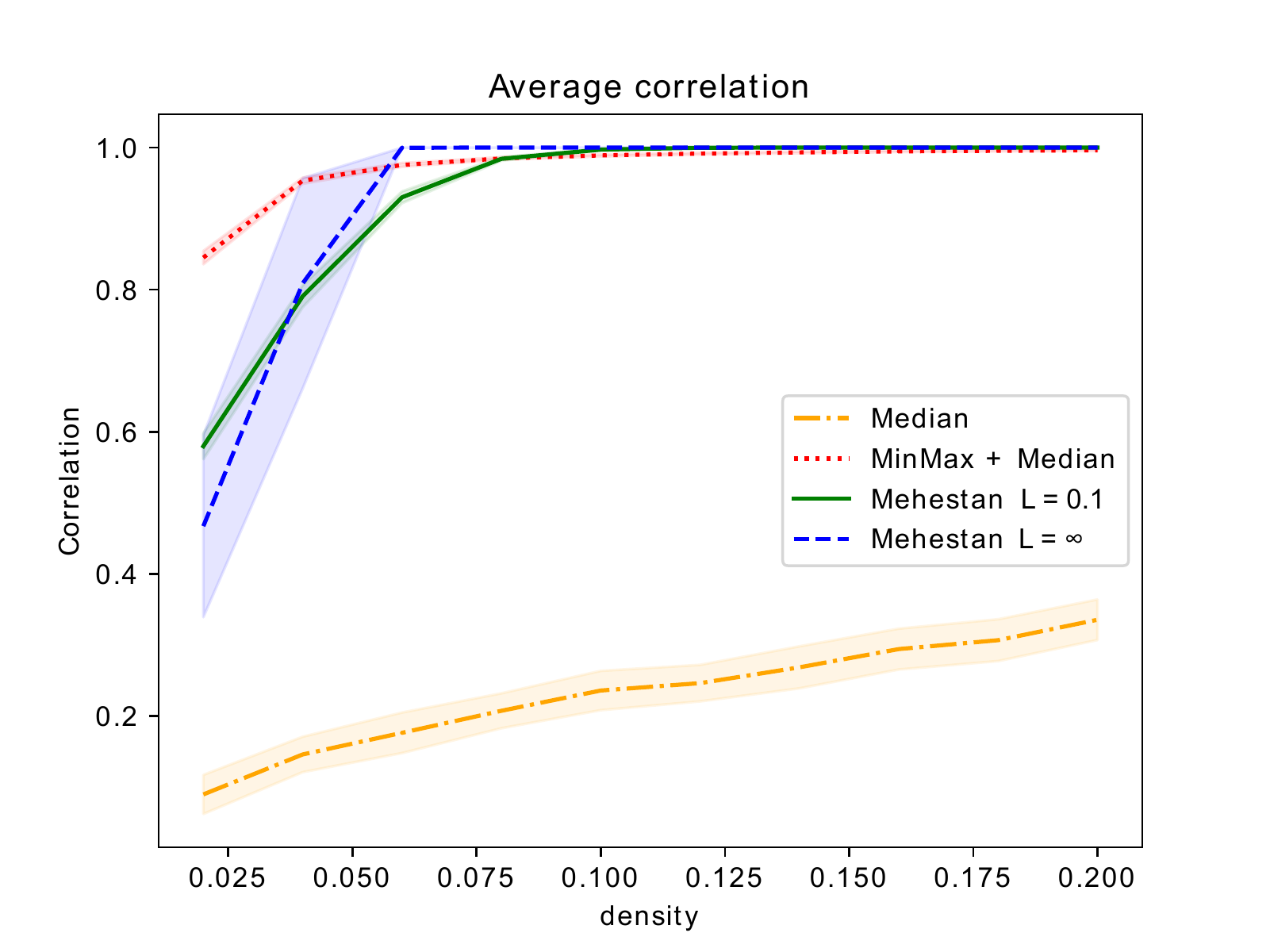}
  \caption{Influence of the \textbf{density} (without bias in sparsity)}
  \label{fig:sm}
\end{subfigure}\hfill
\begin{subfigure}{.48\textwidth}
  \centering
  \includegraphics[width=\linewidth]{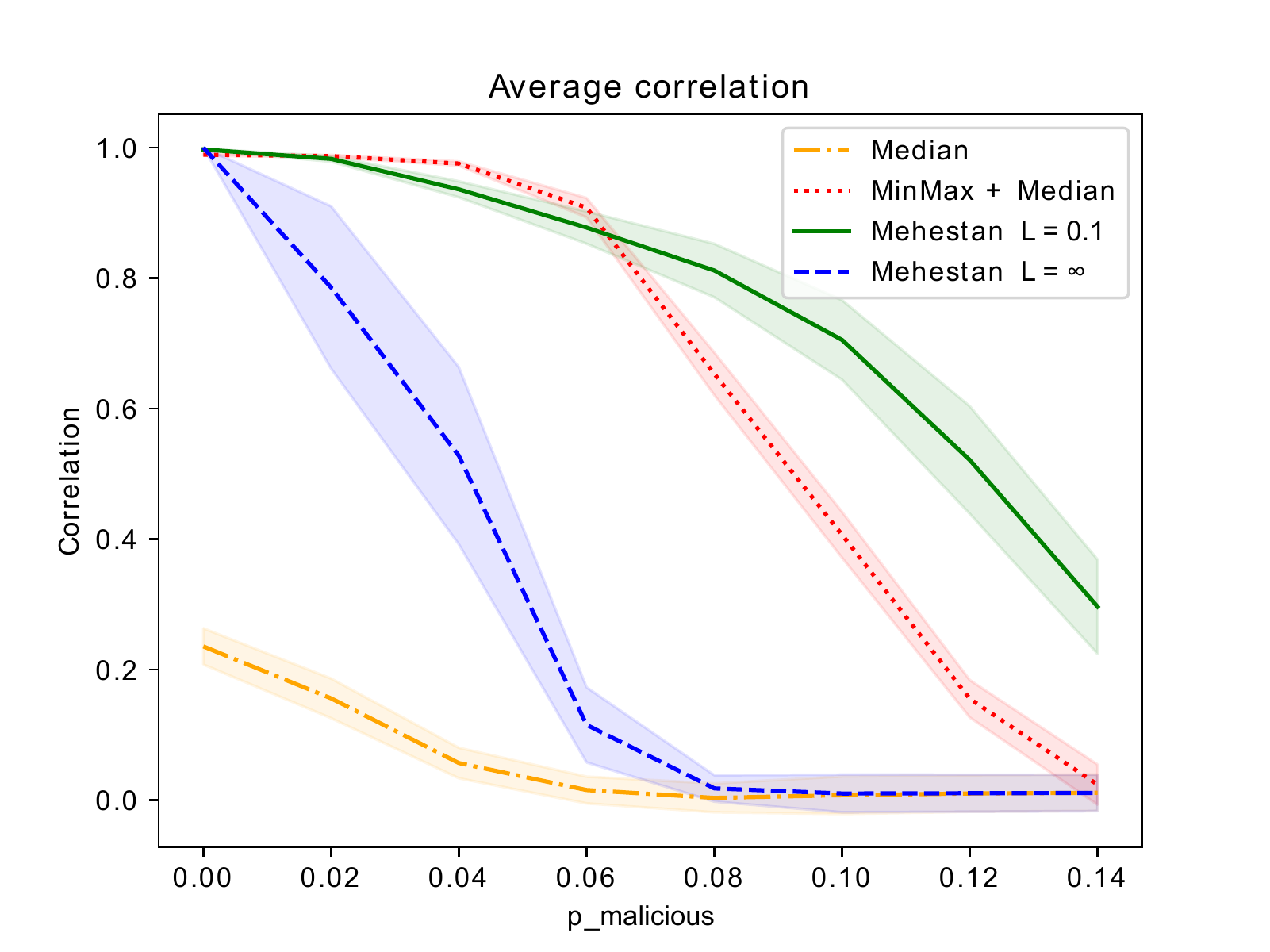}
  \caption{Influence of the \textbf{fraction of malicious voters} (without bias in sparsity)}
  \label{fig:nextreme}
\end{subfigure}
\caption{
Performance of \mehestan{} under sparsity, with and without malicious voters. (\textbf{Gaussian} distribution of $\theta_*$)
}
\label{fig:gaussian}
\end{figure}

\begin{figure}[H]
\centering
\begin{subfigure}{.48\textwidth}
  \centering
  \includegraphics[width=\linewidth]{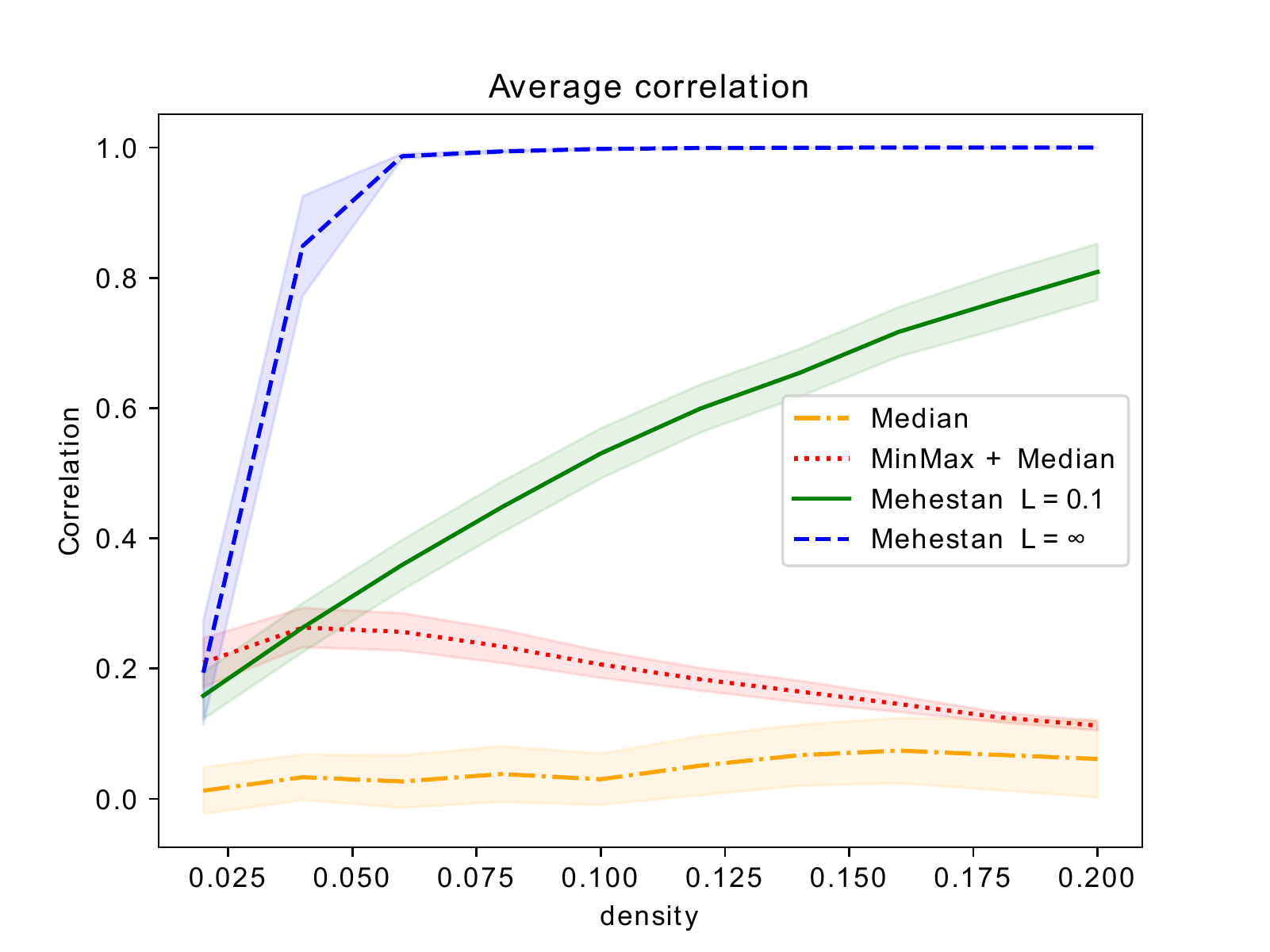}
  \caption{Influence of the \textbf{density} (with bias in sparsity)}
  \label{fig:sm}
\end{subfigure}\hfill
\begin{subfigure}{.48\textwidth}
  \centering
  \includegraphics[width=\linewidth]{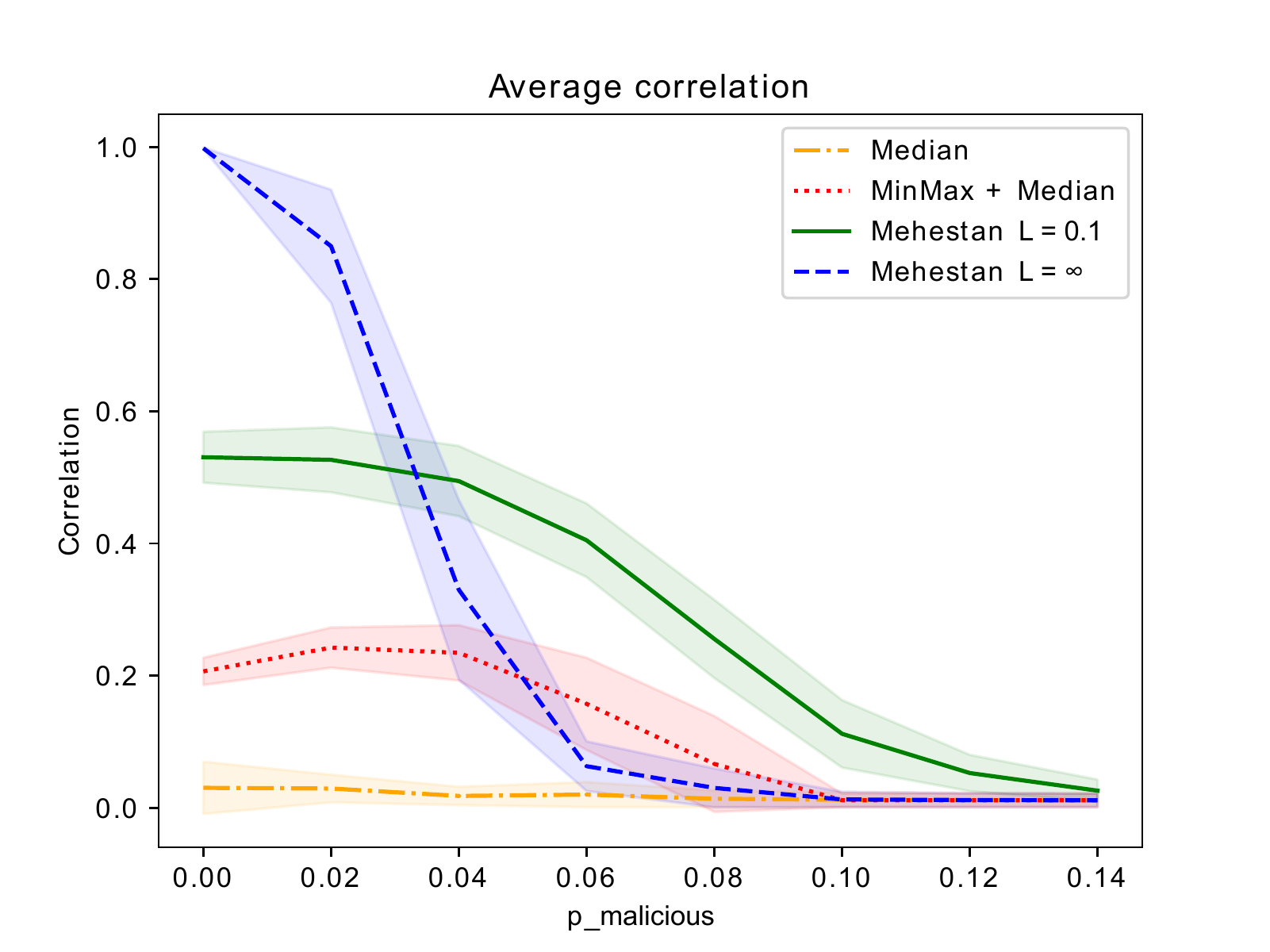}
  \caption{Influence of the \textbf{fraction of malicious voters} (with bias in sparsity)}
  \label{fig:nextreme}
\end{subfigure}
\begin{subfigure}{.48\textwidth}
  \centering
  \includegraphics[width=\linewidth]{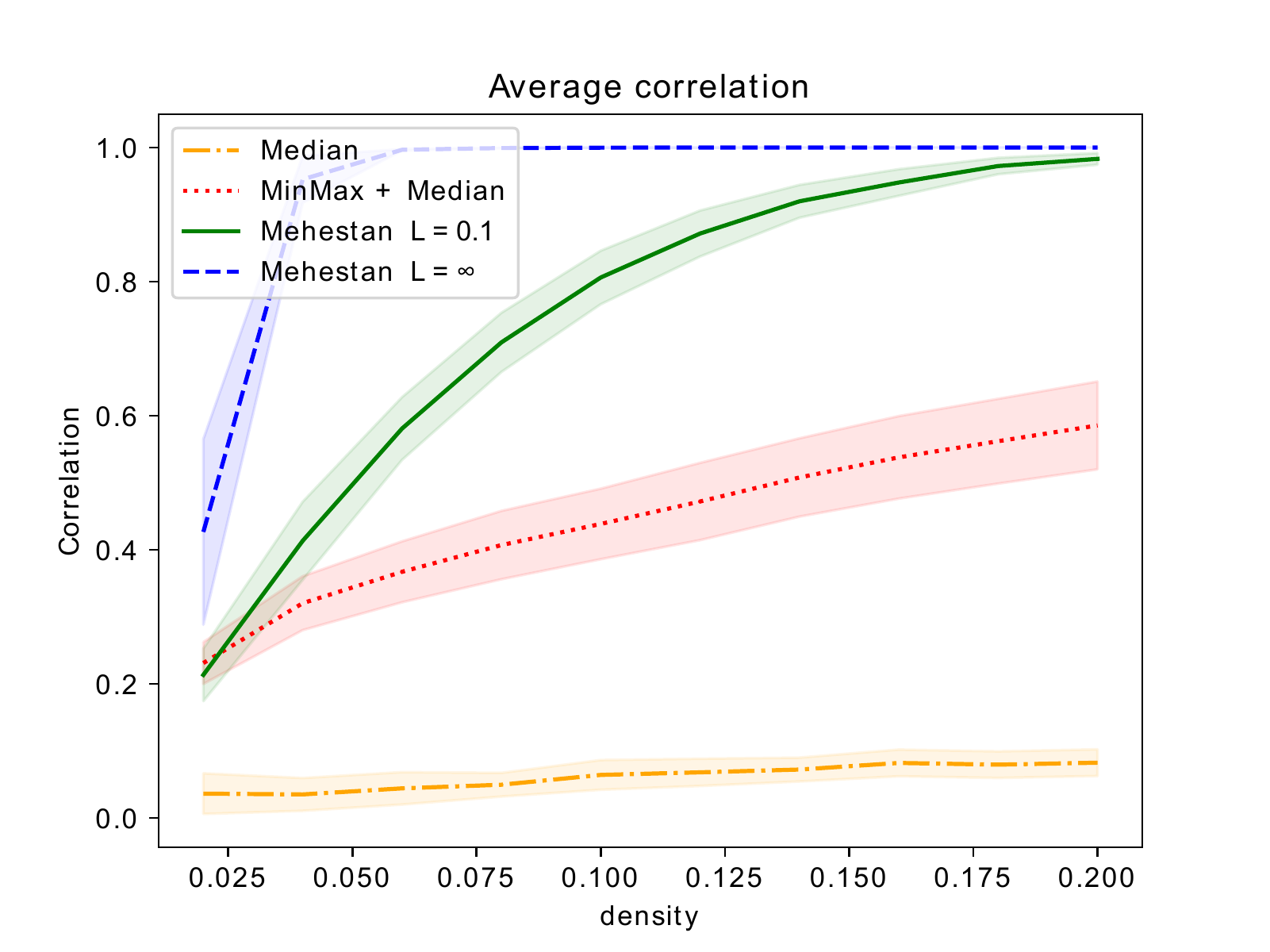}
  \caption{Influence of the \textbf{density} (without bias in sparsity)}
  \label{fig:sm}
\end{subfigure}\hfill
\begin{subfigure}{.48\textwidth}
  \centering
  \includegraphics[width=\linewidth]{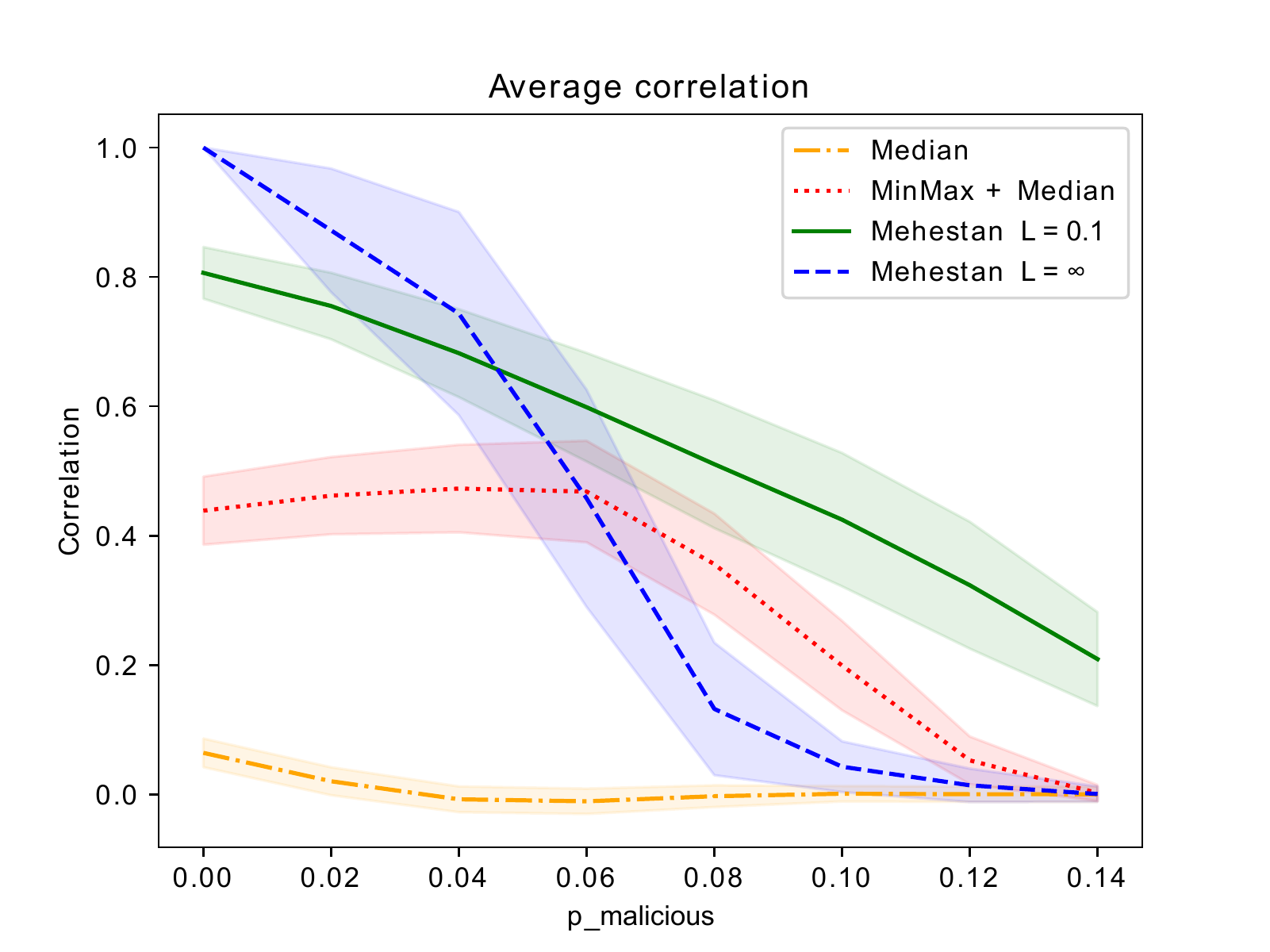}
  \caption{Influence of the \textbf{fraction of malicious voters} (without bias in sparsity)}
  \label{fig:nextreme}
\end{subfigure}
\caption{
Performance of \mehestan{} under sparsity, with and without malicious voters. (\textbf{Cauchy} distribution of $\theta_*$)
}\label{fig:cauchy}
\end{figure}

\begin{figure}[H]
\centering

\begin{subfigure}{.48\textwidth}
  \centering
  \includegraphics[width=\linewidth]{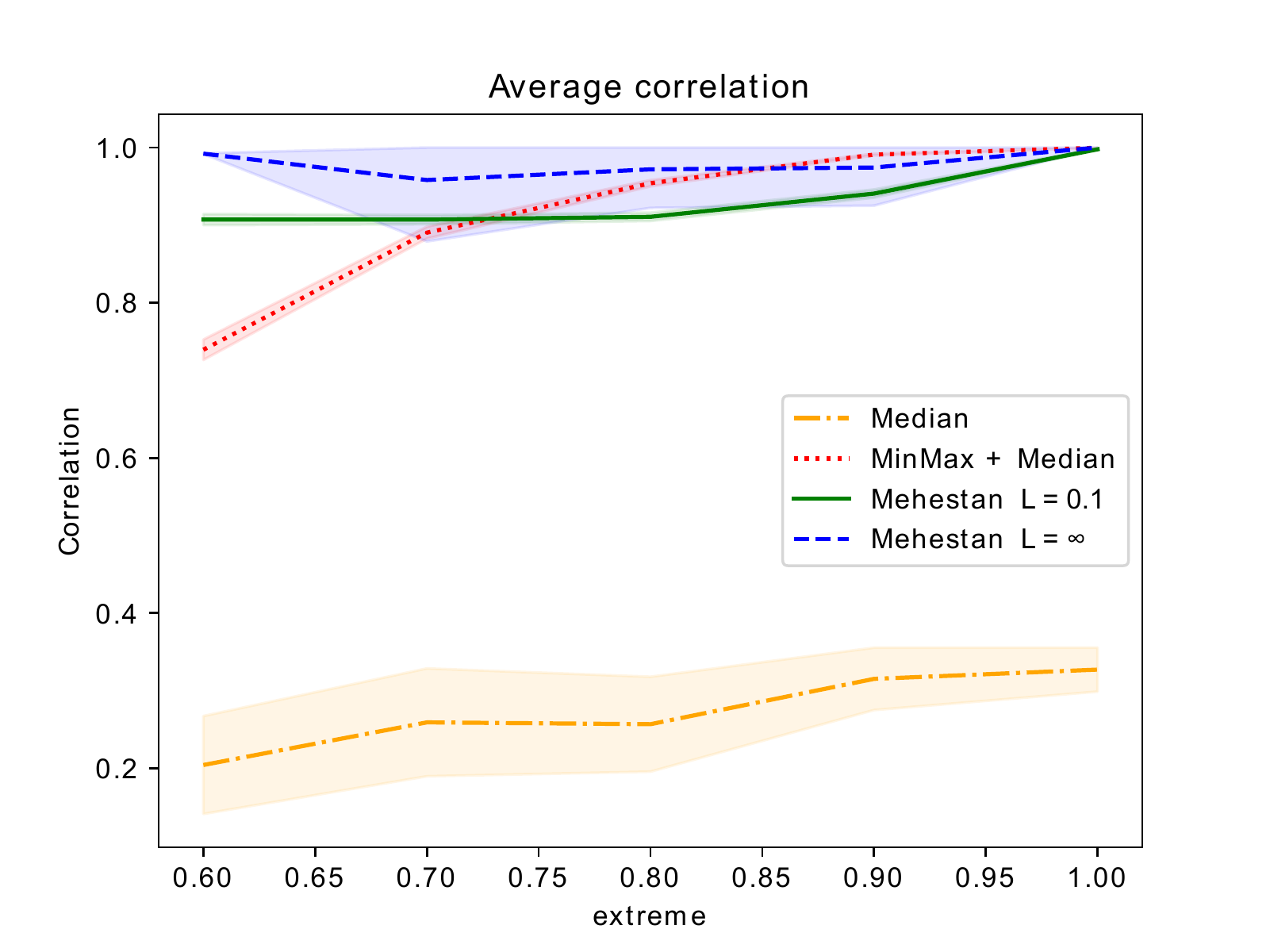}
  \caption{\textbf{Uniform} distribution of $\theta_*$}
  \label{fig:nextreme}
\end{subfigure}
\begin{subfigure}{.48\textwidth}
  \centering
  \includegraphics[width=\linewidth]{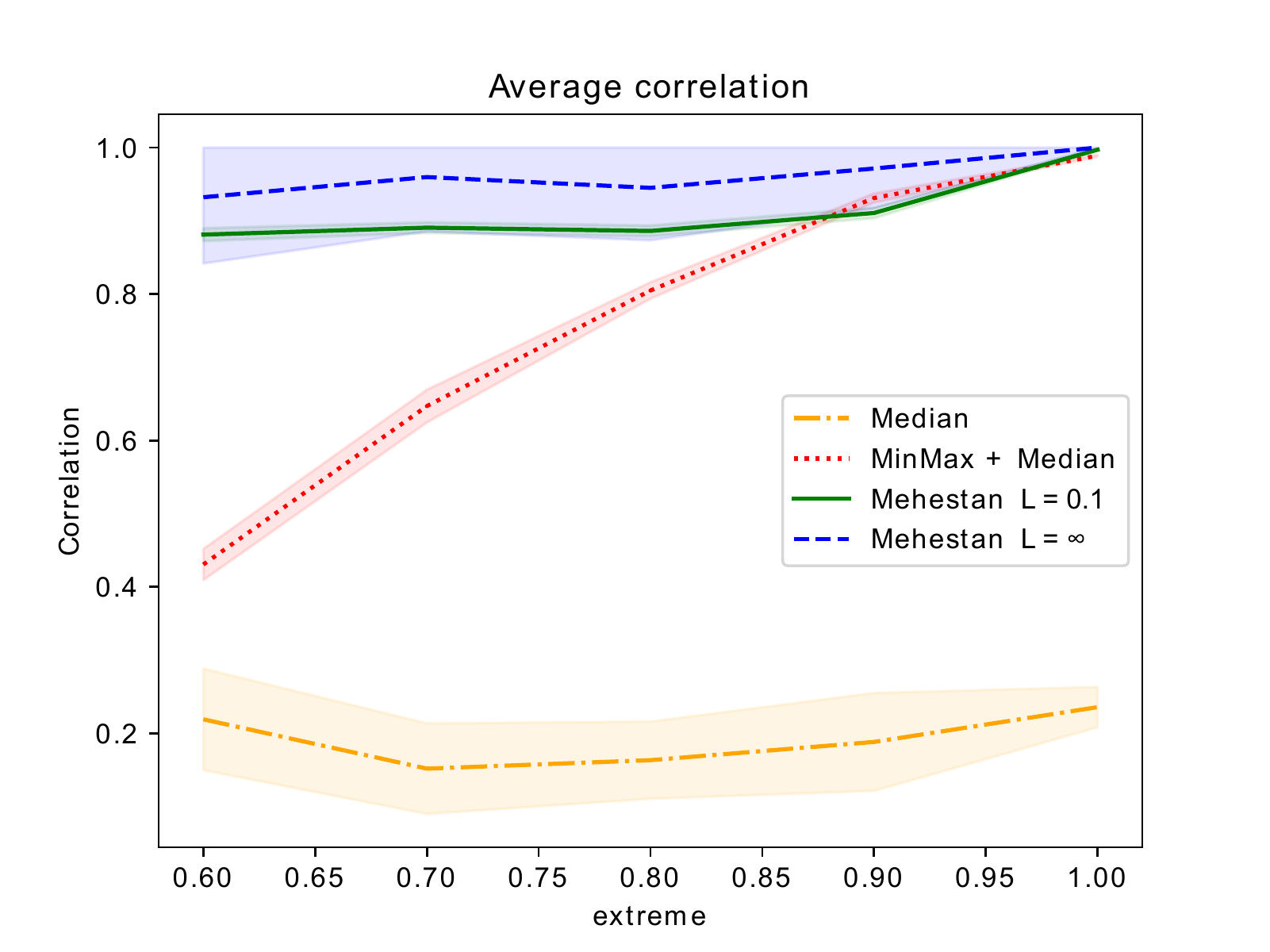}
  \caption{\textbf{Gaussian} distribution of $\theta_*$}
  \label{fig:sm}
\end{subfigure}\hfill
\begin{subfigure}{.48\textwidth}
  \centering
  \includegraphics[width=\linewidth]{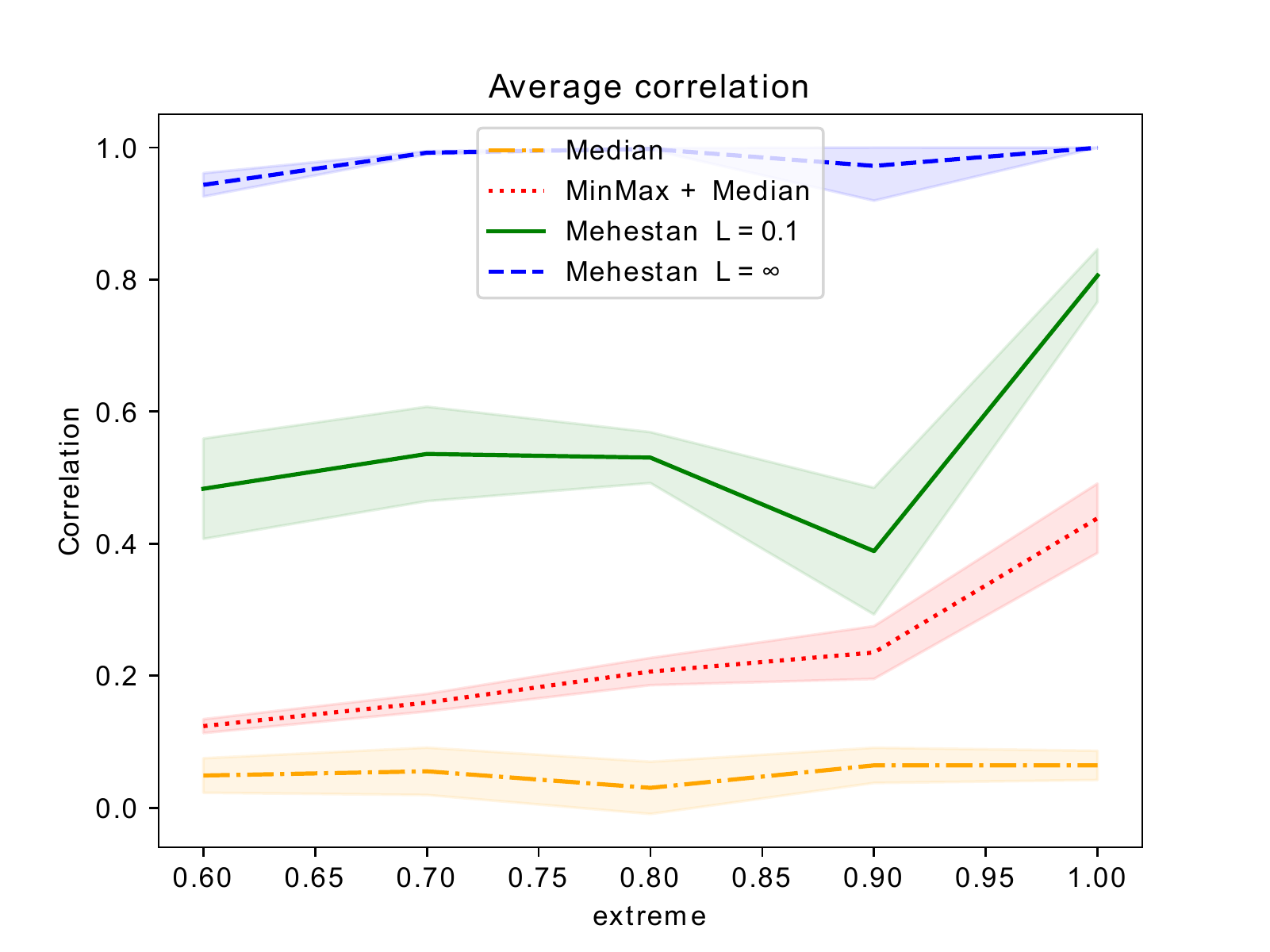}
  \caption{\textbf{Cauchy} distribution of $\theta_*$}
  \label{fig:nextreme}
\end{subfigure}
\caption{Performance of \mehestan{} depending on the \textbf{sparsity bias}.
}
\label{fig:bias}
\end{figure}
The \texttt{extreme} parameter indicates the proportion of alternatives each voter has access to (either the top or bottom \texttt{extreme}). Out of these, the voter rates each of them with probability 0.1.

\paragraph{Influence of the distribution.} From these additional experiments, we observe that the distribution of the ground-truth preferences $\voterscore{*}$ can significantly influence the performance of the voting algorithms. More specifically, it appears that the more the distribution is heavy-tailed, the harder it is for the voting algorithms to recover the preferences. 
This is particularly visible for the Cauchy distribution. Indeed, the baseline \textit{MinMax+Median} performs very poorly on Figure~\ref{fig:cauchy} (\textit{Median} alone is even worse), whereas \mehestan{} is less affected by the change. A possible explanation for this is that the more the distribution is heavy-tailed, the more there are strong outliers in $\theta_*$. This leads to a concentration of the rest of the scores after the min-max normalization, thus showing that global normalization (see Section~\ref{sec:mehestan}) is critical.

\paragraph{Influence of the sparsity bias.} 
Figure~\ref{fig:bias} displays the performance of the voting algorithm as the bias in sparsity decreases. As expected, \mehestan{} is more robust to high levels of bias in the sparsity (left of the plots) than the baseline \textit{MinMax+Median}. For the highest level of bias in sparsity on Figure~\ref{fig:bias}, each voter only rates either the top 60\% or the bottom 60\% of the alternatives. This shows that the global normalization step (see Section~\ref{sec:mehestan}) is crucial, and we can see that the algorithms not performing it (\textit{Median} and \textit{MinMax+Median}) fare poorly in this setting.
\section{Extensions}
\label{sec:remarks}

We discuss here how additional desirable properties can be guaranteed for \emph{robust sparse voting}, by tweaking \mehestan{}.

\subsection{Differential Privacy}
\label{sec:differential_privacy}

In several 
applications, guaranteeing the privacy of \emph{robust sparse voting} may be critical to prevent voter coercion.
In this section, we prove that \mehestan{} can be easily made differentially private.
Let us first recall the definition of (voter-level) differential privacy.

\begin{definition}[Voter-level differential privacy]
Let $\epsilon > 0$. A (randomized) vote $\vote{}$ is $\epsilon$-differentially private if, 
for any $\byzantine \in [\VOTER]$
and all subsets $\mathcal{S} \subset \setR^\ALTERNATIVE$, we have
\begin{equation}
    \pb{ \vote{} (\votingrightsfamily{}, \voterscorefamily{}) \in \mathcal{S} }
    \leq e^{\epsilon} \pb{ \vote{} (\votingrightsfamily_{-\byzantine}, \voterscorefamily{-\byzantine}) \in \mathcal{S} }.
\end{equation}
\end{definition}

For any parameter $\varepsilon > 0$, we define the $\varepsilon$-differentially private \mehestan{} voting algorithm by simply adding a Laplacian noise to each returned global score, whose scale is proportional to $\lipschitz/\varepsilon$.
More formally, for any alternative $\alternative \in [\ALTERNATIVE]$, we define
\begin{equation}
    \DPMehestan{}_{\lipschitz \varepsilon \alternative} (\votingrightsfamily{}, \voterscorefamily{}) 
    \triangleq \mehestan{}_{\lipschitz \alternative} (\votingrightsfamily{}, \voterscorefamily{}) 
    + \Laplace\left (0,\frac{\lipschitz \norm{\votingrightsfamily{}}{\infty}}{\epsilon} \right ),
\end{equation}
where $\Laplace(0,b)$ is a random variable drawn from the Laplace distribution with mean $0$ and scale $b$.

\begin{theorem}
\label{th:differential_privacy}
For all $\epsilon>0$, $\DPMehestan{}_{\lipschitz \varepsilon}$ is $\epsilon$-differentially private.
\end{theorem}

\begin{proof}
This follows directly from the $\varepsilon$-differential privacy of the Laplace mechanism, combined with the resilience guarantee of \mehestan{}.
\end{proof}

\subsection{Uncertainty-aware Voting}

In practice, reported scores are noisy, with potentially different levels of noise.
Here, we show how \mehestan{} can be enhanced to account for uncertainty in the input.
Our key solution is to leverage a new operator called the \emph{mean-risk distance} $\mrdistance{}$.
Essentially, $\mrdistance{}$ simulates the fact that a voter's uncertainty-aware vote is its expected vote, 
when the voter's score is drawn from our Bayesian prior on their actual score.
More formally, given a prior probability distribution $\distribution$ of finite expectation and a point $z \in \setR$, $\mrdistance{}$ is defined as:
\begin{equation}
    \mrdistance{}\left(z \st \distribution \right) \triangleq \expectVariable{\voterscore{} \sim \distribution}{\absv{z -\voterscore{}}}.
\end{equation}
$\qrmedian{}$ can then be easily made uncertainty-aware, by replacing the absolute values as follows:
\begin{equation}
    \qrmedian{\lipschitz} (\votingrightsfamily{}, \distribution) 
    \triangleq \argmin_{z \in \setR} \frac{1}{2 \lipschitz} z^2 + \sum_{\voter \in [\VOTER]} \mrdistance{} (z | \distribution_\voter).
\end{equation}
Crucially, since $\mrdistance{}$ is a mean of functions whose subderivatives are always of absolute value at most 1, 
the subderivatives of $\mrdistance{}$ are also always of absolute value at most 1.
Intuitively, this means that if $z$ falls into a voter $\voter$'s uncertainty set, 
then the voter $\voter$ will only slightly pull $z$ towards their maximum-a-posterior score.
However, if $z$ falls very far from this uncertainty set, the voter $\voter$ will be pulling with its entire voting rights $\votingrights{\voter}$
In any case, the fact that any voter's pull remains bounded by their voting rights guarantees the $\lipschitz$-Lipschitz resilience of this generalization of $\qrmedian{}$.
Additionally, $\mehestan{}$ can be similarly adapted, though handling collaborative scaling normalization is not straightforward.
Interestingly, $\mrdistance{}$ yields a closed form expression for some parameterized priors, like the Laplacian prior.

\begin{proposition}
For a Laplacian prior $\distribution = \Laplace(\mu,\uncertainty{})$, $\mrdistance{}\left(z \st \distribution \right) = \absv{z-\mu}-\mu+\uncertainty{}e^{-\frac{\absv{z-\mu}}{\uncertainty{}}}$.
\end{proposition}

\begin{proof}
This is a straightforward integral computation.
\end{proof}

\subsection{Measuring Preference Polarization}

We now propose a resilient polarization measure on a given alternative $\alternative$.
To do this, we divide the voters $\voter \in \VOTER_\alternative$ who scored $\alternative$ in two equal subsets of high and low-scoring voters. 
More precisely, define $\votingrights{\alternative} \triangleq \sum_{\voter \in \VOTER_\alternative} \votingrights{\voter}$
and $m_\alternative \triangleq \median \set{ \votingrights{\voter}, \scaling{\voter} \normalizedscore{\voter \alternative} + \translation{\voter} \st \voter \in \VOTER_\alternative }$.
Now consider $\VOTER_\alternative^+ \triangleq \set{\voter \in \VOTER_\alternative \st \scaling{\voter} \normalizedscore{\voter \alternative} + \translation{\voter} > m_\alternative}$
and $\VOTER_\alternative^- \triangleq \set{\voter \in \VOTER_\alternative \st \scaling{\voter} \normalizedscore{\voter \alternative} + \translation{\voter} < m_\alternative}$.
Define finally $\votingrights{\alternative}^+ \triangleq \sum_{\voter \in \VOTER_\alternative^+} \votingrights{\voter}$
and $\votingrights{\alternative}^- \triangleq \sum_{\voter \in \VOTER_\alternative^-} \votingrights{\voter}$.
By definition of the median, we must have $\votingrights{\alternative}^- \geq \frac{1}{2} \votingrights{\alternative}$
and $\votingrights{\alternative}^+ \geq \frac{1}{2} \votingrights{\alternative}$.

Now denote $\globalscore{\alternative} \triangleq \mehestan{}_{\lipschitz \alternative} (\votingrights{}, \voterscorefamily{})$.
For any voter $\voter \in \VOTER_\alternative$ who scored $\alternative$, we define 
$\polarization{\voter \alternative}^+ \triangleq \max (0, \scaling{\voter} \normalizedscore{\voter \alternative} + \translation{\voter} - \globalscore{\alternative})$.
Similarly, denote $\polarization{\voter \alternative}^- \triangleq \max (0, \globalscore{\alternative} - \scaling{\voter} \normalizedscore{\voter \alternative} - \translation{\voter})$.
We define the positive polarization $\psi_\alternative^+$ and the negative polarization $\psi_\alternative^-$, on alternative $\alternative$ by
\begin{align}
    \psi_\alternative^\star 
    &\triangleq 1 + \qrmedian{\lipschitz} \set{
    \set{
      \votingrights{\voter}, \eta_{\voter \alternative}^\star - 1
      \st \voter \in \VOTER_\alternative^\star
    } \cup \set{\frac{1}{2} \votingrights{\alternative} - \votingrights{\alternative}^\star, \max \set{ m_\alternative - \globalscore{\alternative}, -1} }
    }, 
\end{align}
for $\star \in \set{-,+}$.
In other words, the polarization is initially assumed to be 1.
Voters who believe that $\globalscore{\alternative}$ is underestimated will then pull $\psi_\alternative^+$ towards larger values; but they can only do so with a unit force.
This prevents malicious voters from hacking polarization measures.

\end{document}